\documentclass{article}
\usepackage[utf8]{inputenc}
\usepackage{arxiv}
\usepackage{todonotes}
\usepackage[onehalfspacing]{setspace}
\usepackage{etoolbox}
\usepackage{fancybox}
\usepackage{graphicx}
\usepackage[tableposition=above]{caption}

\makeatletter
\patchcmd{\@maketitle}{\LARGE}{\Large}{}{}
\makeatother

\usepackage[british]{babel}
\usepackage{amsmath}
\usepackage{amssymb}
\usepackage{amsthm}
\usepackage{dsfont}
\usepackage[round]{natbib}
\usepackage{hyperref}
\usepackage{float}
\usepackage{mathtools}
\usepackage[ruled, vlined]{algorithm2e}
\usepackage{subcaption}
 \usepackage{enumitem}
 \usepackage{stix} % for \sslash used as alternative modulo sign

\newtheorem{theorem}{Theorem}[section]
\newtheorem{defi}[theorem]{Definition}
\newtheorem{lemma}[theorem]{Lemma}
\newtheorem{cor}[theorem]{Corollary}

\newcommand{\R}{\mathbb{R}}
\newcommand{\N}{\mathbb{N}}

\newcommand{\1}{\mathds{1}} 
\DeclareMathOperator*{\argmin}{argmin}

\DeclareMathOperator*{\arginf}{arginf}
\newcommand{\doubleslash}{{\hspace{-1pt}\sslash\hspace{-1pt}}} % alternative slash for alternative, pointwise modulo
\newcommand{\dist}{\operatorname{d}} % avoid notation clash with dimension d

\title{Regression in quotient metric spaces with a focus on elastic curves}
\date{\today}
\author{Lisa Steyer, Almond St\"ocker \& Sonja Greven \\
for the Alzheimer’s Disease Neuroimaging Initiative\thanks{Data used in preparation of this article were obtained from the Alzheimer’s Disease
Neuroimaging Initiative (ADNI) database (\url{adni.loni.usc.edu}). As such, the investigators
within the ADNI contributed to the design and implementation of ADNI and/or provided data
but did not participate in analysis or writing of this report. A complete listing of ADNI
investigators can be found at:
\url{http://adni.loni.usc.edu/wp-content/uploads/how_to_apply/ADNI_Acknowledgement_List.pdf}}}

\begin{document}
\maketitle
\begin{abstract}
We propose regression models for curve-valued responses in two or more dimensions, where only the image but not the parametrization of the curves is of interest. Examples of such data are handwritten letters, movement paths or outlines of objects. In the square-root-velocity framework, a parametrization invariant distance for curves is obtained as the quotient space metric with respect to the action of re-parametrization, which is by isometries. With this special case in mind, we discuss the generalization of 'linear' regression to quotient metric spaces more generally, before illustrating the usefulness of our approach for curves modulo re-parametrization. We address the issue of sparsely or irregularly sampled curves by using splines for modeling smooth conditional mean curves.  We test this model in simulations and apply it to human hippocampal outlines, obtained from Magnetic Resonance Imaging scans. Here we model how the shape of the irregularly sampled hippocampus is related to age, Alzheimer's disease and sex.
\end{abstract}

\newpage

\section{Introduction: Regression for metric spaces}

Regression is a widely used statistical technique for exploring the relationship between covariates and response variables. In the simplest case of linear regression, these variables are elements in the Euclidean space and the relationship between the variables is assumed to be affine linear. Since linear operations are also defined in general Hilbert spaces, the linear regression model can be extended to these spaces \citep{ramsay_1991} and in particular to functional data in $\mathbb{L}_2$ \citep{ramsay_silverman}. For more general (metric) response spaces, analogues of linear models are less straightforwardly to define.

The focus of this paper is to develop an 'elastic' regression model for curves modulo parametrization. More precisely, we consider the quotient space $\mathcal{A}/\Gamma$ as response space, where $\mathcal{A}$ is the set of absolutely continuous curves $\boldsymbol{y}: [0,1] \to \R^d$, $d\in \N$, and $\Gamma$ is the set of boundary-preserving diffeomorphisms $\gamma: [0,1] \to [0,1]$. These curves occur naturally when we look at the outlines of (e.g., anatomical) objects such as the corpus callosum \citep{joshi} where only the image but not the parametrization of the curves is of interest. Furthermore, handwritten letters or symbols \citep[e.g.][]{dryden}, protein structures \citep{srivastava} or centerlines of the internal carotid artery \citep{sangalli} can be viewed as curves modulo parametrization in 2d or 3d. In this work, we investigate the variability in outlines of (a representative slice of) the hippocampus of patients suffering from Alzheimer's disease and of a control group, with the aim of differentiating changes due to Alzheimer's from normal aging
(Fig. \ref{fig:hippocampus_data}). These outlines were extracted from the Alzheimer’s Disease Neuroimaging Initiative (ADNI) database \citep{petersen_adni}.

In this elastic setting, the response space is a quotient space that only has a metric space structure with no notion of linearity, such that linear models cannot be directly defined. It seems natural to use constant speed geodesics instead of affine linear functions in metric spaces, since they coincide in the case of vector spaces. \cite{fletcher} considers such a geodesic regression model for a scalar covariate and the response variable being an element of a smooth Riemannian manifold. Here, the tangent bundle of the manifold serves as a convenient parametrization of the set of possible geodesics. Conversely, in general metric spaces, there is no such parameterization of the geodesics. Hence, although a geodesic regression model could be defined here as well, the estimation of the minimizing geodesics is difficult to accomplish and the result difficult to interpret. For this reason, to our knowledge, the geodesic model has not been considered for general metric spaces.

\cite{petersen} develop a non-geodesic global regression model for responses that are elements of general metric spaces and Euclidean covariates. The regression function here is implicitly defined for each possible combination of covariate values as a (potentially negatively) weighted Fréchet mean. This means that no global model parameters are estimated, which makes interpretation difficult.
Overall, defining a regression model in general metric spaces that is both interpretable and computable appears difficult if not infeasible. To build such a model for response data in a metric space, it thus seems necessary to make use of the specific structure of the response space.

Besides our current structure of interest $\mathcal{A}/\Gamma$, there are various situations where observations can be naturally seen as elements of a quotient space, for instance if the objects of interest are either subject to certain invariances or not fully observed. 
Classic examples arise in statistical shape data analysis \citep{dryden2016statistical}, where objects are considered invariant under translation, rotation, and scaling, as well as occasionally also under reflection, or a subset of these invariances. Other examples include analysis of unlabeled networks \citep{calissano2023}, data on a Grassmannian \citep{hong2016parametric}, data of 3D rotations \citep{fletcher}, compositional data \citep{pawlowsky2015modeling} and density data \citep{van2014bayes}. \cite{srivastava_book} also combines parametrization invariance with statistical shape analysis -- to analyze shapes of curves and also surfaces -- and, more recently, also with analysis of unlabeled graphs to model brain arterial networks \citep{guo2020statistical}.

Given the relevance and variety of data in quotient spaces in the literature, we will motivate our elastic regression model for curves in $\mathcal{A}/\Gamma$ with a more general discussion of a regression approach for responses in certain quotient metric spaces. More precisely, we will consider quotient spaces where the distance is induced by an isometric group action, since this is the case for $\mathcal{A}/\Gamma$ if we equip $\mathcal{A}$ with a semi-metric based on the Fisher-Rao metric\citep{srivastava}. This semi-metric can
be simplified to the $\mathbb{L}_2$ distance using the square-root-velocity (SRV) transformation (essentially $\boldsymbol{y} \mapsto \frac{\dot{\boldsymbol{y}}(t)}{\sqrt{\|\dot{\boldsymbol{y}}(t)\|}}$, $\boldsymbol{y} \in \mathcal{A}$) and minimization over all possible re-parametrizations in $\Gamma$ yields a suitable ``elastic'' distance on $\mathcal{A}/\Gamma$ modulo translation. While not all of the above examples correspond to such isometric group actions, they comprise -- besides re-parameterization groups -- also rotation, reflection and permutation groups.
 
The considered approach, which we refer to as ``quotient regression'', is straight-forward and natural in two ways: a) the structure of the model predictor is simply obtained by projecting a suitable predictor in the original space to the quotient, and b) the model is fit based on the distance in the quotient metric space obtained by minimizing the distance in the original space over all possible group actions. Due to the more general perspective, beyond the target ``quotient linear regression'' for elastic curves, our results on consistency and existence of estimators, as well as inclusion of geodesics in the model space, are also applicable to other quotient regression scenarios. 
It also allows us to point out close connections to approaches for other response quotient metric spaces, such as the recent approach of \citet{calissano} for unlabeled network responses, corresponding to quotient linear regression over the permutation group, and intrinsic Riemannian regression for responses in shape spaces \citep{CorneaEtAl2017RegRiemannianSymSpaces}, combining rotation invariance with invariance with respect to non-isometric re-scaling. Curves in $\mathcal{A}/\Gamma$, in particular, have not been directly considered before as responses in an elastic regression model. One existing approach \citep{tucker} examines the case of elastic curves as covariates instead. They introduce elastic functional principal component regression (fPCR) for scalar response variables and 1d-functions as covariates. Here they first align the data curves to their Fréchet mean and then perform principal component analysis (PCA) for both the aligned curves as well as the optimal re-parametrizations and use both parts in a functional regression model. \citep{guo} proceed similarly but use the principal component scores of the pre-aligned SRV curves as covariates and response in a regression model.

Given this related work, we consider regression (on SRV or on curve level) after pre-aligment natural benchmarks to our model. Specifically, we compare our quotient linear model for curves to 1) linear regression after pre-alignment, a simpler approach that can be used for  regression in the quotient of any Hilbert space, 2) to linear regression on curve instead of SRV level basing only alignment on the SRV framework, 3) to the combination of the simplifications in 1) and 2), and 4) to Fréchet regression \citep{petersen} for general metric spaces, which we adapt and implement for this purpose for the case of $\mathcal{A}/\Gamma$. In simulation studies, we illustrate when a clear performance gain by our model can be expected and when alternatives yield comparably good results.

In applications, such as in our example of hippocampus outlines, it is often necessary to handle sparsely or irregularly observed curves. We achieve this via employing spline bases, as often done in (sparse) functional data analysis. This is motivated by the work of \cite{steyer} on spline-based unconditional elastic mean estimation, where we show identifiability of spline coefficients modulo parameterization and the adequateness of the approach for sparsely or irregularly observed curves. We provide a ready-to-use implementation of our elastic regression model in the R package \texttt{elasdics}. In a simulation study, we validate bootstrap confidence regions either based on spline coefficients - specific to our spline-based modeling - or more generically on distances, and discuss when each is recommended in practice. Both approaches enable data based model selection and assessment of estimation uncertainty. The proposed inference methods allow us to reveal and assess systematic patterns in the hippocampus outlines, which are visually hard to distinguish  due to considerable subject-to-subject variation (Fig. \ref{fig:hippocampus_data}). Specifically, we are able to compare the effect of Alzheimer's disease to that of normal aging -- two mechanisms that have been related to each other in the literature before \citep{henneman} -- in a more detailed and visually intuitive way (Fig. \ref{fig:hippocampus_effects}).

We proceed as follows. In Section \ref{sec:quotient_model}, we first construct the model for responses in general quotient metric spaces before developing the elastic model and our estimation strategy in the particular case of curves modulo parameterization. Here we build on the spline modeling and alignment methods for sparsely and irregularly sampled curves developed in \cite{steyer}. In Section \ref{sec:alternatives}, we present different alternatives to our model, which have not yet been discussed in this form in the literature, either, but deemed natural competitors by us. Section \ref{sec:inference} proposes inference methods for our model. Section \ref{sec:simulation}  compares the performance of our model with the alternative methods described in Section \ref{sec:alternatives} and validates inference based on the spline coefficients. Finally, in Section \ref{sec:application}, we use our method to model the outline of the human hippocampus as a function of age, Alzheimer's disease status and sex, before concluding in Section \ref{sec:discussion}.
\section{Quotient space regression and the particular case of elastic curves}
\label{sec:quotient_model}
Regression models for elastic curves are a particular case of regression models for quotient metric spaces, where the quotient is induced by an isometry, and we will define a regression model for the quotient by using the structure of the original space. In the case of elastic curves, the reparametrization group acts by isometries on $\mathbb{L}_2$, the space of SRV-transformed curves \citep[cf.][]{srivastava_book}. This means the original space here is $\mathbb{L}_2$, which is a Hilbert space and therefore has a linear structure and allows us to base our models on linear regression in $\mathbb{L}_2$. 
With this goal in mind, it is worthwhile to begin with a more general discussion of regression models in metric spaces.  In particular, we discuss reasonable model spaces $\mathcal{F}$ for regression in quotient spaces $\mathcal{Y}/G$ over a more general original space $\mathcal{Y}$ on which the group $G$ acts by isometries. This is of independent interest and shows direct connections to regression for unlabeled networks \citep{calissano} and on shape/form spaces \citep{CorneaEtAl2017RegRiemannianSymSpaces, stoecker}.
%The same is true in the case of unlabelled networks considered by \cite{calissano}. However, we will not restrict ourselves to quotients of Hilbert spaces, since more general quotients are of practical interest as well. For instance the functional shape space studied by \cite{srivastava} arises as a quotient of the sphere, which is a submanifold of $\mathbb{L}_2$, on which both rotation and re-parametrization act by isometries. Therefore, in the following we discuss properties of the original space that are sufficient to allow generalization of 'linear' regression to quotient metric spaces induced by isometries.

As stated by \citet{petersen} in very general terms, traditional regression for the mean is naturally generalized to metric spaces by modeling the conditional Fréchet mean given covariates. This generalizes the least squares problem via replacing the Euclidean metric in the risk minimization with the distance in the metric space. More precisely, for $\mathcal{X}$ being the space of covariates, $(\mathcal{Y}, \dist)$ a metric space, and $(X,Y)$ random variables taking values in $\mathcal{X} \times \mathcal{Y}$, the conditional Fréchet mean of $Y$  given $X$ is given by
\begin{align} \label{eq:metric_reg}
    \mathcal{E}(Y|X = x) = \argmin_{\mu \in \mathcal{Y}} \mathbb{E}(\dist(Y, \mu)^2| X = {x}).
\end{align}
\citet{petersen} point out that without assuming an algebraic structure on $\mathcal{Y}$, it is not feasible to directly define a parametric regression model, that is to define a suitable function space $\mathcal{F}$ such that $x \mapsto \mathcal{E}(Y|X = x) $ is an element of $\mathcal{F}$. For this reason, they develop a generalization of multiple linear regression as a set of weighted Fréchet means, where the weights are given by a known function of the covariates. This allows them to define a regression model in general metric spaces without an explicit model equation or global model parameters. In contrast, as soon as there is any additional structure given on $\mathcal{Y}$, it can potentially be used to motivate a suitable % is a quotient metric space and investigate what condition must be imposed on the original space in order to define an interpretable 
function space $\mathcal{F}$, which we refer to as model space in the following. 
\begin{defi}[Model-based conditional Fréchet mean] \label{def:model_based_Frechet_mean}
Given a model space $\mathcal{F}$, we define the \emph{model-based conditional Fréchet mean} as
\begin{align} 
    f^*  = \argmin_{f \in \mathcal{F}} \mathbb{E}( \mathbb{E}(\dist(Y, f(X))^2| X )) = \argmin_{f \in \mathcal{F}} \mathbb{E}(\dist(Y, f(X))^2),
\end{align}
assuming the total variation $\mathbb{E}(\dist(Y, f(X))^2) < \infty$ is finite for some $f\in\mathcal{F}$.
\end{defi}

Note that, in contrast to Equation \eqref{eq:metric_reg}, the minimization is here over $\mathcal{F}$ rather than point-wise over $\mathcal{Y}$ and that, in general, there does not exist a unique minimizer but $f^* \subseteq \mathcal{F}$ is a set of models. $f^*(x) = \mathcal{E}(Y \mid X=x)$ coincide, if the model is correctly specified. In practice, this is of course hard to verify and it might be more truthful to model $f^*$ and assume that it reasonably approximates $\mathcal{E}(Y \mid X=x)$.
Since the distinction is subtle, we nonetheless simply refer to $f^*(x) = \mathcal{E}(Y\mid X = x)$ as conditional Fréchet mean in the following, when $\mathcal{F}$ is clear from the context, while considering $f^*$ as given in Definition \ref{def:model_based_Frechet_mean}.  

For a corresponding estimator $\hat{f} \subseteq \mathcal{F}$ of $f^*$ the following properties are desirable: a) good interpretability, presenting one central advantage of using the structure of $\mathcal{Y}$,
b) consistency and c) computational feasibility, which is
practically necessary.
While interpretability and computation depend on the structure and will be discussed for quotient space regression in Section \ref{subsec:quotient_reg}, we may discuss consistency already here at a higher level of generality. 

For a given model space $\mathcal{F}$ consider the conditional sample Fréchet mean
\begin{align} \label{eq:empirical_metric_reg}
 \hat{f} = \hat{f}_n =  \argmin_{f \in \mathcal{F}} \sum_{i = 1}^n \dist(y_i, f(x_i))^2
\end{align}
for a given set of observations $\{(x_i, y_i) , i = 1, \dots, n \}$ drawn independently from $(X, Y)$. 
We first show that the estimator $\hat{f}_n \subset \mathcal{F}$, again in general not a unique function, is a consistent estimator of $f^\ast$ in very general metric spaces $\mathcal{Y}$ in the weaker sense established by \cite{ziezold} for the (unconditional) Frechét mean. He showed that for independently and identically distributed random variables the set of empirical Fréchet means converges to the set of expected elements. Since also neither the (conditional) Fréchet mean \eqref{eq:metric_reg} nor its empirical analogue, the (conditional) sample Fréchet mean \eqref{eq:empirical_metric_reg} need to be unique, we can only expect a set version of consistency to hold here as well.

\begin{lemma}[Consistency] \label{lem:consistency}
Let $\mathcal{X}$ be compact and $\mathcal{Y}$ separable. Let $\mathcal{F} \subseteq C(\mathcal{X}, \mathcal{Y})$ be a subset of the continuous functions from  $\mathcal{X}$ to $\mathcal{Y}$ equipped with the metric $\dist_\mathcal{F}(f_1, f_2) = \sup_{x \in \mathcal{X}} \dist(f_1(x), f_2(x))$, $\forall f_1, f_2 \in \mathcal{F}$ and let $\mathbb{E}(\dist(Y, f(X))^2) < \infty$ $\forall f \in \mathcal{F}$. Then $\hat{f}_n$ is a strongly consistent estimator of $f^\ast \subseteq \mathcal{F}$ in the sense of \cite{ziezold}, that is $\bigcap_{n=1}^\infty~\overline{\bigcup_{k = n}^\infty \hat{f}_k}~\subseteq~f^\ast$.
\end{lemma}

This statement is a consequence of a theorem on strong consistency of generalized Fréchet means \citep{huckemann}, here given as Theorem \ref{theo:huckemann}. See Subsection \ref{subsec:proof_consistency} for more details. Lemma \ref{lem:consistency} shows that the sample conditional Fréchet mean is a consistent estimator (for $n \to \infty$) for continuous regression models, which means that consistency does not impose serious constraints on the quotient metric spaces for which we will define our regression model.

Note that this statement on consistency also holds true if $f^\ast = \emptyset$, i.e if there is no $f \in \mathcal{F}$ which minimizes the total variation. To ensure $f^\ast \neq \emptyset$ strong additional assumptions on $\mathcal{F}$ need to be imposed such as in the following statement (proof in Appendix \ref{proof:existence}).
\begin{lemma}[Existence] \label{lem:existence}
Let $\mathcal{X}$ be compact, $\mathcal{Y}$ complete and $\mathcal{F} \subseteq C(\mathcal{X}, \mathcal{Y})$ closed and totally bounded. Then $f \mapsto \mathbb{E}(\dist(Y, f(X))^2)$ attains its minimum on $\mathcal{F}$, i.e. $f^\ast \neq \emptyset$.
\end{lemma}

%The limitations result from the fact that an algorithm for solving the optimization problem in \ref{eq:empirical_metric_reg} must be available and that we want to restrict the space 
%possible models $\mathcal{F}$ to avoid overfitting the data given by a finite sample ($n \in \mathbb{N}$) in an interpretable way. 
To motivate now a natural and interpretable model space $\mathcal{F}$, linear regression will serve as a prototype: in the case where $\mathcal{Y}$ is a Hilbert space and $\mathcal{X} \subset  \R^k$, $\mathcal{F}$ can be chosen as the space of affine linear functions $\mathcal{X}\rightarrow\mathcal{Y}$. The minimization problem in \eqref{eq:empirical_metric_reg} then yields an analytical solution, the minimizer $\hat{f}$ is unique and corresponds to the usual linear predictor $\hat{f}(x_1, \dots, x_k) = \hat{\beta}_0 + \sum_{j = 1}^k \hat{\beta}_j x_j$ with coefficients estimated analogously as for $\mathcal{Y} = \mathbb{R}$. That is for a design matrix $\Xi=(x_{ij})_{i=1, \dots, n; j=0, \dots k}$ with $x_{i0}\equiv 1$ and $\mathbf{y} = (y_1, \dots, y_n) \in \otimes_{i = 1}^n \mathcal{Y}$, a minimizing function $\hat{f} \in \mathcal{F}$ is given by the coefficients 
\begin{align} \label{ols}
    \hat{\boldsymbol{\beta}} = (\hat{\beta_0}, \dots \hat{\beta_k})^T = (\Xi^T \Xi)^{-1}\Xi^T \mathbf{y},
\end{align}
where the matrix times vector multiplication in $\otimes_{i = 1}^n \mathcal{Y}$ is defined as in $\R^n$ and $\Xi^T \Xi$ is assumed to be invertible. A proof for this statement can be found in \cite{ramsay_1991}.
Since general metric spaces, however, lack the notion of linearity, linear models cannot be directly defined here. Instead, we will use the quotient space structure to motivate a suitable generalization. 
%Yet, in the case of one covariate, $\mathcal{X} \subset \R$,  natural candidates  to generalize the concept of affine linear functions are geodesic curves, i.e.\ paths with shortest possible length between two points and constant speed parametrization. %, are natural candidates to generalize the concept of affine linear functions from Euclidean spaces. 
%As such, it deems natural to request that geodesics are contained in $\mathcal{F}$. Without the accustomed properties of a linear predictor, however, restricting to geodesics in each variable of a multiple regression model is problematic \citep{Huckemann2010intrinsicMANOVA}.
%Moreover, in our case, the proposed model space will not only allow for additional, natural dependency structures beyond geodesics, but will be crucial to make model estimation feasible.
%However, computation of distance minimizing geodesics is not feasible in general metric spaces. Estimation strategies only exist in special cases, for example for Riemannian manifold valued data \citep{fletcher}, but rely on the tangent space structure of the manifold and can therefore not be generalized to arbitrary metric spaces. Therefore, we restrict this discussion to the case of metric quotient spaces induced by isometries, and examine which structure we need on the original space to propose a suitable regression space $\mathcal{F}$.

\subsection{Regression in quotient metric spaces}
\label{subsec:quotient_reg}
Since we are considering regression in quotient metric spaces arising from an isometric group action, we briefly review the relevant concepts before defining our regression model for quotient spaces. We first summarize how a quotient space resulting from an isometric group action can be turned into a metric space. %Similar constructions exist for arbitrary equivalence relations \citep{burago} (not necessarily resulting from an isometric group action) but are not discussed here.

\begin{defi}[Quotient metric space] \label{defi:quotient_metric}
Let $(\mathcal{Y}, \dist)$ be a metric space and $G$ a group acting on $\mathcal{Y}$ by isometries. The quotient pseudometric $\dist_G$ is defined as
\begin{align*}
\dist_G(y_1, y_2) = \inf_{g \in G} \dist(y_1, g \circ y_2)
\end{align*}
for all $y_1, y_2 \in \mathcal{Y}$.
Since $\dist_G$ defines an equivalence relation on $\mathcal{Y}$ via $y_1 \sim y_2 \Leftrightarrow \dist_G(y_1, y_2) = 0$, there is a natural quotient metric space $(\mathcal{Y}/\dist_G, \dist_G)$ of $\mathcal{Y}$ under the action $G$. Elements of $\mathcal{Y}/\dist_G$ are denoted by $[y]$ for $y \in \mathcal{Y}$, and $\dist_G$ naturally defines a metric on the equivalence classes in $\mathcal{Y}/\dist_G$.
\end{defi}
A proof that $\dist_G$ is indeed a pseudometric on $\mathcal{Y}$ and therefore a metric on $\mathcal{Y}/\dist_G$ can be found in \cite{burago}. We will denote $(\mathcal{Y}/G, \dist_G) = (\mathcal{Y}/\dist_G, \dist_G)$, although the topological quotient $\mathcal{Y}/G$ defined via the equivalence relation $y_1 \sim y_2 \Leftrightarrow \exists g \in G: y_1 = g \circ y_2$ and $\mathcal{Y}/\dist_G$ do in general not coincide. In fact, $\dist_G$ does not in general define a metric on the topological quotient $\mathcal{Y}/G$, since there can be elements with $\dist_G(y_1, y_2) = 0$ for which there is no $g \in G$ such that $y_1 = g \circ y_2$. Nevertheless, the notation $\mathcal{Y}/G$ is common, for example in the SRV-framework \citep{srivastava_book}, where $\mathcal{A}/\Gamma$ is used instead of $\mathcal{A}/d$ to denote the set of equivalence classes with respect to the elastic distance $d$, and we thus use it here for consistency. This notation emphasizes the dependence on the group $G$ instead of the metric $\dist_G$ it induces.

The following lemma shows that separability and completeness carry over from the original space $\mathcal{Y}$ to the quotient. Thus, assumptions on $\mathcal{Y}/G$ (e.g. such as needed in Lemma \ref{lem:consistency} and \ref{lem:existence}) can be reduced to those on $\mathcal{Y}$.
%allows to reduce the assumptions on $\mathcal{Y}/G$ needed in Lemma \ref{lem:consistency} to show consistency to assumptions on the original space $\mathcal{Y}$. 
\begin{lemma} \label{lem:separable_complete}
\begin{enumerate}
    \item[i)]  $\mathcal{Y}$ separable $\Rightarrow$ $(\mathcal{Y}/G, \dist_G)$ separable.
    \item[ii)] $\mathcal{Y}$ complete $\Rightarrow$ $(\mathcal{Y}/G, \dist_G)$ complete.
\end{enumerate}
\end{lemma}

A proof for these statements can be found in the appendix. For the special case of elastic curves, \ref{lem:separable_complete} ii) was also shown in \citet[Lemma 13]{bruveris2016optimal}. 

\subsubsection{Quotient regression models} \label{subsubsec:quotienregression}

Given the construction of such a quotient metric space $(\mathcal{Y}/G, \dist_G)$, there is a natural way to induce a model space $\mathcal{F}$ for regression on $\mathcal{Y}/G$ from a given model space $\Phi$ of functions $\varphi: \mathcal{X} \rightarrow \mathcal{Y}$. Given $\Phi$, e.g.\ affine linear functions for the case of $\mathcal{Y}$ a Hilbert space, we let $\mathcal{F}$ be the space of (point-wise) projections $x \mapsto [\varphi(x)]$ of functions $\varphi \in \Phi$ on $\mathcal{Y}/G$ which we denote by $\Phi\doubleslash G = \{f: \mathcal{X} \rightarrow \mathcal{Y}/G, x \mapsto [\varphi(x)] \mid \varphi \in \Phi\}$. 
$\Phi\doubleslash G$ is the quotient space of $\Phi$ with respect to the equivalence relation $\varphi_1 \sim \varphi_2 \Leftrightarrow \forall x\in\mathcal{X} : \varphi_1(x) \sim \varphi_2(x)$. %$\Leftrightarrow \dist_{\Phi\doubleslash G}([\varphi_1], [\varphi_2]) = 0$ for $\varphi_1,\varphi_2\in\Phi$, where $\dist_{\Phi\doubleslash G}([\phi_1], [\phi_2]) = \sup_{x\in\mathcal{X}} \dist([\phi_1(x)], [\phi_2(x)])$ in accordance with the notation in Lemma \ref{lem:consistency}.
%If the $\varphi\in\Phi$ are bounded, $\Phi$ can be equipped with the supremum metric $\dist_\Phi(\varphi_1, \varphi_2) = \sup_{x\in\mathcal{X}} \dist(\varphi_1(x), \varphi_2(x))$ and $\Phi\doubleslash G = \Phi / \mathcal{G}$ presents itself a quotient metric space of $(\Phi, \dist_\Phi)$ under the point-wise action of the group $\mathcal{G}$ of functions $\mathcal{X} \rightarrow G$, $x \mapsto g_x$ given by $x \mapsto g_x \circ \varphi(x)$ for any $\varphi\in\Phi$. 
We refer to regression with model space $\mathcal{F} = \Phi\doubleslash G$ for a conditional Fréchet mean in $\mathcal{Y}/G$ as \emph{quotient regression} (over $\Phi$). Note that we now focus on regression on $\mathcal{Y}/G$ instead of the original space $\mathcal{Y}$, i.e we replace $\mathcal{Y}$ by $\mathcal{Y}/G$ in Definition \ref{def:model_based_Frechet_mean}, while keeping the model space denoted as $\mathcal{F}$ for simplicity.

\begin{defi}[Quotient regression] \label{defi:quotient_reg}
Let $x_1, \dots, x_n \in \mathcal{X}$ be realizations of a random variable $X$ and let $[y_1], \dots, [y_n] \in \mathcal{Y}/G$ be realizations of a random variable $[Y]$ taking values in $\mathcal{Y}/G$, where $(\mathcal{Y}, \dist)$ is a metric space and $G$ a group acting on $\mathcal{Y}$ by isometries. Then, for a model space $\Phi = \{\varphi:\mathcal{X} \rightarrow \mathcal{Y}\}$ we define the quotient regression model (over $\Phi$) on $\mathcal{Y}/G$ as the conditional Fréchet mean $\mathcal{E}([Y]| X = x) = f^*(x)$ assuming
\begin{align*}
  f^* = \argmin_{f \in \Phi\doubleslash G} \mathbb{E}(\dist_G([Y], f(X)))^2),
\end{align*}
which is estimated as $\hat{f}(x) = \hat{f}_n(x) = [\hat{\varphi}(x)]$ with
\begin{align} \label{eq:quotient_reg}
\hat{\varphi} = \argmin_{\varphi \in \Phi} \sum_{i = 1}^n  \dist_G([y_i], [\varphi(x_i)])^2 = \argmin_{\varphi \in \Phi} \sum_{i = 1}^n \inf_{g_i \in G} \dist(g_i \circ y_i, \varphi(x_i))^2.
\end{align}
\end{defi}

While it is not immediately clear that quotient regression is a good model for every combination of $\mathcal{Y}$, $G$ and $\Phi$, we will, in this section, give some evidence that it  is in several cases. In particular, we will later illustrate its benefits in our example of elastic curve modeling based on a multiple linear spline predictor.
Another example of quotient regression with model space $\Phi = \{\varphi: \mathbb{R}^k \to \mathcal{Y}, \varphi \text{ affine linear}\}$ has been suggested by \cite{calissano} for the special case of $\mathcal{Y}$ being the set of networks and $G$ being the permutation group on the set of nodes. In particular, our result on consistency for the quotient regression model (Corollary \ref{cor:consistency}) also applies to their case. %Our construction shows that their model is an example of a general class of models, which can be defined for the quotient of an arbitrary Hilbert space by  a group which acts on $\mathcal{Y}$ by isometries.
Note that, by contrast, approaches inducing a probability distribution on $\mathcal{Y}/G$ via some distribution on $\mathcal{Y}$ \citep[such as in, e.g.\ offset normal shape distributions, ][Chap. 11]{dryden2016statistical} are, in general, fundamentally different from our distribution-free approach that constructs the model space via projection  while the mean is defined to minimize the distance $\dist_G$. 

Consistency of the quotient regression model carries over from Lemma \ref{lem:consistency} using that %the results from Lemma \ref{lem:separable_complete} i) for
separability of $\mathcal{Y}/G$ (based on Lemma \ref{lem:separable_complete} i)) and continuity of $\Phi \doubleslash G \subset C(\mathcal{X}, \mathcal{Y}/G)$ carry over from $\mathcal{Y}$ and $\Phi$, respectively.

\begin{cor}[Consistency for quotient regression] \label{cor:consistency}
Let $\mathcal{X}$ be compact and $\mathcal{Y}$ separable. Let $\Phi \subseteq C(\mathcal{X}, \mathcal{Y})$ be a subset of the continuous functions from  $\mathcal{X}$ to $\mathcal{Y}$, $\Phi\doubleslash G$ equipped with the metric $\dist_{\Phi\doubleslash G}$, and $\mathbb{E}(\dist(Y, [\varphi(X)])^2) < \infty$ for all $\varphi \in \Phi$. Then $\hat{f}_n$ is a strongly consistent estimator of $f^\ast \subseteq \mathcal{F}$ in the sense of Lemma \ref{lem:consistency}.
\end{cor}

We can also formulate requirements as in Lemma \ref{lem:existence} to ensure that the quotient regression model is not empty. Note that all requirements are given for the original space $\mathcal{Y}$ and the model space $\Phi$ instead of $\mathcal{Y}/G$ and $\mathcal{F}$.

\begin{cor}[Existence for quotient regression] \label{lem:existence_quotient}
Let $\mathcal{X}$ be compact, $\mathcal{Y}$ complete and $\Phi \subseteq C(\mathcal{X}, \mathcal{Y})$ closed and totally bounded. Then
$\varphi \mapsto \mathbb{E}(\dist_G([Y], [\varphi(X)])^2)$ attains its minimum on $\Phi$.
\end{cor}

In the remainder of this subsection, we discuss computational aspects of quotient regression estimators. 
Here, quotient regression offers a straight-forward estimation scheme if, for realizations $\tilde{y}_1, \dots, \tilde{y}_n$ of a random variable $\tilde{Y}$ in the original space $\mathcal{Y}$, an estimator $\tilde{\varphi}$ of $\varphi^* = \argmin_{\varphi\in\Phi} \mathbb{E}[\dist(Y, \varphi(X))^2]$ is available:
in this case, we address the minimization problem in \eqref{eq:quotient_reg} via alternating 1) updating $\hat{f}(x) = [\hat{\varphi}(x)]$ by setting $\hat{\varphi}$ to the $\tilde{\varphi}$ fitting the data $(\tilde{y}_i, x_i), i = 1, \dots, n$ for current response realizations $\tilde{y}_i\in[y_i]\subset\mathcal{Y}$, and 2) optimally aligning the data, i.e. finding $\tilde{y}_i = \argmin_{y \in [y_i]} \dist(y, \tilde{\varphi}(x_i))$ for each $y_i$ and a current estimator $\tilde{\varphi} \in \Phi$.\newline
Alternating algorithms are natural in settings such as ours and corresponding estimation schemes have successfully been used for estimation of Fréchet means in different quotient space scenarios, including, for instance, conditional mean estimation for unlabeled networks \citep{calissano} or unconditional estimation of Procrustes means in shape analysis \citep{dryden2016statistical}, of elastic mean curves \citep{steyer}, elastic mean shape \citep{srivastava_book}, and elastic full Procrustes mean shape estimation \citep{stoecker2022elastic}.

In practice it is necessary to compute numerical approximations $\tilde{\varphi} \in \Phi$ and $\tilde{y}_i = \tilde{g}_i \circ y_i$ for some $\tilde{g}_i \in G$, where true optima need not be unique or even exist in general. The algorithm iteratively reduces the loss in each step and returns a single $\hat{f}\in \Phi\doubleslash G$ even in cases where the set of empirical conditional Fréchet means  does not contain exactly one function. The resulting estimator $\hat{f}$ is expected to give a good fit to the data even if technically there exists a $f\in \mathcal{F}$ with a (slightly) lower empirical loss. Such differences are likely to be small compared to the variability introduced by finite samples, and the practically relevant issue of multiple local minima can be addressed by testing different initial values. 

%Depending on the setting, the alignments in 2) are not trivial, nor can we, in general and analogously to the alternating algorithms cited above, provide any guarantees for convergence of the described algorithm. 
%In fact, the algorithm will in practice involve computation of unique numerical approximations $\tilde{\varphi} \in \Phi$ and $\tilde{y}_i = \tilde{g}_i \circ y_i$ for some $\tilde{g}_i \in G$, whereas true optima do not necessarily have to be unique or even exist in general. Moreover, the case that the empirical conditional Fréchet mean $\hat{f} = \emptyset$ is empty is in theory also covered by the definition and previous results, whereas the algorithm will return a single $\hat{f}\in \Phi\doubleslash G$.
%Practically, however,  the algorithm iteratively decreases the loss function in each step and is expected to return useful results also in cases, where no existence and uniqueness results for these estimators are available. %Resulting from iterative reduction of the numerical square-error $\sum_{i=1}^n \dist(\tilde{g}_i \circ y_i, \tilde{\phi}(x_i))^2$ (potentially also with different starting values), 
%The resulting estimator $\hat{f}$ might still yield a good fit to the data even if, technically, there exists an $f\in \mathcal{F}$ with a (slightly) lower empirical loss. Such differences will likely be small compared to the variability induced by finite samples, and the real challenge of multiple local minima can be addressed by testing different starting values. 

\subsubsection{Quotient geodesic regression and geodesics on the quotient space} \label{subsubsec:quotiengeodesicregression}

For the case of a Riemannian manifold $\mathcal{Y}$, geodesic regression has been discussed by various authors \citep[e.g.,][]{fletcher} as natural generalization of simple linear regression on a single covariate $x\in\mathcal{X} \subset \mathbb{R}$ to curved spaces. In the context of manifolds, geodesics are typically defined as curves $c:(-\varepsilon, \varepsilon) \rightarrow \mathcal{Y}$ around some $\boldsymbol{\mu} = c(0)$ with a constant velocity $\boldsymbol{\beta} = \dot{c}(0)$ in the tangent space $T_{\boldsymbol{\mu}}\mathcal{Y}$ at $\boldsymbol{\mu}$. Locally, they correspond to paths of shortest length. 
For general metric spaces, and in particular for a quotient metric space $\mathcal{Y}/G$ with some general $G$, ``geodesics'' commonly directly refer to shortest paths, due to the lack of a manifold structure. As such they are less tangible, do not offer the same parameterization in terms of ``intercept'' $\boldsymbol{\mu}$ and ``slope'' $\boldsymbol{\beta}$, and the set of geodesics in $\mathcal{Y}/G$ does, in general, not yield a convenient model space $\mathcal{F}$. %In fact, even when geodesics between two points $[y_1], [y_2] \in \mathcal{Y}/G$ can be computed, finding a minimizing geodesic $\hat{f}$ is not necessarily feasible, as in our motivating example of elastic regression.

The next lemma gives a characterization of the shortest paths and therefore geodesics on the quotient metric space $(\mathcal{Y}/G, \dist_G)$ if $(\mathcal{Y}, \dist)$ is a length metric space, i.e.\ if the distance $d$ coincides with the intrinsic metric, which is the infimum of the lengths of all paths from one point to another.

\begin{lemma}[Shortest paths in quotient metric spaces] \label{lem:shortest_path_quotient}
Let $(\mathcal{Y}, \dist)$ be a length metric space and $G$ a group acting on $\mathcal{Y}$ by isometries. Let $y_1, y_2 \in \mathcal{Y}$ and assume there is a $\tilde{g} \in G$ with $\dist_G([y_1], [y_2]) = \dist(y_1, \tilde{g} \circ y_2)$, in which case we call $y_1$ and $\tilde{g} \circ y_2$  aligned. Furthermore assume there is a shortest path $\gamma:[a, b] \to \mathcal{Y}$ with $\gamma(a) = y_1$ and $\gamma(b) = \tilde{g} \circ y_2$, i.e $\gamma$ is a continuous function connecting $y_1$ and $\tilde{g} \circ y_2$ with minimal length. Then $[\gamma]$ is a shortest path in $(\mathcal{Y}/G, \dist_G)$ between $[y_1]$ and $[y_2]$, where $[\gamma]$ is the projection of $\gamma$ onto $\mathcal{Y}/G$, i.e. $[\gamma](t) = [\gamma(t)]$ for all $t \in [a, b]$.
\end{lemma}

A proof of this statement based on the argumentation of \cite{burago} can be found in the appendix. It shows that shortest paths in the quotient metric space, and therefore geodesics, are essentially a subset of those in the original space $(\mathcal{Y}, \dist)$, for which start $y_1$ and end point $y_2$ are aligned, that is $\argmin_{g \in G} \dist(y_1, g \circ y_2) = \text{id}$. %Again note that although this lemma tells us how to compute shortest paths between two given points in the quotient space with respect to the quotient metric, finding the geodesic in $\mathcal{Y}/G$ that minimizes the squared loss in \eqref{eq:empirical_metric_reg} with respect to $\dist_G$ is still not feasible in general settings. %For this reason we suggest to use as a model space $\mathcal{F}$ the larger space of projections  into $\mathcal{Y}/G$ of geodesics with respect to $d$ in the original space $\mathcal{Y}$. This space, according to Lemma \ref{lem:shortest_path_quotient}, contains the geodesics in $\mathcal{Y}/G$ as a subset and is easier to work with if $\mathcal{Y}$ has additional structure, such as a Hilbert space structure in the case of elastic curves (cf.\ Subsection \ref{subsec:quotientelastic}). 

%\begin{defi}[Quotient space regression with one covariate] \label{defi:quotient_reg}
%Let $x_1, \dots, x_n \in \mathcal{X} \subset \R$ be realizations of a random variable $X$ and let $[y_1], \dots, [y_n] \in \mathcal{Y}/G$ be realizations of a random variable $[Y]$ taking values in $\mathcal{Y}/G$, where $(\mathcal{Y}, \dist)$ is a length metric space and $G$ a group acting on $\mathcal{Y}$ by isometries. Then we define a regression model on $\mathcal{Y}/G$ as the conditional Fréchet mean 
%\begin{align*}
%  f(x) = \mathcal{E}([Y]| X = x) = \arginf_{[\mu] \in \mathcal{Y}/G} \mathbb{E}(\dist_G([Y], [\mu])^2| X = x),
%\end{align*}
%which we estimate as $\hat{f}(x) = [\hat{\varphi}(x)]$ with
%\begin{align} \label{eq:quotient_reg}
%\hat{\varphi} = \arginf_{\varphi \in \Phi} \sum_{i = 1}^n  \dist_G([y_i], [\varphi(x_i)])^2 = %\arginf_{\varphi \in \Phi} \sum_{i = 1}^n \inf_{g_i \in G} \dist(g_i \circ y_i, \varphi(x_i))^2,
%\end{align}
%where $\Phi$ consists of constant speed geodesics on $\mathcal{Y}$.
%\end{defi}
Lemma \ref{lem:shortest_path_quotient} tells us how to compute shortest paths between two given points in the quotient space with respect to the quotient metric. Yet, finding the geodesic in $\mathcal{Y}/G$ that minimizes the squared loss in \eqref{eq:empirical_metric_reg} with respect to $\dist_G$ is still not feasible in general settings, and there is even no numerical estimation algorithm available that would promise at least reasonable practical solutions. This is, in particular, the case in our motivating example of elastic regression.

As suitable alternative, we suggest \emph{quotient geodesic regression} for the case where $\mathcal{Y}$ carries a Riemannian manifold structure (or in particular a Hilbert space structure) that allows for geodesic (or linear) modeling $\Phi$, and show that the resulting model space $\Phi\doubleslash G$ in fact contains the geodesics in $\mathcal{Y}/G$. Moreover, Simulation 5 (Fig. \ref{fig:simulation_5}) in Section \ref{sec:simulation} gives one illustrative example of a non-geodesic model $f\in \Phi\doubleslash G$ that is likely desirable to also have  included in the model space, a further argument for a larger model space in practical data scenarios. 

\begin{defi}[Quotient geodesic regression] \label{defi:quotient_geodesic}
Referring to the setting of Definition \ref{defi:quotient_reg}, we call quotient regression on a single covariate $X$ in $\mathcal{X} \subset \R$ for a response $[Y]$ in $\mathcal{Y}/G$ quotient geodesic regression if $\Phi$ is the set of geodesics on $\mathcal{Y}$.
\end{defi}

Given the requirements on $\mathcal{X}$, $\mathcal{Y}$ and $(X,Y)$ in Lemma \ref{lem:consistency}, quotient regression %function $\hat{f}$ is a 
yields a consistent estimator $\hat{f}$ for all true $f^* \in \mathcal{F} = \Phi\doubleslash G$. %\{[\varphi]| \varphi \in \Phi\}$. 
Accordingly, in particular all true $f^*$ that are geodesics in $\mathcal{Y}/G$, which form a subset in this space (see Lemma \ref{lem:shortest_path_quotient}), can be consistently estimated by quotient geodesic regression, using the quotient $\Phi\doubleslash G$ of the geodesics $\Phi$ in $\mathcal{Y}$ as a larger model space than the geodesics in $\mathcal{Y}/G$. 

%Enlarging the model space in this way is advantageous if the loss minimizing geodesic in the quotient space cannot be computed, but the one in the original space can. In this setting, we can address the minimization problem in \eqref{eq:quotient_reg} via alternating between estimating a closest geodesic $\varphi \in \Phi$ to the data $g_i \circ y_i, i = 1, \dots, n$ in for fixed actions $g_i$, and optimally aligning the data, i.e. finding actions $g_i$ as minimizer of $\dist(g_i \circ y_i, \varphi(x_i))$, for a fixed function $\varphi \in \Phi$.

\subsubsection{Quotient linear models}
In quotient geodesic regression, we considered the special case of simple regression with a single covariate $x \in \mathcal{X} \subset \R$. We now consider multiple regression with covariates $\boldsymbol{x} \in \mathcal{X} \subset \R^k$  as a  basis for our main goal of elastic regression, and focus in the following on linear models. To facilitate a suitable linear structure, we consider the 
important special case where %The simplest setting in which 
%we can estimate geodesics that minimize the squared distance to the data is that 
$\mathcal{Y}$ is a Hilbert space, where geodesics are straight lines and the model space $\Phi$ %considered in Definition \ref{defi:quotient_reg} 
can be chosen as (a linear subspace of) the space of affine linear functions $\mathcal{X} \rightarrow \mathcal{Y}$. %, and a minimiser of the sum of squared distances for given actions $g_i, i = 1, \dots, n$ can be computed analytically in $\mathcal{Y}$ as discussed above (see \eqref{ols}). %Therefore, although we consider the more general case to be interesting as well, in particular if $\mathcal{Y}$ is a Riemannian manifold (e.g., a pre shape space), from now on
%In the following, we will focus on the case where $\mathcal{Y}$ is a Hilbert space, since this is the case when we consider elastic curves, and defer for the interesting more general case  of a Riemannian manifold  $\mathcal{Y}$  to Remark \ref{quotient:riemannian}.
In Section \ref{quotient:riemannian}, we will then also briefly discuss extensions to the more general case where $\mathcal{Y}$ is a Riemannian manifold.

%As geodesics are univariate functions mapping to a metric space,  we only defined a model in the univariate case $\mathcal{X} \subset \R$ so far.  If $(\mathcal{Y}, \|\cdot \|_\mathcal{Y})$ is a Hilbert space, our approach allows to also generalize the quotient space regression model \eqref{eq:quotient_reg}  to the case of multiple regression, $\mathcal{X} = \R^k$, by considering for $\Phi$ affine linear functions from $\R^k$ to $\mathcal{Y}$.

\begin{defi}[Quotient linear regression with multiple scalar covariates] \label{defi:mult_quotient_reg}
Let $\boldsymbol{x}_i = (x_{i,1}, \dots, x_{i,k})^\top \in \mathcal{X} \subset \R^k$, $i = 1, \dots, n$ be realizations of a random vector $\boldsymbol{X} = (X_1, \dots, X_k)^\top$ and let $[y_1], \dots, [y_n] \in \mathcal{Y}/G$ be realizations of a random variable $[Y]$ taking values in $\mathcal{Y}/G$, where $(\mathcal{Y}, \|\cdot \|_\mathcal{Y})$ is a Hilbert space and $G$ a group acting on $\mathcal{Y}$ by isometries. Then the quotient linear regression model %over a subspace $\Phi$ of the affine linear functions $\mathcal{X} \rightarrow \mathcal{Y}$ 
is a quotient regression over a model space $\Phi$ of affine linear functions $\mathcal{X} \rightarrow \mathcal{Y}$ given by $(k+1)$ parameters in a %$\Phi \cong \mathcal{B}^{k+1}$ with
subspace $\mathcal{B} \subset \mathcal{Y}$, that is $\mathcal{F}=\Phi \doubleslash G$ with elements
\begin{align*} 
    f:\R^k \to \mathcal{Y}/G; \quad
    \boldsymbol{x} = (x_1, \dots, x_k)^\top \mapsto %\mathcal{E}([Y]| X_1 = x_1, \dots X_k = x_k) =
    \left[\beta_0 + \sum_{j = 1}^k \beta_j x_{j} \right],
\end{align*}
where we assume $\mathcal{E}([Y]| X_1 = x_1, \dots X_k = x_k) = f(\boldsymbol{x}) = \left[\beta_0 + \sum_{j = 1}^k \beta_j x_{j} \right]$ and the coefficients $\beta_0, \dots, \beta_k\in \mathcal{B} \subset \mathcal{Y}$ are estimated as
\begin{align} \label{eq:mult_quotient_reg}
(\hat{\beta}_0, \dots, \hat{\beta}_k) =\argmin_{\beta_0, \dots, \beta_k \in \mathcal{B}} \sum_{i = 1}^n  \dist_G([y_i], [\beta_0 + \sum_{j = 1}^k \beta_j x_{i,j}])^2 = \argmin_{\beta_0, \dots, \beta_k \in \mathcal{B}} \sum_{i = 1}^n \inf_{g_i \in G} \left\| \beta_0 + \sum_{j = 1}^k \beta_j x_{i,j} - g_i \circ y_i \right\|_\mathcal{Y}^2.
\end{align} 
%for covariates $\\boldsymbol{x}_1, \dots, \\boldsymbol{x}_n \in \R^k, \\boldsymbol{x}_i = (x_{i,1}, \dots, x_{i,k})^{T}$, and response variables $[y_1], \dots, [y_n] \in \mathcal{Y}/G$. 
Thus, the estimated regression function becomes $\hat{f}(\boldsymbol{x}) = [\hat{\beta}_0 + \sum_{j = 1}^k \hat{\beta}_j x_{j}]$.
\end{defi}

 For a quotient over a linear space, this generalizes the definition of the univariate quotient geodesic model \ref{defi:quotient_geodesic} since for $k = 1$ and $\mathcal{B} = \mathcal{Y}$ the set of constant speed geodesics coincides with the set of affine linear functions $\Phi$ and therefore $\mathcal{F}=\Phi \doubleslash G = \{f: \R \to \mathcal{Y}/G,\, x_1 \mapsto [\beta_0 + \beta_1 x_1]\}$ is the set of projections of constant speed geodesics.
 
The following corollary shows that the model space $\Phi\doubleslash G$ of quotient linear regression includes geodesics on $\mathcal{Y}/G$ not only in coordinate directions but also in any direction in the covariate space that is a convex linear combination of coordinate directions. We proof this statement in the appendix via showing that the set of elements which are aligned to one point form a convex cone (Lemma \ref{lem:convex_cone}).

\begin{cor} \label{cor:multivariate_geodesic}
Let $(\mathcal{Y}, \| \cdot \|_\mathcal{Y})$ be a Hilbert space and $G$ act on $\mathcal{Y}$ by isometries.
%a real inner product space, such as a Hilbert space, and $G$ act on $\mathcal{Y}$ by isometries with $g \circ 0 = 0$ the identity element of $G$ for all $g \in G$.
Let $f: [0,1]^k \to \mathcal{Y}/G, \ (x_1, \dots, x_k)^T \mapsto [\beta_0 + \sum_{j = 1}^k x_j \beta_j]$ with $\beta_0, \beta_1, \dots, \beta_k \in \mathcal{Y}$ be such that $\beta_0 + \beta_j$ is aligned to $\beta_0$ for all $j = 1, \dots, k$. Then $f|_{x_j}: [0,1] \to \mathcal{Y}/G, x_j \mapsto [\beta_0 + x_j \beta_j]$ is a constant speed geodesic for all $j = 1, \dots, k$ due to Lemma \ref{lem:shortest_path_quotient}. Furthermore, let $\lambda_1, \dots, \lambda_k \in [0, 1]$ with $ \sum_{j = 1}^k \lambda_j = 1$. Then
$$\tilde{f}: [0,1] \to \mathcal{Y}/G, \ x \mapsto \left[ \beta_0 + x \sum_{j = 1}^k \lambda_j \beta_j \right]$$ is a constant speed geodesic in $\mathcal{Y}/G$ between $[\beta_0]$ and $[\beta_0 + \sum_{j = 1}^k \lambda_j \beta_j]$.
\end{cor}

This generalizes geodesics to the multiple covariate setting as well as possible given the lack of a linear space structure for $\mathcal{Y}/G$. 
%This multiple regression model on the quotient space with model space $\Phi = \{\varphi: \mathbb{R}^k \to \mathcal{Y}, \varphi \text{ affine linear}\}$ 
Such a quotient linear model has been suggested by \cite{calissano} for the special case of $\mathcal{Y}$ being the set of networks and $G$ being the permutation group on the set of nodes. Our construction shows that their model is an example of a general class of models, which can be defined for the quotient of an arbitrary Hilbert space by a group which acts on $\mathcal{Y}$ by isometries, and points out the inherent connection to other such cases.\newline
In practice, the coefficients $\beta_j$ will usually be modeled within a suitable finite-dimensional subspace $\mathcal{B} \subset \mathcal{Y}$, such that also $\Phi \cong \mathcal{B}^{k+1}$ will be finite-dimensional. While $\Phi\doubleslash G$ then no longer  necessarily contains the geodesics on $\mathcal{Y}/G$ precisely, it  may still yield good approximations to them. That the model space $\Phi \subseteq C(\mathcal{X}, \mathcal{Y})$ is a finite dimensional subspace  allows us to conclude that the regression model is non-empty under weaker  assumptions than in Lemma \ref{lem:existence_quotient}. 

\begin{theorem}[Existence in finite dimensional model spaces] \label{theo:existence}
Let $\mathcal{Y}$ be a Hilbert space, $\mathcal{X} \subset \R^k$ compact and $\Phi \subseteq C(\mathcal{X}, \mathcal{Y})$ a finite dimensional subspace. If $[Y]$ is bounded and supp$(X) = \mathcal{X}$, there is a minimizer of $\Psi(\varphi) = \mathbb{E}(\dist_G([Y], [\varphi(X)])^2)$ in $\Phi$.
\end{theorem}

A proof of this statement can be found in the appendix. It shows that for any finite dimensional model space $\Psi$ we can expect $f^\ast \neq \emptyset$, i.e. that the quotient regression model $f^\ast$ in Definition \ref{defi:quotient_reg} is not the empty set.

\subsubsection{Side-remark on quotient regression over a Riemannian manifold}
%\begin{remark}[Quotient regression over a Riemannian manifold]
\label{quotient:riemannian}
While for a single covariate geodesic regression is the canonical generalization of simple linear regression to a Riemannian manifold $\mathcal{Y}$, transfer of multiple linear regression to curved spaces is somewhat less straight-forward. 
Yet, a still natural option is given by generalized linear model (glm) type intrinsic regression \citep{zhu2009intrinsic, CorneaEtAl2017RegRiemannianSymSpaces} with a ``Riemannian Log-link'', i.e. with the model space $\Phi$ consisting of functions $\varphi:\boldsymbol{x}\mapsto\operatorname{Exp}_{\beta_0}(\beta_1x_1+\dots+\beta_kx_k)$ with intercept $\beta_0 \in \mathcal{Y}$, coefficients $\beta_1,\dots, \beta_k \in T_{\beta_0}\mathcal{Y}$ in the tangent space at $\beta_0$, and the Riemannian exponential map $\operatorname{Exp}$ at $\beta_0$ as response-function. 
The model models and estimates the conditional Fréchet mean with respect to the intrinsic Riemannian distance $\dist$ and reduces to geodesic regression for $k=1$. Quotient intrinsic regression over a Riemannian manifold can then be defined using $\mathcal{F} = \Phi //G$ with the above glm-type intrinsic $\Phi$.
Intrinsic regression on Kendall's shape space $\Sigma^m$ of 2D landmark configurations $\boldsymbol{y} \in \mathbb{C}^m$ modulo translation, scale and rotation, discussed as an example by \citep{CorneaEtAl2017RegRiemannianSymSpaces}, can, in fact, be considered a special case of quotient intrinsic regression with $\mathcal{Y} = \mathbb{S}^{2(m-1)}$ the sphere of dimension $2(m-1)$, $\Phi$ the model space of intrinsic regression on $\mathbb{S}^{2(m-1)}$, and the 2D rotations $G = \{\exp(\omega \sqrt{-1}) \mid \omega \in [-\pi, \pi)\}$ the isometric group action.
In this case, $\mathcal{Y}/G = \Sigma^m$ carries itself a Riemannian manifold structure (of the complex projective space $\Sigma^m \cong \mathbb{C}P^{m-2}$). 
For shapes in higher dimensions, $\mathcal{Y}/G$ does not carry a manifold structure anymore \citep{huckemann2010intrinsic}, but an analogous quotient intrinsic regression model could also be formulated. 
Additionally, an intrinsic regression model of the 2D form/size-and-shape space of $\boldsymbol{y}$ modulo translation and rotation \citep{stoecker} with $\mathcal{Y} = \mathbb{C}^{m-1}$ yields another example of quotient linear regression.
Hence, intrinsic regression on manifolds does not only yield a further, more general, underlying model space $\Phi$ for quotient regression, but also further motivation for the quotient (linear) model approach, since in special cases intrinsic regression models on manifolds present specially tailored quotient regression models.

%The case where $\mathcal{Y}$ is a Riemannian manifold is also of practical relevance, e.g., when $\mathcal{Y}$ is a pre shape space. Estimating procedures for geodesics that minimize the squared distance also exist in the Riemannian case \citep{fletcher}, allowing direct application of Definition \ref{defi:quotient_reg}. Moreover, we propose to extend Definition \ref{defi:mult_quotient_reg} in this case to define a multiple quotient space regression model as the projections of linear functions in the tangent space mapped back to the manifold by the exponential map. By Lemma \ref{lem:shortest_path_quotient}, these model spaces will include functions that are geodesics in coordinate directions, which are geodesics on the original space with start and end points aligned.
%\end{remark}

%\subsection{Quotient regression for elastic curves in the SRV framework}
\subsection{Elastic regression for curves via quotient linear models in the SRV framework}
\label{subsec:quotientelastic}

In this subsection we will develop quotient regression for the particular case of curves %in $\mathcal{A}$ 
modulo re-parametrization (and translation) in order to obtain an elastic regression model for curves. %However, the re-parameterization group $\Gamma$ does not act by isometries on $\mathcal{A}$, so we will not consider the quotient space regression model for the curves directly, but for their SRV transformation.
To achieve that the re-parameterization group $\Gamma$ acts by isometries, we will not consider the quotient space regression model for the curves $\boldsymbol{y}$ directly, but for their SRV transformation.
Considering SRV transforms in the Hilbert space $\mathbb{L}_2$ of square integrable functions $q: [0,1] \rightarrow \mathbb{R}^d$ induces a suitable metric on the space of absolutely continuous curves $\mathcal{A}$ modulo translation.

\begin{lemma}[SRV transformation \citep{srivastava_book}]\label{lem:SRVtrafo}
The SRV transformation $Q$ defined via
\begin{align*}
Q(\boldsymbol{y})(t) = 
\begin{cases}
\frac{\dot{\boldsymbol{y}}(t)}{\sqrt{\|\dot{\boldsymbol{y}}(t)\|}} 
\quad 
&\text{if } \dot{\boldsymbol{y}}(t) \neq 0 \\
0 & \text{if } \dot{\boldsymbol{y}}(t) = 0,
\end{cases}
\end{align*}
gives a one-to-one correspondence between the absolutely continuous curves $\mathcal{A}$ modulo translation and the Hilbert space $\mathbb{L}_2$, on which $\Gamma = \{ \gamma:[0,1] \to [0,1] \mid \gamma \text{ monotonically increasing, onto and differentiable} \}$ acts by isometries.
\end{lemma}

More precisely, the action of $\Gamma$ on the SRV transformed curves becomes
$\Gamma \times \mathbb{L}_2 \to \mathbb{L}_2$, $(\gamma, \boldsymbol{q}) = (\boldsymbol{q} \circ \gamma) \sqrt{\dot{\gamma}}$, which is by isometries since $\|(\boldsymbol{q}_1 \circ \gamma) \sqrt{\dot{\gamma}} - (\boldsymbol{q}_2 \circ \gamma) \sqrt{\dot{\gamma}}\|_{\mathbb{L}_2}^2 = \int_0^1 (\boldsymbol{q}_1(\gamma(t)) - \boldsymbol{q}_2(\gamma(t)))^2 \dot{\gamma}(t) dt = \int_0^1 (\boldsymbol{q}_1(t) - \boldsymbol{q}_2(t))^2 dt = \| \boldsymbol{q}_1  - \boldsymbol{q}_2 \|_{\mathbb{L}_2}^2$ for all $\gamma \in \Gamma$. That means we can define an elastic distance $\dist$ on $\mathcal{A}/\Gamma$ modulo translation as the quotient metric ($\dist_G$ in Definition \ref{defi:quotient_metric}) on $\mathbb{L}_2/\Gamma$.

\begin{defi}[Elastic distance \citep{srivastava_book}]\label{lem:curve_srv}
Let $[\boldsymbol{y}_1], [\boldsymbol{y}_2]$ be equivalence classes in $\mathcal{A}/\Gamma$ modulo translation. Then the elastic distance
\begin{equation} \label{def:elastic_dist}
\dist([\boldsymbol{y}_1], [\boldsymbol{y}_2]) =
\displaystyle \inf_{\gamma_1, \gamma_2 \in \Gamma} \| Q(\boldsymbol{y}_1 \circ \gamma_1) - Q(\boldsymbol{y}_2 \circ \gamma_2) \|_{\mathbb{L}_2},
\end{equation}
is a proper metric. Here $\| \mathbf{q}\|_{\mathbb{L}_2} = (\int_0^1 \|\mathbf{q}(t)\|^2 dt)^{1/2}$, $\mathbf{q} \in \mathbb{L}_2$, denotes the usual $\mathbb{L}_2$ norm.
\end{defi}

Thus, we can define a quotient regression model for SRV curves modulo re-parametrization as in Subsection \ref{subsec:quotient_reg}. We formulate a regression model for the elastic curves themselves using the inverse of the SRV transformation $Q$, which is given via $Q^{-1}(\boldsymbol{q})(t) = \int_0^t \boldsymbol{q}(s) \| \boldsymbol{q}(s) \| \ ds$ for all $\boldsymbol{q} \in \mathbb{L}_2$.

\begin{defi}[Quotient SRV-linear regression for elastic curves] \label{defi:quotient_reg_elastic}
Let $\boldsymbol{x}_i = (x_{i,1}, \dots, x_{i,k})^\top \in \R^k$, $i = 1, \dots, n$ be realizations of a random vector $\boldsymbol{X} = (X_1, \dots, X_k)^\top$ and $\boldsymbol{q}_1, \dots, \boldsymbol{q}_n \in \mathbb{L}_2$ be SRV transformations of realizations of a random variable $[\boldsymbol{Y}]$ taking values in $\mathcal{A}/\Gamma$, where $\mathcal{A}$ is the set of absolutely continuous curves from $[0,1]$ to $\R^d$ and $\Gamma$ the set of monotonically increasing, onto and differentiable re-parametrizations. On curve level, the quotient linear regression model then becomes
\begin{align*}
    f(\boldsymbol{x}) = f(x_1, \dots, x_k) = \mathcal{E}([\boldsymbol{Y}]| X_1 = x_1, \dots, X_k = x_k) = \left[Q^{-1}\left(\varphi(\boldsymbol{x})\right)\right] 
\end{align*}
with linear predictor 
$$
 \varphi(\boldsymbol{x}) = \boldsymbol{\beta}_0 + \sum_{j = 1}^k \boldsymbol{\beta}_j x_{j}
$$
on SRV-level.
%Here $Q$ acts similarly to the link function in a generalized linear model, linking the linear predictor in $\mathbb{L}_2$ to the conditional expectation of the response.
The coefficients $\boldsymbol{\beta}_0, \dots, \boldsymbol{\beta}_k \in \mathbb{L}_2$ of the regression function are estimated as
\begin{align*}
\argmin_{\boldsymbol{\beta}_0, \dots, \boldsymbol{\beta}_k \in \mathbb{L}_2
} \sum_{i = 1}^n \inf_{\gamma_i \in \Gamma} \left\| \boldsymbol{\beta}_0 + \sum_{j = 1}^k \boldsymbol{\beta}_j x_{i,j} - (\boldsymbol{q}_i \circ \gamma_i)\sqrt{\gamma_i} \right\|_{\mathbb{L}_2}^2.
\end{align*}
\end{defi}

We further assume that the parameters lie in a spline space, that is $\boldsymbol{\beta}_j(t)= \sum_{m = 1}^M \boldsymbol{\xi}_{j,m} B_m(t)$, $j = 1, \dots, k$,  where $\{B_m, m = 1, \dots M\}$ is a spline basis (e.g. linear B-splines) and $\boldsymbol{\xi}_{j,m} \in \mathbb{R}^d$ for all $j = 1, \dots, k$ and $m = 1, \dots, M$. We showed identifiability modulo warping of splines from several spline spaces in \cite{steyer}.

For SRV-transforms $Q(\boldsymbol{Y})$ this model directly corresponds to a quotient linear model (Definition \ref{defi:mult_quotient_reg}, with original space $\mathcal{Y} = \mathbb{L}_2$ and the respective isometric group action $\Gamma$ implied by re-parameterization of a curve $\boldsymbol{y}$ for its SRV transform $\boldsymbol{q} = Q(\boldsymbol{y})$.
As such, it enjoys consistency in the sense of Corollary \ref{cor:consistency} and, using the finite-dimensional spline space for modeling, also existence of a Fréchet mean, i.e.\ $f^\ast \neq \emptyset$, as we showed in Theorem \ref{theo:existence} in a more general setting. Due to Lemma \ref{lem:SRVtrafo}%{lem:curve_srv} 
we can equivalently understand the model on curve level. 

The minimization needed to estimate this quotient regression model for elastic curves is tackled via alternating between fitting a function-on-scalar model in each of the $d$ dimensions for fixed $\gamma_i$, and updating the optimal re-parametrizations $\gamma_i, i = 1, \dots, n$ for fixed $\boldsymbol{\beta}$s, see Algorithm \ref{algo:quotient_regression} below. The two alternated steps are generic in the sense that suitable warping and L2 fitting steps can be combined that are tailored to the situation at hand (e.g. densely vs. sparsely observed curves). In our own implementation in the R-package \texttt{elasdics} \citep{elasdics}, since the data $\boldsymbol{q}_i, i = 1, \dots, n$ are SRV transformations of usually discretely observed curves, we use our methods specifically developed in \cite{steyer} for potentially sparse settings for both steps. That is we replace $\boldsymbol{q}_i$ by $\boldsymbol{\check{q}}_i$, the SRV transformation of the polygon $\boldsymbol{\check{y}}_i$ which is constructed via connecting the observed points linearly and choosing a constant speed parameterization. Note that this parameterization does not play a role for our model itself but only provides a suitable initial value. Also note that the relevant error made in this approximation, i.e. the difference between the polygon $\boldsymbol{\check{q}}_i$ and the unobserved curves $\boldsymbol{q}_i$, is the one at the SRV level. Accordingly, relatively densely observed points drawn with error at the curve level cause large errors at the SRV level (since the polygonal approximation corresponds to computing derivatives via finite differences). In this case it can be advantageous to coarsen the observed points first or to smooth them by a spline approximation on curve level.

%We further assume that the parameters lie in a spline space, that is $\boldsymbol{\beta}_j(t)= \sum_{m = 1}^M \boldsymbol{\xi}_{j,m} B_m(t)$, $j = 1, \dots, k$,  where $\{B_m, m = 1, \dots M\}$ is a spline basis (e.g. linear B-splines) and $\boldsymbol{\xi}_{j,m} \in \mathbb{R}^d$ for all $j = 1, \dots, k$ and $m = 1, \dots, M$. We showed identifiability modulo warping of splines from several spline spaces in \cite{steyer}. 

\begin{algorithm}[ht]
\caption{Quotient SRV-linear regression for elastic open curves \label{algo:quotient_regression}}
\KwIn
{data pairs $(\boldsymbol{x}_i, \boldsymbol{\check{q}}_i)$, $i = 1, \dots, n$, where  $\boldsymbol{\check{q}}_i$ are the SRV transformations of observed polygons $\boldsymbol{\check{y}}_i$ and  $\boldsymbol{x}_i = (x_{i,1}, \dots, x_{i,k})  \in \R^k$, $i = 1, \dots, n$ are observed covariates;
 convergence tolerance $\epsilon > 0$
}
Compute initial estimate $\hat{\boldsymbol{\beta}}_{0,new}, \dots, \hat{\boldsymbol{\beta}}_{k,new} =\argmin_{\boldsymbol{\beta}_0, \dots, \boldsymbol{\beta}_k \in \mathbb{L}_2
} \sum_{i = 1}^n \| \boldsymbol{\beta}_0 + \sum_{j = 1}^k \boldsymbol{\beta}_j x_{i,j} - \boldsymbol{\check{q}}_i \|_{\mathbb{L}_2}^2$;\\
Set $\hat{\boldsymbol{\beta}}_{j,old} = \text{Inf} \quad \forall j = 0, \dots, k$;\\
\While{$\max_{j = 0, \dots, k} \| \hat{\boldsymbol{\beta}}_{j,old} - \hat{\boldsymbol{\beta}}_{j,new} \|_{\mathbb{L}_2}^2 > \epsilon $}{
$\hat{\boldsymbol{\beta}}_{j,old} = \hat{\boldsymbol{\beta}}_{j,new}  \quad \forall j = 0, \dots, k$\;
$\gamma_i = \argmin_{\gamma} \left\| \hat{\boldsymbol{\beta}}_{0,old} + \sum_{j = 1}^k \hat{\boldsymbol{\beta}}_{j,old} x_{i,j} - (\boldsymbol{\check{q}}_i \circ \gamma)\sqrt{\gamma} \right\|_{L_2}^2, \quad \forall i = 1, \dots, n$  \tcp*{warping step}
$\hat{\boldsymbol{\beta}}_{0,new}, \dots, \hat{\boldsymbol{\beta}}_{k,new} = \arginf_{\boldsymbol{\beta}_0, \dots, \boldsymbol{\beta}_k \in \mathbb{L}_2
} \sum_{i = 1}^n  \left\| \boldsymbol{\beta}_0 + \sum_{j = 1}^k \boldsymbol{\beta}_j x_{i,j} - (\boldsymbol{\check{q}}_i \circ \gamma_i)\sqrt{\gamma_i} \right\|_{\mathbb{L}_2}^2$ \\ \tcp*[f]{$L_2$ spline fit via least-squares}}
\Return{$\hat{\boldsymbol{\beta}}_{j} = \hat{\boldsymbol{\beta}}_{j,new} \quad \forall j = 0, \dots, k$}
\end{algorithm}

Note that the spline model assumption is not compatible to a geodesic model assumption. Although geodesic lines are contained in the quotient space regression model assumption as shown in Lemma \ref{lem:shortest_path_quotient}, geodesics between two spline curves do in general not lie in a spline space (see Subsection \ref{subsec:geodesic_splines}), since aligning one spline curve to another does in general not result in a spline curve. Thus, a model can not be a geodesic model and a spline model at the same time, but we can use a spline model to approximate  a geodesic model. %On the other hand, if the true model lies in the spline space, using splines as a model space will also yield consistent estimates due to Lemma \ref{lem:consistency}.

\subsection{Extensions to closed curves} \label{subsec:closed_curves}
Since the space of SRV curves belonging to closed curves, $\{ \boldsymbol{p} \in \mathbb{L}_2 | \int_0^1 \boldsymbol{p}(t) \|\boldsymbol{p}(t)\| dt = \boldsymbol{0} \in \R^d \}$, does not form a linear subspace in $\mathbb{L}_2$, regression of closed curves cannot be treated analogously to that of open curves. While in principle it would be possible to consider the space of closed curves as a submanifold of $\mathbb{L}_2$ and then define the quotient regression model on this submanifold modulo warping, to the best of our knowledge there are no methods to compute minimizing geodesics on this submanifold. (\citet{srivastava_book} provide algorithms for numerical computation of geodesics between two closed curves -- extending this to finding a minimizing geodesic through a sample of curves is, however, not straightforward). For this reason, we do not focus on closed curves here. However, as closed curves often appear naturally in practical applications, we describe at least a heuristic method for the regression of closed curves based on quotient regression for open curves. This method is also implemented in the R-package \texttt{elasdics} \citep{elasdics}.

Specifically, we treat the curves as open curves in the $\mathbb{L}_2$ fitting step, but restrict the splines we use for modeling their SRV transforms to be closed %fitting on SRV level to closed curves 
(which is necessary but not sufficient for closedness of the modeled curves, ensuring matching derivatives at starting and end points). %While this does not ensure closedness of the modeled curves, it does force the derivatives of the curves at the start and the end point to match.
Then we close the predictions via projecting them onto the space of derivatives belonging to closed curves: %(See Algorithm \ref{algo:quotient_regression_closed} in the Appendix for more details.) 
Since we model the SRV transform $\boldsymbol{p}$ as a spline  and therefore bounded curve, the corresponding derivative $\boldsymbol{p} \|\boldsymbol{p}\|$ is also bounded and therefore in $\mathbb{L}_2$. Hence we can consider the space $\{ \boldsymbol{p} \|\boldsymbol{p}\| \in \mathbb{L}_2 | \int_0^1 \boldsymbol{p}(t) \|\boldsymbol{p}(t)\| dt = 0 \}$, which is a linear subspace of the Hilbert space $\mathbb{L}_2$, and compute the orthogonal projection of $\boldsymbol{p} \|\boldsymbol{p}\|$ onto this space as $\boldsymbol{p} \|\boldsymbol{p}\| - \int_0^1 \boldsymbol{p}(s) \|\boldsymbol{p}(s)\| ds$. Thus, the prediction on curve level becomes $t \mapsto \int_0^t \boldsymbol{p}(s) \|\boldsymbol{p}(s)\| ds- t \cdot \int_0^1 \boldsymbol{p}(s) \|\boldsymbol{p}(s)\| ds$, which is a closed curve. We use these closed predictions in the iterative algorithm \ref{algo:quotient_regression} to replace the $\hat{\boldsymbol{\beta}}_{j,old}$ when aligning  the observations in each iteration (warping step). See Algorithm \ref{algo:quotient_regression_closed} in the Appendix for details.
\section{Alternative regression approaches} \label{sec:alternatives}
Although there are so far no direct competitors available to our quotient regression for curves modulo re-parametrization, we discuss in the following different approaches that we consider natural alternatives. %approaches that can be used to model curves modulo re-parametrization as a function of Euclidean covariates.
Comparison to these alternatives may be relevant %also offer clues in different directions 
beyond our specific focus as they exemplify 
a) pre-alignment as natural alternative to quotient regression with responses in any quotient metric space, 
b) statistical modeling on curve level with only alignment based on SRV transforms,
and c) usage of a generic approach for metric spaces without using the quotient structure.
The first three alternatives we give are new proposals reflecting combinations of a) and b), while for c)  Fr\'echet regression in Subsection \ref{subsec:frechet} constitutes an existing general approach, which has to be adapted to and implemented for our setting, and for which we give a novel concrete implementation for the elastic regression case. All methods discussed here will then be used as comparison methods to benchmark our quotient regression approach in simulations in Section \ref{sec:simulation}.

\subsection{Regression after pre-alignment}
For elastic regression as for general quotient metric spaces $\mathcal{Y}/G$ where $(\mathcal{Y}, \|\cdot \|_\mathcal{Y})$ is a Hilbert space, an obvious competitor of the quotient linear model is to fit a linear model on the original space $\mathcal{Y}$ after once pre-aligning the data $y_i$, $i = 1, \dots, n$ to its (marginal) Fréchet mean $\mu_0 =  \arginf_{\mu \in \mathcal{Y}} \sum_{i = 1}^n \dist_G([y_i], [\mu])^2$. Here, we consider the model with predictor $f: x \mapsto [\varphi(x)]$ and the estimator $\hat{\varphi}(\boldsymbol{x}) = \hat{\beta}_0 + \hat{\beta}_1 x_1 + \dots + \hat{\beta}_k x_k$ given by
\begin{align*}
\boldsymbol{\hat{\beta}} = (\hat{\beta}_0, \dots, \hat{\beta}_k)^\top = \argmin_{\beta_0, \dots, \beta_k \in \mathcal{Y}} 
\sum_{i = 1}^n  \left\| \beta_0 + \sum_{j = 1}^k \beta_j x_{i,j} - g_i^\ast \circ y_i \right\|^2_\mathcal{Y}
\end{align*}
where $g_i^\ast = \argmin_{g \in G} \|\mu_0 - g \circ y_i \|_\mathcal{Y}$. Here we assume that there exists an optimal alignment to the mean for all $y_i$'s. The minimiser $\boldsymbol{\hat{\beta}}$ can be computed as $ \boldsymbol{\hat{\beta}} = (\boldsymbol{\Xi}^\top \boldsymbol{\Xi})^{-1}\boldsymbol{\Xi}^\top (\mathbf{g^\ast \circ y})$, where $\boldsymbol{\Xi} \in \mathbb{R}^{n \times k}$ is the design matrix and $\mathbf{g^\ast \circ y} = (g_1^\ast \circ y_1, \dots, g_n^\ast \circ y_n)^\top \in \otimes_{i = 1}^n \mathcal{Y}$. 

Although the model space $\Phi$ also consists of affine linear functions in $\mathcal{Y}$, this is not an intrinsic regression, i.e. we do not truly consider its projection to the quotient $\Phi\doubleslash G$ as model space on $\mathcal{Y}/G$ here. That means no attempt is made to minimize the empirical risk \eqref{eq:mult_quotient_reg} with respect to the quotient space distance $\dist_G$ and therefore, this risk will always be greater than or equal to that for the quotient space regression model. 

In the specific case that we want to model curves with respect to the elastic distance \eqref{def:elastic_dist}, this means computing a linear model for the SRV transformed curves in $\mathbb{L}_2$ after pre-aligning the corresponding data curves to the elastic mean. That is
\begin{align*}
(\hat{\boldsymbol{\beta}}_0, \dots, \hat{\boldsymbol{\beta}}_k) = \argmin_{\boldsymbol{\beta}_0, \dots, \boldsymbol{\beta}_k \in \mathbb{L}_2
} \sum_{i = 1}^n \left\| \boldsymbol{\beta}_0 + \sum_{j = 1}^k \boldsymbol{\beta}_j x_{i,j} - (\boldsymbol{q}_i \circ \gamma_i)\sqrt{\gamma_i} \right\|_{\mathbb{L}_2}^2
\end{align*}
with $\gamma_i = \argmin_{\gamma \in \Gamma} \|\boldsymbol{\mu}_0 - (\boldsymbol{q}_i \circ \gamma_i)\sqrt{\gamma_i} \|_{\mathbb{L}_2}$ and $\boldsymbol{\mu}_0$ is the SRV transformation of the elastic mean curve. \citep{guo} propose a similar procedure, where they then use the principal component scores of the pre-aligned SRV curves in a simple regression model. In contrast, we use splines to model the $\boldsymbol{\beta}$s and the alignment methods developed in \cite{steyer} to enable fitting of irregularly and/or sparsely observed curves and to allow better comparison with our quotient regression model for elastic curves (Definition \ref{defi:quotient_reg_elastic}). We refer to this procedure as 'pre-align, srv fit' in the following.

\subsection{Alternative procedures with fit on curve level}
Considering pre-alignment of the data curves $\boldsymbol{y}_1, \dots, \boldsymbol{y}_n$ with SRV transformations $\boldsymbol{q}_1, \dots, \boldsymbol{q}_n$ to their elastic mean curve a pre-processing step, it might also deem natural to compute the regression model on curve level instead of on SRV level. We call this approach 'pre-align, curve fit'.
Here, the fitted predictor is given by $\hat{\boldsymbol{f}}(\boldsymbol{x}) = \hat{\boldsymbol{\beta}}_0 + \hat{\boldsymbol{\beta}}_1 x_1 + \dots + \hat{\boldsymbol{\beta}}_k x_k$ with
\begin{align*}
(\hat{\boldsymbol{\beta}}_0, \dots, \hat{\boldsymbol{\beta}}_k) = \argmin_{\boldsymbol{\beta}_0, \dots, \boldsymbol{\beta}_k \in \mathbb{L}_2
} \sum_{i = 1}^n \left\| \boldsymbol{\beta}_0 + \sum_{j = 1}^k \boldsymbol{\beta}_j x_{i,j} - \boldsymbol{y}_i \circ \gamma_i^\ast \right\|_{\mathbb{L}_2}^2
\end{align*}
where $\gamma_i^\ast = \argmin_{\gamma \in \Gamma} \|\boldsymbol{\mu}_0 - (\boldsymbol{q}_i \circ \gamma)\sqrt{\gamma} \|_{\mathbb{L}_2}$ and $\boldsymbol{\mu}_0$ is again the SRV transformation of the elastic mean curve. This is tempting in particular if we want to fit closed curves since, on curve level, closed curves can be modeled without further modifications using a closed spline basis for the model coefficients $\boldsymbol{\beta}_0, \dots, \boldsymbol{\beta}_k$.

We further consider a heuristic procedure in which we alternate between optimal alignment and regression fit as in the quotient regression approach, but fit the linear model on curve level rather than on SRV level ('iterate align, curve fit'). This is not a suitable method for fitting the quotient regression model with respect to the elastic distance, because the elastic distance becomes the usual $\mathbb{L}_2$ metric only for SRV transforms. Fitting the linear model on curve instead of on SRV level will not return a minimizer of the squared elastic distances to the data curves. In fact, there is no risk function that this algorithm aims to minimize, and the procedure is thus only defined by the iterative algorithm rather than being the fitting algorithm of a regression model.

Moreover, both procedures with linear model fit on curve level do not include geodesics with respect to the elastic distance in their model space, i.e., they are not suitable to generalize linear regression in this sense.

\subsection{Fréchet regression} \label{subsec:frechet}
So far we considered models that exploit the linear space structure of either the space on SRV or on curve level to define regression models for curves with respect to the elastic distance. In contrast, \cite{petersen} developed a regression model they call Fréchet regression for random objects lying in arbitrary metric spaces with covariates in $\mathbb{R}^k$, which does not rely on any linear structure. They achieve this by noting that in standard linear regression, the regression function can be viewed as a function mapping the input $\boldsymbol{x} \in  \mathbb{R}^k$ to a weighted mean of the $y_i$, where only the weights depend on $\boldsymbol{x}$. Their Fréchet regression model then extends standard linear regression by using the same weights with an arbitrary metric instead of the Euclidean distance, i.e. using a weighted Fréchet mean. Although this implicitly defines a regression model for arbitrary metric spaces, without explicit model equation however, details and complexity of the estimation depend on the specific space considered. \cite{petersen} discuss in their paper the case of propability distributions equipped with the Wasserstein metric as well as the case of covariance matrices. For both cases, there is an implementation in the \texttt{R} package \texttt{frechet}\citep{frechet}. To the best of our knowledge, the case of curves with respect to the elastic distance has not yet been considered, so we describe below how we estimate the Fréchet regression model in this case. 

For observed curves with SRV transforms $\boldsymbol{q}_1, \dots, \boldsymbol{q}_n \in \mathbb{L}_2$, the predictor $\boldsymbol{f}$ for an input vector $\boldsymbol{x} \in \R^k$ is given by
\begin{align*}
\boldsymbol{f}(\boldsymbol{x})  = \argmin_{\boldsymbol{p} \in \mathbb{L}_2}
\sum_{i = 1}^n s(\boldsymbol{x}_i, \boldsymbol{x}) \inf_{\gamma_i \in \Gamma} \|\boldsymbol{p} - (\boldsymbol{q}_i \circ \gamma_i) \sqrt{\gamma_i}\|_{\mathbb{L}_2}^2,
\end{align*}
via a point-wise optimization function where the weights \citep{petersen} are given as $s(\boldsymbol{x}_i, \boldsymbol{x}) = 1 + (\boldsymbol{x}_i - \bar{\boldsymbol{x}})^\top \boldsymbol{\hat{\Sigma}}^{-1} (\boldsymbol{x} - \bar{\boldsymbol{x}})$. Here $\bar{\boldsymbol{x}} = \frac{1}{n} \sum_{i = 1}^n \boldsymbol{x}_i$ is the mean of the observed covariates and $\boldsymbol{\hat{\Sigma}} = \frac{1}{n} \sum_{i = 1}^n (\boldsymbol{x}_i - \bar{\boldsymbol{x}})(\boldsymbol{x}_i - \bar{\boldsymbol{x}})^\top$ their empirical covariance matrix.
Thus, for a given input value $\boldsymbol{x} \in \mathbb{R}^k$, the conditional mean response curve is computed as a weighted Fréchet mean with respect to the elastic distance using weights $s(\boldsymbol{x}_i, \boldsymbol{x})$. For the particular case of the space of SRV curves with the elastic metric, we propose to consider the observed polygons $\boldsymbol{\check{y}}_1, \dots, \boldsymbol{\check{y}}_n$ with SRV transformations $\boldsymbol{\check{q}}_1, \dots, \boldsymbol{\check{q}}_n$ as we do for our quotient space regression model to handle discretely observed curves. Then we estimate the weighted Fréchet mean via alternating between updating the optimal re-parametrizations $\gamma_i \in \Gamma$ as $\argmin_{\gamma_i \in \Gamma} \|\boldsymbol{p} - (\boldsymbol{\check{q}}_i \circ \gamma_i) \sqrt{\gamma_i}\|_{\mathbb{L}_2}$ for a given $\boldsymbol{p} \in \mathbb{L}_2$, using our alignment methods developed in \cite{steyer} to align discretely observed curves to a model based curve, and computing the weighted $\mathbb{L}_2$-mean 
$ \argmin_{\bar{\boldsymbol{p}}} \sum_{i = 1}^n s(\boldsymbol{x}_i, \boldsymbol{x}) \left\| \bar{\boldsymbol{p}} - (\boldsymbol{\check{q}}_i \circ {\gamma_i}) \sqrt{\dot{\gamma_i}} \right\|_{\mathbb{L}_2}^2$ for given alignments $\gamma_i$. For the mean estimation step we propose to use splines, as we do for the quotient regression model. Details are given as Algorithm \ref{algo:frechet_reg} in the Appendix.

One disadvantage of this approach is that the regression function is not given by a set of parameters, such as slopes and intercepts. In fact, for every given input vector $\boldsymbol{x}\in \mathbb{R}^k$, the value of the regression function has to be estimated separately as a weighted mean. This makes interpretation of the model more challenging and estimation more time consuming. One advantage in the SRV context is that handling closed curves is straightforward, as we can compute closed (weighted) Fréchet means using results in \cite{steyer}.

%\subsection{Comparison of the proposed approaches}
\subsection{Differences in curve alignment implied by the different approaches}
To gain an understanding of the differences between the proposed approaches, we compare how the observed curves are aligned during the fitting process and discuss the implications of these differences in specific data scenarios. When fitting the quotient regression model, we align the observed curves $\boldsymbol{\check{y}}_i$ to the model based predictions $Q^{-1}(\hat{\boldsymbol{\beta}}_0 + \sum_{j = 1}^k \hat{\boldsymbol{\beta}}_j x_{i,j})$ for all $i = 1, \dots, n$. This means that each observed $\boldsymbol{\check{y}}_i$ is aligned to a model-based curve that is expected to have similar features as the observation. Likewise, in the 'iterate align, curve fit' approach, the observed $\boldsymbol{\check{y}}_i$ is aligned to its associated prediction.

In contrast, the pre-alignment methods 'pre-align, srv fit' and 'pre-align, curve fit' align the curves to the elastic mean, hence may not properly align certain features of the curves if these features occur in specific directions of $\boldsymbol{x}$ that are missing in the mean curve. Similarly, in the fitting algorithm for the Fréchet regression model (Algorithm \ref{algo:frechet_reg}) the observed curves $\boldsymbol{\check{y}}_i$  are aligned to the model prediction for each considered new value of $\boldsymbol{x}$, which is usually different from $\boldsymbol{x}_i$. Accordingly, we also expect less convincing results for this model in situations where certain features of the curves occur only for some values of $\boldsymbol{x}$ but not for others.

Overall, we expect all five methods to provide satisfactory results in scenarios where all observed curves have similar features, and that the quotient regression model outperforms the fits after pre-alignment as well as the Fréchet regression model when some features are missing in the elastic mean curve respectively some of the curves. For the 'iterate align, curve fit' approach, the behavior is more difficult to anticipate, as its iterative procedure optimizes no loss function. 
We will investigate these expectations for model performance in simulations in Section \ref{sec:simulation}. 
Besides that Section \ref{sec:simulation} will also cover simulations on methods of inference in quotient regression, which we describe beforehand in the next section.

\section{Inference and model selection} \label{sec:inference}
\subsection{A generalized coefficient of determination}
Both Fréchet regression and the quotient regression are defined as empirical risk minimization problems in one way or another. 
%This means that no underlying probability distributions are available, for example, to draw new data according to a particular regression model. Accordingly, only the empirical risk can be used for inference and model selection. 
\cite{petersen} generalize the coefficient of determination $R^2$ to models with values $y_1, \dots, y_n$, $n \in \mathbb{N}$ in metric spaces $(\mathcal{Y}, \dist)$. For an estimated model equation $\hat{f}$ their Fréchet coefficient of determination is given as
\begin{align*}
    \tilde{R}^2 = 1 - \frac{\sum_{i = 1}^n \dist(y_i, \hat{f}(\boldsymbol{x}_i))^2}{\sum_{i = 1}^n \dist(y_i, \hat{\mu}_0)^2}
\end{align*}
where $\hat{\mu}_0 =  \argmin_{\mu \in \mathcal{Y}} \sum_{i = 1}^n \dist(y_i, \mu)^2$ is the Fréchet mean of the data. Note that if constant functions are contained in the model space in which $\hat{f}$ is estimated, we have $\tilde{R}^2 \in [0, 1]$ as for $R^2$ in standard linear regression. In this case testing the global null hypothesis of no effect, that is ${f}(\boldsymbol{x}) \equiv \mu_0$ constant, is equivalent to testing $H_0 :\tilde{R}^2 = 0$. The distribution of the test statistic $\tilde{R}^2$ under $H_0$ is available via permutation re-sampling of the data, i.e randomly permuting the labels of the response variable $y_i$ while keeping the covariates $\boldsymbol{x}_i$ fixed. They further suggest to use an adjusted coefficient of determination $\tilde{R}^2_{\text{adj}} = 1 - (1 - \tilde{R}^2) \frac{n -1}{n-k-1}$ for model selection, where $k$ accounts for the number of covariates $\boldsymbol{x}_i = (x_{i1}, \dots, x_{ik})^T$ in the model.

\subsection{Distance-based bootstrap confidence regions}\label{subsec:curveCI}
To obtain confidence regions for the predicted curves we propose to bootstrap the data $(\boldsymbol{x}_i, \boldsymbol{\check{y}}_i)$, $i = 1, \dots, n$ to obtain an approximate sample of the model predictions, $\boldsymbol{\hat{y}}_1, \dots, \boldsymbol{\hat{y}}_{N_{boot}}$, for a given $\boldsymbol{x}$. From this we construct a $(1- \alpha)$-confidence region as a generalized convex hull \citep[$\alpha$-shapes]{edelsbrunner}, of the (centered, i.e we subtract the center of mass for each predicted curve) $\lceil (1- \alpha) N_{boot} \rceil$ closest curves to the bootstrap mean with respect to the elastic distance. Note that when the bootstrapped curves form a relatively dense set, directly plotting the (1-$\alpha$) closest curves gives a good and simple visual approximation to plotting the generalized convex hull in practice. 

\subsection{Bootstrap confidence regions based on spline coefficients} \label{subsec:coefCI}
Inference as described above can be conducted for approaches without a parametric model equation, such as Fréchet regression, and parametric models, such as the quotient regression model, which provide estimates for intercept and slope parameters. However, since our quotient linear model for elastic curves is a parametric model, we are not only interested in the global null hypothesis of none of the covariates having an effect, but also want to assess the relevance of individual parameters. We propose to test individual hypotheses by bootstrapping the data $(\boldsymbol{x}_i, \boldsymbol{\check{y}}_i)$, $i = 1, \dots, n$ to obtain an approximate sample from the distribution of the estimated model parameters $\hat{\boldsymbol{\beta}}_0, \dots, \hat{\boldsymbol{\beta}}_k$.  Confidence regions for the parameters can then be constructed from this sample and used to decide whether a particular parameter, for instance $\boldsymbol{\beta}_j = \boldsymbol{0}$ corresponding to no effect, is plausible given the observed data, as detailed below. 

Our proposed representation of the coefficient functions $\boldsymbol{\beta}_j(t)= \sum_{m = 1}^M \boldsymbol{\xi}_{j,m} B_m(t)$, $t \in [0,1]$, $j = 1, \dots, k$ has the additional advantage that using a linear combination of spline basis functions $B_m(t)$, $m = 1, \dots, M$ with local support, such as B-splines, also allows to test local individual hypotheses on subintervals of $[0,1]$, i.e.\ to test where a given covariate affects the response curve. 
We have shown in \cite{steyer} that linear splines on SRV level (among other splines) are identifiable via their spline coefficients modulo parametrization, and that the mapping between the spline coefficients and the elastic curves is a homeomorphism. 
We can thus use the variation in the spline coefficients as representative of that in the estimated effects and construct alternative confidence regions as outlined in the following.
Note that this alternative to Section \ref{subsec:curveCI} is, however, only recommended when estimates are sufficiently concentrated, as we will briefly discuss in Section \ref{subsec:curveCI_vs_coefCI}. %, as long as the sample size $n$ is sufficiently large to allow for practical matching of spline coefficients across different bootstrap estimates, i.e.\ intuitively spoken for identifiability to take effect in practice. %carry over to the approximating estimated curve. 
%Especially when we estimate very flexible curves with many basis functions relative to the sample size, %the variability in the estimated aligned 
%This assumption is crucial because otherwise, due to different 
%parameterizations can 
%Otherwise, the variability in the estimated aligned parameterizations can lead to basis coefficients no longer corresponding to related locations on the estimated curves, 
%especially when we estimate very flexible curves with many basis functions relative to the sample size. %We expect this especially when we estimate very flexible curves with many basis functions but have observed relatively few curves. 
%In these cases, inference based on the spline coefficients will lead to a loss in power due to the added variability of the parameterization,  and the distance-based methods described above should be used instead.

%If one considers the variation in spline coefficients to be representative of the variation in estimated effects, we propose to 
We construct a $(1- \alpha)$-confidence region for $\boldsymbol{\beta}_j$ based on the bootstrapped spline coefficients $\boldsymbol{\xi}_{j,m}^{(b)}$, $b = 1, \dots, N_{boot}$ as the $d$-dimensional ellipse
\begin{align*}
C_{j,m, \alpha} = \{\boldsymbol{\xi}\in \R^d | (\boldsymbol{\xi} - \bar{\boldsymbol{\xi}}_{j,m})^T \boldsymbol{\hat{\Sigma}}^{-1}_{j,m} (\boldsymbol{\xi} - \bar{\boldsymbol{\xi}}_{j,m}) \leq c_{j,1 - \alpha} \},
\end{align*}
where $\bar{\boldsymbol{\xi}}_{j,m} = \frac{1}{N_{boot}}\sum_{b = 1}^{N_{boot}} \boldsymbol{\xi}_{j,m}^{(b)}$ is the bootstrap mean, $\boldsymbol{\hat{\Sigma}}_{j,m} = \frac{1}{N_{boot} - 1} \sum_{b = 1}^{N_{boot}} (\boldsymbol{\xi}_{j,m}^{(b)} - \bar{\boldsymbol{\xi}}_{j,m})(\boldsymbol{\xi}_{j,m}^{(b)} - \bar{\boldsymbol{\xi}}_{j,m})^T$ is the bootstrap sample covariance and $c_{j,m ,1 - \alpha}$ the empirical $(1-\alpha)$-quantile of the studentized bootstrap sample $\{ (\boldsymbol{\xi}_{j,m}^{(b)} - \bar{\boldsymbol{\xi}}_{j,m})^T \boldsymbol{\hat{\Sigma}}^{-1}_{j,m} (\boldsymbol{\xi}_{j,m}^{(b)} - \bar{\boldsymbol{\xi}}_{j,m})| b = 1, \dots, N_{boot}\}$ for all $j = 1, \dots, k$. From this confidence regions for the coefficients $\boldsymbol{\xi}_{j,m}$ on can proceed to construct pointwise confidence regions for the corresponding effect functions $\boldsymbol{\beta}_j$. Moreover, $C_{j,m, \alpha}$ can also be used to test the local individual hypothesis $H_{0,j,m}: {\boldsymbol{\xi}}_{j,m} =\boldsymbol{0}$ by checking for overlap with $\boldsymbol{0}$. 
Using these confidence regions for the single spline coefficients, we construct a joint $(1 - \alpha)$-confidence region for the matrix of spline coefficients $\boldsymbol{\xi}_j = (\boldsymbol{\xi}_{j,1}, \dots \boldsymbol{\xi}_{j,M})^T \in \R^{M \times d}$ corresponding to the effect function $\boldsymbol{\beta}_j$ as $C_{j, \alpha} = \bigtimes_{m = 1}^M C_{j, m, \frac{\alpha}{M}}$, where $\frac{\alpha}{M}$ is a Bonferroni-type correction of the confidence level. Hence $P(\boldsymbol{\xi}_j  \in C_{j, \alpha}) = P(\bigcap_{m = 1}^M \{ \boldsymbol{\xi}_{j,m} \in  C_{j,m, \frac{\alpha}{M}} \}) =   1 - P(\bigcup_{m = 1}^M {\{ \boldsymbol{\xi}_{j,m} \notin  C_{j,m, \frac{\alpha}{M}} \}}) \geq 1 - \sum_{m = 1}^M P({\{ \boldsymbol{\xi}_{j,m} \notin C_{j,m, \frac{\alpha}{M}}}) \geq 1 - \alpha$, if $C_{j,m, \frac{\alpha}{M}}$ is a valid confidence region, i.e.\ fulfills $P({\{ \boldsymbol{\xi}_{j,m} \notin  C_{j,m, \frac{\alpha}{M}}}) %= 1 - P(\boldsymbol{\xi}_{j,m} \in  C_{j,m, \frac{\alpha}{M}}) 
\leq \frac{\alpha}{M}$.

The constructed confidence region can be utilized to test the individual hypothesis $H_{0,j}: \boldsymbol{\beta}_j = \boldsymbol{0}$.  This is done by rejecting $H_{0,j}$ if and only if $\boldsymbol{0} \neq C_{j,m}$, which is equivalent to $\bar{\boldsymbol{\xi}}_{j,m}^T \hat{\Sigma}^{-1}_{j,m} \bar{\boldsymbol{\xi}}_{j,m} \geq c_{j,1 - \frac{\alpha}{M}}$ for at least one $m = 1, \dots, M$. We thus use $\displaystyle \max\{ \bar{\boldsymbol{\xi}}_{j,m}^T \hat{\Sigma}^{-1}_{j,m} \bar{\boldsymbol{\xi}}_{j,m} | m = 1, \dots, M \}$ as a test statistic. Since the resulting test  relies  on the local representation property of the spline coefficients for   the effect functions and, as a bootstrap method, also on the interchangeability of the data generating distribution with the empirical distribution, we examine the validity and power of the test in a simulation in the following subsection. 

\subsection{Distance vs. spline coefficient based confidence regions} \label{subsec:curveCI_vs_coefCI}

The idea of the spline coefficient based confidence regions proposed in Section \ref{subsec:coefCI} is based on the assumption that the distribution of the $\hat{\boldsymbol{\beta}}_j$, or alternatively the bootstrap samples $\hat{\boldsymbol{\beta}}^{(b)}_j$, is reflected well by an elliptical distribution of the respective spline coefficients $\hat{\boldsymbol{\xi}}^{(b)}_j$. 
Despite identifiability of the used piece-wise linear splines modulo re-parameterization \citep{steyer}, this does not necessarily have to be the case. In particular, if two estimators $\hat{\boldsymbol{\beta}}^{(1)}_j$ and $\hat{\boldsymbol{\beta}}^{(2)}_j$ differ too much, different curve alignment may result in the $m$-th spline coefficients $\hat{\boldsymbol{\xi}}^{(1)}_{j,m}$ and $\hat{\boldsymbol{\xi}}^{(2)}_{j,m}$ of each of them corresponding to different segments of the curves, which might occur especially when we estimate very flexible curves with many basis functions relative to the sample size. 
%, as long as the sample size $n$ is sufficiently large to allow for practical matching of spline coefficients across different bootstrap estimates, i.e.\ intuitively spoken for identifiability to take effect in practice. %carry over to the approximating estimated curve. 
%Especially when we estimate very flexible curves with many basis functions relative to the sample size, %the variability in the estimated aligned 
%This assumption is crucial because otherwise, due to different 
%parameterizations can 
%Otherwise, the variability in the estimated aligned parameterizations can lead to basis coefficients no longer corresponding to related locations on the estimated curves, 
%especially when we estimate very flexible curves with many basis functions relative to the sample size. %We expect this especially when we estimate very flexible curves with many basis functions but have observed relatively few curves. 
In these cases, inference based on the spline coefficients will lead to a loss in power due to the added variability of the parameterization,  and the distance-based methods described in Section \ref{subsec:curveCI} should be used.
Conversely, if the estimators $\hat{\boldsymbol{\beta}}^{(b)}_j$ are sufficiently concentrated, spline coefficient based methods allow for local investigation and might yield more power since they make use of elliptical confidence regions rather than depending on the distance to the bootstrap mean only.
\section{Simulations}
\label{sec:simulation}
We first compare in simulations the quotient regression model with the alternative procedures presented in Section \ref{sec:alternatives}. Then, in the second part of this section, we examine the test for the parameters of the quotient regression model based on the bootstrapped spline coefficients.
\subsection{Comparison of model performance}
We compare the quotient linear model to the procedures described in Section \ref{sec:alternatives}. To this end, we choose three simulation scenarios for each of which we add errors of different magnitude and draw a varying number of points per curve. The predictive performance is then determined on an independent test set drawn according to the same principle, using the mean squared (elastic) distance (MSE) of the new observations to their predicted curves. Evaluation on a test set rather than on a true underlying model is necessary because quotient regression as well as Fréchet regression are defined as risk minimazion problems and no distribution is available that would allow us to draw random curves with a specific conditional Fréchet mean structure.
With the auxiliary sampling scheme used instead, we may specify a template model but there is no precise `true' model explicitly available that we can compare the model estimates to. Each of the $3\times 4 = 12$ simulations is then repeated 100 times to obtain a stable estimate of the MSE.

\begin{table}[ht]
\centering
\resizebox{\textwidth}{!}{
\begin{tabular}{|p{0.5cm} p{0.3cm} p{0.4cm} p{1.1cm} | p{1cm} p{0.9cm} p{1cm} p{0.9cm} p{1cm} | p{1cm} p{0.9cm} p{1cm} p{0.9cm} p{1cm}|}
\hline
& & & & \multicolumn{5}{c|}{MSE} & \multicolumn{5}{c|}{Average run time in seconds} \\
\hline
 sce-nario &sim & sd & $\kappa_i \in$ & quotient space regression & pre align, SRV fit & iterate align, curve fit & pre align, curve fit & Fréchet regression & quotient space regression & pre align, SRV fit & iterate align, curve fit & pre align, curve fit & Fréchet regression \\ 
\hline
1 & 1 & 0.4 & [15, 20] & \textbf{0.57} & 0.66 & 0.65 & 0.71 & 0.59 & 12 & 5 & 9 & 5 & 56 \\ 
1 &  2 & 0.8 & [15, 20] & \textbf{0.88} & 0.95 & 0.96 & 1.00 & 0.89 & 13 & 5 & 10 & 5 & 67 \\ 
1 &  3 & 0.4 & [30, 40] & \textbf{0.32} & 0.39 & 0.37 & 0.44 & 0.33 & 83 & 28 & 54 & 24 & 528 \\
1 &  4 & 0.8 & [30, 40] & \textbf{0.74} & 0.82 & 0.80 & 0.85 & 0.76 & 51 & 18 & 34 & 17 & 338 \\ \hline
2 &  5 & 0.2 & [15, 20] & \textbf{0.35} & 0.59 & 0.38 & 0.54 & 0.37 & 14 & 14 & 25 & 14 & 105 \\
2 &  6 & 0.4 & [15, 20] & \textbf{0.43} & 0.67 & 0.46 & 0.62 & 0.45 & 4 & 4 & 8 & 4 & 33 \\ 
2 &  7 & 0.2 & [30, 40] & \textbf{0.19} & 0.41 & 0.22 & 0.37 & 0.22 & 16 & 11 & 28 & 10 & 141 \\
2 &  8 & 0.4 & [30, 40] & \textbf{0.31} & 0.55 & 0.35 & 0.50 & 0.34 & 16 & 11 & 31 & 11 & 195 \\ \hline
3 &  9 & 0.1 & [15, 20] & \textbf{0.78} & 0.89 & 0.81 & 0.93 & 0.84 & 32 & 90 & 44 & 82 & 1922 \\ 
3 &  10 & 0.2 & [15, 20] & \textbf{1.37} & 1.49 & 1.39 & 1.52 & 1.41 & 38 & 102 & 58 & 97 & 1769 \\ 
3 &  11 & 0.1 & [30, 40] & \textbf{0.95} & 1.06 & \textbf{0.95} & 1.08 & 0.99 & 14 & 47 & 22 & 47 & 707 \\ 3 &  12 & 0.2 & [30, 40] & 2.80 & 2.96 & \textbf{2.79} & 2.95 & 2.81 & 13 & 48 & 20 & 48 & 528 \\ 
\hline
\end{tabular}
}
\caption{\label{tab:model_performance} Mean squared elastic distance (MSE) estimated out of sample (smallest per row in bold) and average run time of one estimation for the five methods in three different scenarios with a varying error magnitude (sd) and number of points drawn per curve, where the number of points $\kappa_i$ is drawn uniformly on a given range (15 to 20 or 30 to 40 points per curve). This gives a total of 12 simulations (sim).}
\end{table}

The three different simulation scenarios differ regarding which curves are used as models for $x = -1$ and for $x = 1$ and whether the trajectory between them is modeled linearly on SRV or on curve level. For the first scenario (simulations 1-4, see Fig. \ref{fig:simulated_models} (left) for an example of simulation 1), we use similar fish shapes to model the curves for $x = -1$ and for $x = 1$, and consider the geodesic between them (i.e. linear on SRV level with curves aligned). This setting is meant to be advantageous to methods using pre-alignment to the mean curve (fish), which should give good alignment among all curves, and we expect that all five methods should be able to model this type of data well. In contrast, in the second scenario (simulations 5-8, see Fig. \ref{fig:simulated_models} (middle) for an example of simulation 5) we also consider a linear relationship of the covariate $x$ with the SRV curves, but not a geodesic (i.e.\ no alignment with respect to the elastic distance between endpoints) between the curves for $x = -1$ (fish with open mouth) and $x = 1$ (fish with closed mouth). This seems natural since aligning the modeled curves here would not match the back end of the open mouth ($x = -1$) with the tip of the closed mouth ($x = 1$). In this setting we expect that pre-aligning the data to the elastic mean will not properly align the open/closed mouth of the fish and therefore %an iterative approach such as fitting 
the quotient linear model is beneficial. In the last simulation scenario (simulations 9-12, see Fig. \ref{fig:simulated_models} (right) for an example of simulation 11), we consider model misspecification in the sense that the effect of the covariate $x \in [-1, 1]$ is simulated linearly on curve instead of on SRV level. Additionally, in this setting we investigate the quality of our approach to modeling smooth, closed contours, here simulating closed quadratic spline curves.

To generate observations for the first scenario, we first obtain $n=11$ smooth curves for $x = -1, -0.8, \dots, 0.8, 1$ as the convex combinations of the SRV transformed modeled curves for $x = -1$ and $x = 1$. Next we evaluate them on a regular grid of 51 points from which we compute 50 SRV vector via finite difference approximation of the derivative. After adding a Gaussian 1st order random walk error with standard deviation $sd$ to these SRV vectors we back transform them to the curve level and select $\kappa_i$ of the resulting 51 points, where $\kappa_i$ is drawn uniformly on the given interval, to obtain sparse/irregular settings. Here we choose relatively small standard deviations $sd$ for the additional noise compared to the effect, since we want to focus on demonstrating structural differences between possible effects on the curves and how the different methods handle those. % effects on the curves. 

Each of the five regression models/procedures is fitted to these data assuming linear SRV splines with 11 knots for the quotient space regression model, the fit on SRV level after pre-alignment and the Fréchet regression model, and quadratic splines also with 11 knots for the models with fit on curve level. This results in the same model flexibility for all five models modulo translation. Since in this scenario, the modeled curves for $x = -1$ and $x = 1$ are approximately aligned, we expect all five methods to give reasonable results. This is confirmed visually in the model predictions (Fig.~\ref{fig:simulated_models}, left, and Fig. \ref{fig:simulation_1}), but the MSEs in rows 1-4 of Tab. \ref{tab:model_performance} reveal that the quotient space regression model performs best for this scenario and all combinations of $sd$ and  $\kappa_i$ .

The data for the second scenario, i.e., simulations 5-8, are generated in the same way as the data for the first scenario, except that the shape of the modeled curves for $x = -1$ and $x = 1$ differs more and we do not consider the geodesic between them as the generating
model. Since in this scenario the modeled curves have sharp edges around the mouth, we use constant splines on SRV-level corresponding to linear splines on curve level and 51 knots for all five procedures (see Fig.~\ref{fig:simulated_models}, middle, and Fig. \ref{fig:simulation_5} for an example of simulation 5). In this setting pre-aligning the data to the elastic mean (which also corresponds to the model prediction for $x = 0$ of the Fréchet regression model) will not properly align the open/closed mouth of the fish. Thus, a procedure that pre-aligns and then fits a model is not able to fit the open mouth of the fish for $x = -1$ (see Fig. \ref{fig:simulated_models}, middle). Similarly, for the Fréchet model fit the open mouth appears too small, as well as the whole predicted curve for $x = -1$ and $x = 1$. This is the case since for fitting the Fréchet model, we also align fish with open and closed mouth, since for each new value of $x$, we align all data curves to the corresponding new prediction (cf. Algorithm \ref{algo:frechet_reg}). Hence in this setting only the quotient regression model gives visually satisfying results, which is also reflected in the MSEs of the five models (Tab. \ref{tab:model_performance}, simulations 5 to 8). Here the MSE is always the smallest for the quotient regression model followed by Fréchet regression and the heuristic procedure of iterating between alignment and curve-level fit. We expect this to be the case in general if features of the curves (as for example the open mouth of the fish) that occur in certain directions of $x$ are missing in the mean curve.

\begin{figure}[ht]
    \centering
    \includegraphics[scale=0.65]{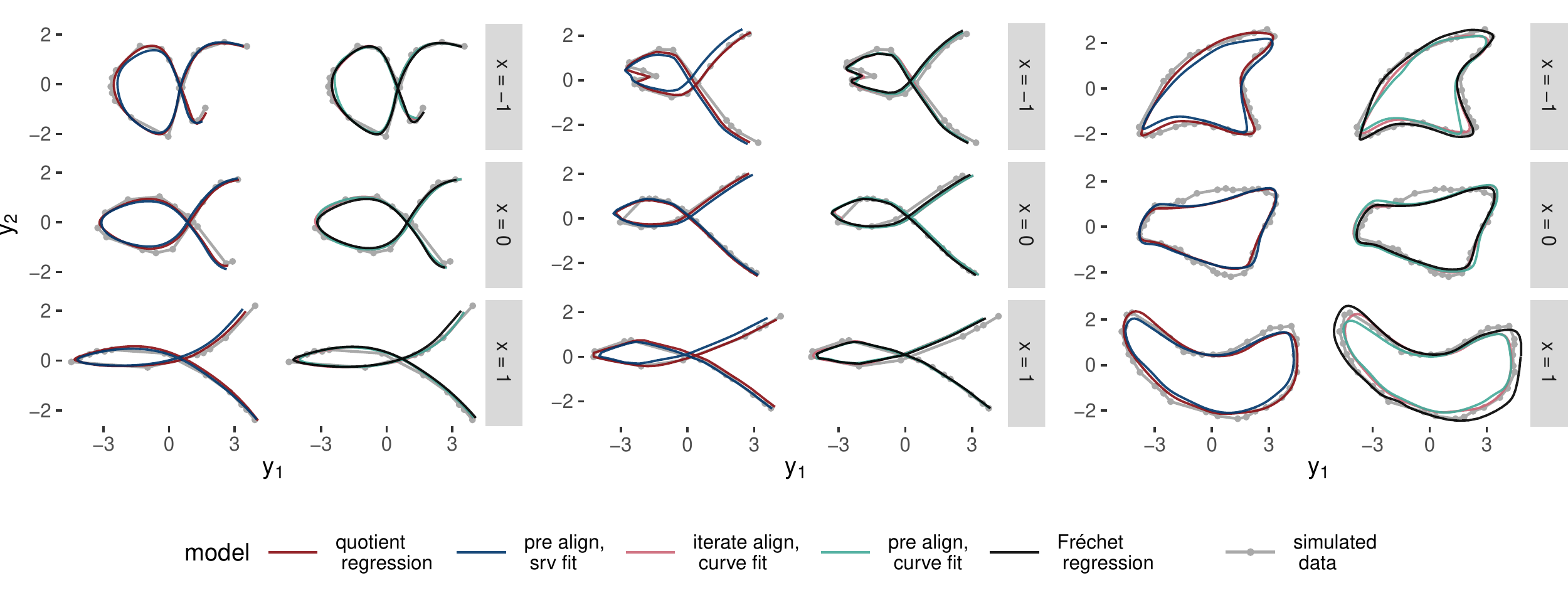}
    \caption{\label{fig:simulated_models} Predictions for $x = -1, 0, 1$ for one typical selected run of simulation 1 (left), simulation 5 (middle) and simulation 11 (right). The predictions for all $x$ values of the same runs can be found in the appendix (Fig. \ref{fig:simulation_1}, Fig. \ref{fig:simulation_5} and Fig. \ref{fig:simulation_11}, respectively).}
\end{figure}

For the last simulation scenario (simulations 9 to 12) we not only choose the model to be linear on curve instead of on SRV level and use closed quadratic spline curves here to generate smooth, closed contours, we also add the random walk error with standard deviation $sd$ directly to the $\kappa_i$ selected points, and not to the observed SRV vectors. This leads to observed curves that suit a curve-level functional model better than an SRV-level model, both in terms of their relationship with the covariate and in terms of error structure. We choose this setting, which neither fits well with the quotient linear model nor with Fréchet regression, to demonstrate the robustness of our method and to validate the adapted algorithm for closed curves (Algorithm \ref{algo:quotient_regression_closed} in the Appendix). For the quotient linear model and the pre-align, SRV fit procedure we use closed linear splines with 21 knots on SRV-level and the procedure for closing the splines described in Subsection \ref{subsec:closed_curves}. For the procedures with fit on curve level we use quadratic closed splines with 21 knots and for the Fréchet regression model we use linear SRV splines with 21 knots and the algorithm for estimating closed mean curves of \citet{steyer} adapted for weighted mean estimation (Algorithm \ref{algo:frechet_reg} in the Appendix).
See Fig.~\ref{fig:simulated_models}, right, and Fig. \ref{fig:simulation_11} in the appendix for an example of simulation 11. Even in this unfavorable setting the quotient linear model performs best in three out of four simulations. Only in the case of $\kappa_i\in [30, 40]$ points per curve and $sd = 0.2$, the procedure where we iterate between alignment and curve-level fit performs slightly better in terms of the MSE (Tab. \ref{tab:model_performance}). This can be explained by the fact that in this case the points are observed relatively densely and therefore errors at the curve level cause large errors at the SRV level (since we calculate the derivative via finite differences).

Visually, the quotient regression model and iterating between alignment and curve-level fit gives satisfying results, while if we fit a model after pre-alignment, the predicted curves for $x = -1$ and $x = 1$ appear too small (see Fig. \ref{fig:simulated_models}, right). This can again be explained by the fact that alignment to the mean does not automatically result in good alignment among the curves. This is similarly problematic for Fréchet regression. Here, the prediction for $x = 1$ appears too large and the prediction for $x = -1$ is a bit too bulky on the left. 

Overall, in the 12 simulations, the quotient regression model performed best in terms of the MSE among the five estimation methods considered, followed by Fréchet regression and the alternation between curve level fitting and alignment. Also, the average time required for one computation is relatively small for the quotient regression model compared to the other methods, especially for the more complex simulations 5 to 12 (Tab. \ref{tab:model_performance}), while it naturally takes somewhat longer than methods with pre-alignment only in most scenarios. 
The increased run times for methods with pre-alignment in scenarios with closed curves (simulation 9-12) stem from the fact that they involve unconditional elastic mean computation explicitly optimizing for closed mean curves \citep{steyer} whereas quotient regression utilizes the simplified approach described in Section \ref{subsec:closed_curves}.
In part, this also explains the long run times of Fréchet regression in these scenarios, building on an adapted version of this unconditional elastic mean computation.
Fréchet regression, however, also generally takes longer than the other methods, since optimization must be performed for each value of $\boldsymbol{x}$ separately; here for our small $n=11$ and single covariate we used all observed covariates $x = -1, -0.8, \dots, 0.8, 1$. This means that for this model the computation time increases not only with the number $n$ of observed curves but also with the number of predictions for covariate combinations desired.

\subsection{Inference based on spline coefficients}
Another advantage of the quotient regression model over Fréchet regression is that it yields parameter estimates for each covariate. These are useful not only for interpretation but also for model inference. In this simulation, we investigate and validate the bootstrap based tests described in Section \ref{sec:inference} for the slope parameters of the quotient regression model. In particular, we focus on the more difficult case of the test based on the associated spline coefficients, which additionally allows to investigate local properties of the slope parameters.

For this purpose, we generate SRV curves as linear splines with 6 equidistant knots as a function of two covariates. To see how the test behaves with stronger and weaker effects, a strong effect $\tilde{\boldsymbol{\beta}}_1$ is used for the association with $x_1$ and a weaker, local effect $\tilde{\boldsymbol{\beta}}_2$ is assumed for the relation with $x_2$, with $\tilde{\boldsymbol{\beta}}_2(t) = 0$ for $t >= 0.4$ (cf.\ Fig. \ref{fig:test_beta_coefs.pdf}). In addition, we assume that there is a third covariate $x_3$, which is independent of the observed curves. 

For the simulation, we first draw samples of sample size $n \in \{10, 30, 60\}$ of the covariates $x_{ji} \sim \text{Unif}(-1,1)$ for $j = 1,2,3$ and $i = 1, \dots, n$. Then, similar as in the previous subsection, for $i = 1, \dots, n$ we randomly select 10 to 15 points on the curve $Q^{-1}(\tilde{\boldsymbol{\beta}}_0 + \tilde{\boldsymbol{\beta}}_1 x_{1i} + \tilde{\boldsymbol{\beta}}_2 x_{2i})$. Note, that this procedure will not generate observations from the quotient regression model with parameters $\tilde{\boldsymbol{\beta}}_0$,  $\tilde{\boldsymbol{\beta}}_1$ and $\tilde{\boldsymbol{\beta}}_2$, as the sampling of the points on the curve generates a not further defined error in the quotient space. Since the quotient regression model is defined only as a minimization problem and there is no generating probability distribution available to sample curves from this model for given model parameters but we have to rely on the described auxiliary sampling scheme. In general, this implies that the model will be misspecified, i.e. $\mathcal{E}([Y]|x_1, x_2, x_3) \neq [Q^{-1}(\tilde{\boldsymbol{\beta}}_0 + \tilde{\boldsymbol{\beta}}_1 x_{1} + \tilde{\boldsymbol{\beta}}_2 x_{2})]$ for the quotient regression model and the data generated as described above. As a consequence, if we estimate $\displaystyle (\hat{\boldsymbol{\beta}_0}, \hat{\boldsymbol{\beta}}_1, \hat{\boldsymbol{\beta}}_2, \hat{\boldsymbol{\beta}}_3) =\argmin_{(\boldsymbol{\beta}_0, \boldsymbol{\beta}_1, \boldsymbol{\beta}_2, \boldsymbol{\beta}_3) \in \mathbb{L}_2} \sum_{i = 1}^n  d(\boldsymbol{y}_i, \boldsymbol{\beta}_0 + \boldsymbol{\beta}_1 x_{1i} + \boldsymbol{\beta}_2 x_{2i} + \boldsymbol{\beta}_3 x_{3i})^2$ with respect to the elastic distance, the parameters $\hat{\boldsymbol{\beta}_0}$, $\hat{\boldsymbol{\beta}}_1$ and $\hat{\boldsymbol{\beta}}_2$ will be different to $\tilde{\boldsymbol{\beta}}_0$,  $\tilde{\boldsymbol{\beta}}_1$ and $\tilde{\boldsymbol{\beta}}_2$ even if $n \to \infty$. However, $\hat{\boldsymbol{\beta}}_3 \to \boldsymbol{0}$ if $n \to \infty$ holds, since  $x_3$ and the curves are assumed to be independent. For the test of the slope parameters, this means that rejections of $H_{01}: \boldsymbol{\beta}_1 = 0$ and $H_{02}: \boldsymbol{\beta}_2 = 0$ correspond to the test's power, while those of $H_{03}: \boldsymbol{\beta}_3 = \boldsymbol{0}$ should keep the type one error rate here specified as $\alpha=0.05$. Looking at the tests for the individual spline coefficients $\boldsymbol{\xi}_{2,m}$,  $m =1, \dots, 6$ of $\boldsymbol{\beta}_2$, we expect the null hypotheses $\boldsymbol{\xi}_{2,1} = \boldsymbol{0}$ and $\boldsymbol{\xi}_{2,2} = \boldsymbol{0}$ to be rejected, but because of the above argument, the other spline coefficients are not guaranteed to be zero.

To obtain an estimate of the rejection probability for the tests of the coefficients being zero given the sample size $n \in \{10, 30, 60\}$ and the number of bootstrap repetitions $N_{boot} \in \{100, 500, 1000\}$, we draw 1000 times a sample consisting of curves $\boldsymbol{y}_1, \dots, \boldsymbol{y}_n$ with covariates $x_{1i}, x_{2i}, x_{3i}$, $i = 1, \dots, n$ as described above.  Next, we draw bootstrap replicates $\boldsymbol{y}_1^{(b)}, \dots, \boldsymbol{y}_n^{(b)}$, $b=1, \dots, N_{boot}$, from the sample and reject the null hypothesis $H_{0j}: \boldsymbol{\beta}_j = \boldsymbol{0}$ if $\bar{\boldsymbol{\xi}}_{j,m}^T \hat{\Sigma}^{-1}_{j,m} \bar{\boldsymbol{\xi}}_{j,m} \geq c_{j, 1 - \frac{\alpha}{M}}$, where $c_{j, 1 - \frac{\alpha}{M}}$ is the $1 - \frac{\alpha}{M}$ percentile, for any of the spline coefficients $\boldsymbol{\xi}_{j,m}$,  $m = 1, \dots, M=6$ of $\boldsymbol{\beta}_j$
(as described in more detail in Section \ref{sec:inference}). The estimated rejection probability (Tab. \ref{tab:bootstrap_test})  then is the relative proportion of the 1000 repetitions in which the null hypothesis is rejected.

\begin{table}[ht]
\centering
\begin{tabular}{|rr|rrr|}
  \hline
$n$ & $N_{boot}$ & $H_{01}: \boldsymbol{\beta}_1 = \boldsymbol{0}$ & $H_{02}: \boldsymbol{\beta}_2 = \boldsymbol{0}$ & $H_{03}: \boldsymbol{\beta}_3 = \boldsymbol{0}$ \\ 
  \hline
  10 & 100 & 0.66 & 0.07 & 0.02 \\ 
  10 & 500 & 0.42 & 0.01 & 0.00 \\ 
  10 & 1000 & 0.38 & 0.01 & 0.00 \\ 
  30 & 100 & 1.00 & 0.76 & 0.12 \\ 
  30 & 500 & 1.00 & 0.67 & 0.06 \\ 
  30 & 1000 & 1.00 & 0.67 & 0.05 \\ 
  60 & 100 & 1.00 & 0.97 & 0.10 \\ 
  60 & 500 & 1.00 & 0.96 & 0.05 \\ 
  60 & 1000 & 1.00 & 0.96 & 0.04 \\ 
   \hline
\end{tabular}
\caption{Estimated rejection probability of the null hypothesis $H_{0j}: \boldsymbol{\beta}_j = 0$, $j=1, 2, 3$, for a sample of size $n$ and $B$ Bootstrap replications.}
\label{tab:bootstrap_test}
\end{table}

For the data constellation described above, table \ref{tab:bootstrap_test} indicates that the rejection  probability of $H_{03}$ keeps the $\alpha$ level of  5\% if the number of  bootstrap replications is sufficiently large, i.e.\ at least about 500-1000.  In this setup, the weak effect $\boldsymbol{\beta}_2$ is found to be significant in 67\% and 97\% of the cases and the strong effect $\boldsymbol{\beta}_1$ even in 100\% of the cases for $n=30$ and $60$, respectively. To see if the distinction of zero and non-zero effects is also possible for parts of the curves, we consider in Fig. \ref{fig:test_beta_coefs.pdf} (right) the rejection probabilities for the tests of the individual spline coefficients. 

\begin{figure}[ht]
    \centering
    \includegraphics[scale=0.7]{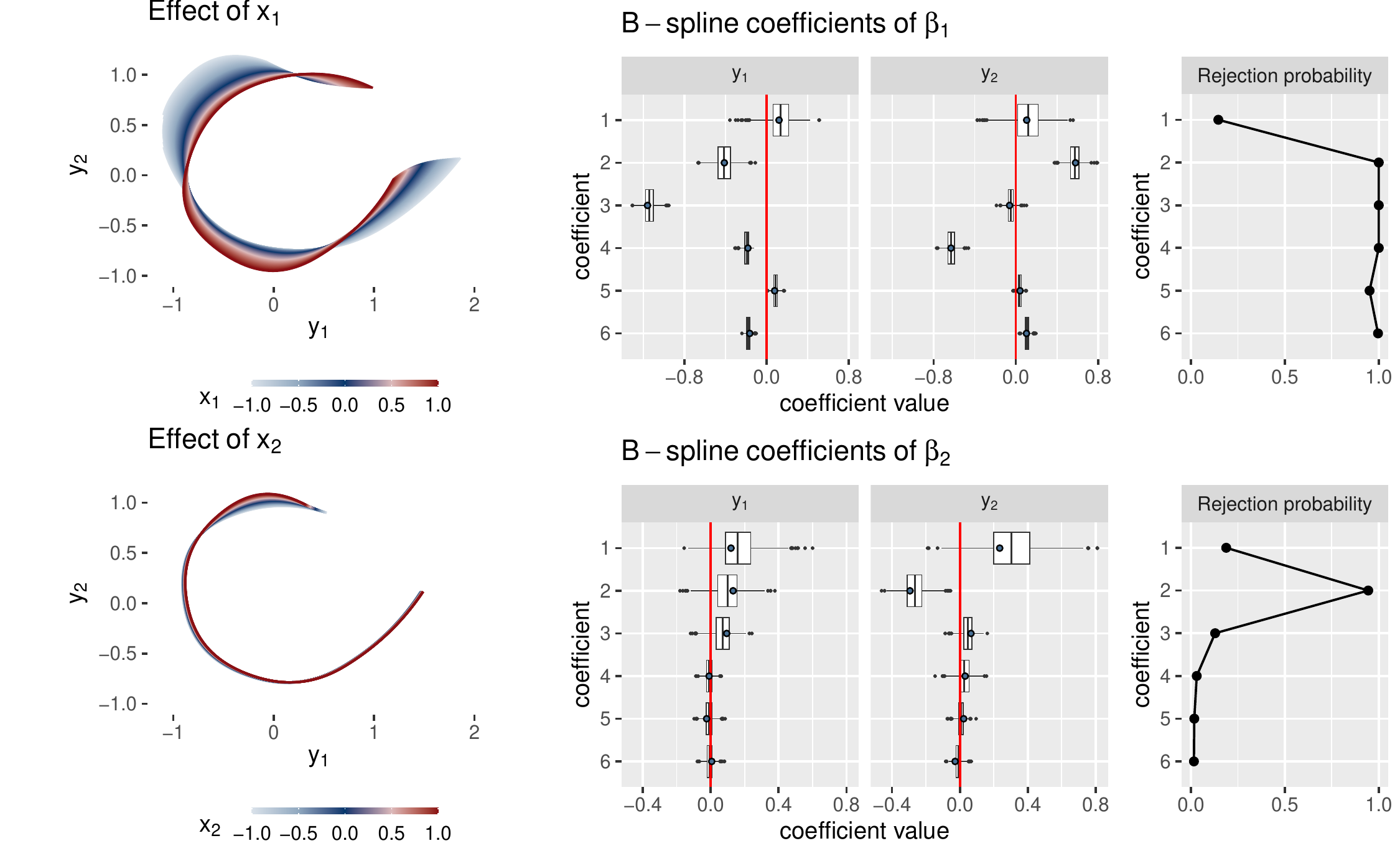}
    \caption{\label{fig:test_beta_coefs.pdf} Left: Effect of $x_1$ given $x_2 = 0$ (top) and of $x_2$ given $x_1 = 0$ (bottom) estimated on a large simulated sample (n = 500) as an approximation of the closest model to the true model in our model space. Middle: Distribution of  bootstrapped coefficients ${\boldsymbol{\xi}}_{j,m}^{(b)}$, $m = 1, \dots, 6$, $j = 1$ (top) and $j = 2$ (bottom) over the 1000 repetitions for the setting with $n = 60$ and $N_{boot} = 1000$; the coefficients of the effect estimated on the large sample are displayed as blue dots. Coefficients ${\boldsymbol{\xi}}_{j,1}, \dots, {\boldsymbol{\xi}}_{j,6}$  correspond to regions on the curves from the top anti-clockwise to the right. Right: Rejection probabilities for the tests of the individual spline coefficients.}
\end{figure}

The plot indicates that the null hypotheses for coefficients of $\beta_1$ are mostly rejected (rejection probabilities 0.15, 1.00, 1.00, 1.00, 0.95, and 1.00 for coefficients 1-6, respectively) while for $\beta_2$ the coefficients changing the lower part of the curves, $\boldsymbol{\xi}_{4}, \boldsymbol{\xi}_{5}$ and $\boldsymbol{\xi}_{6}$, are mostly not rejected (rejection probabilities 0.19,  0.94,  0.13,  0.03,  0.02 and 0.02 for coefficients 1-6, respectively), which is consistent with the absent visible effect (Fig. \ref{fig:test_beta_coefs.pdf} , left). The distribution of the bootstrapped coefficients ${\boldsymbol{\xi}}_{j,m}^{(b)}$ (Fig.\ref{fig:test_beta_coefs.pdf} , middle) essentially scatters around the estimated optimal parameters for large sample size (n = 500), which we use as an approximation of the closest model to the true model in our model space. This indicates good identifiability of the regression model via its spline coefficients.

We further repeat the simulation for n = 60 and B = 1000 with 11 instead of 6 knots still using linear SRV splines for data generation and for modeling to check that the tests are also valid for more complex curves with more spline coefficients. Here we observe a rejection probability of 100\% for $H_{01}$ and $H_{02}$ and of 6\% for $H_{03}$.
The high rejection probability also for the smaller effect $\boldsymbol{\beta}_2$ probably results from the fact that we did not succeed in simulating a local effect (see Fig. \ref{fig:test_beta2_more_coefs} in the appendix). However, the coefficients appear to be well identified here as well. To keep the significance level of $\alpha = 0.05$ for $H_{03}$ exactly, more observations would be necessary for this larger number of spline coefficients.

Overall, we conclude from this simulation that it is possible to test the significance of the slope parameters for the quotient space regression model using the corresponding spline coefficients. As this requires that the model is well identified by the spline coefficients, the sample size should be large enough to ensure that the spline coefficients of the different bootstrap model estimates represent equivalent parts of the curves. In particular, for a larger number of spline coefficients allowing flexible modeling of curves, a relatively large number of observations is needed to a) maintain the $\alpha$ level in the bootstrap setting and b) obtain reasonable power.
We also discuss the choice of test based or not based on spline coefficients in the context of the application in the next section. %, we will see an example of a real data application and discuss, among other things, whether a test based on the spline coefficients is the best choice in this case.

\section{Investigating the effect of age and Alzheimer's disease on hippocampus outlines via elastic regression}
\label{sec:application}
Hippocampal volume loss is associated with both Alzheimer's disease and normal aging \citep{henneman}. Moreover, \cite{frisoni} showed that these covariates affect the hippocampal surface locally using parametric surface mesh models. The surface mesh model, however, depends on a meaningful parametrization of the shape. In contrast, we investigate local effects on the hippocampal volume by modeling the shape of the hippocampus. However, it's essential to note that when we refer to 'shape' we do not mean the classical shape spaces that consider point configurations modulo rotation and scaling. Instead, our approach focuses on curves modulo reparametrization and translation, using the two-dimensional outlines in a quotient linear model that defines an elastic model of the outlines modulo re-parametrization without dependence on any chosen parametrization.

\subsection{Data acquisition and preparation}
Data used in the preparation of this article were obtained from the Alzheimer’s Disease Neuroimaging Initiative (ADNI) database (\url{adni.loni.usc.edu}). The ADNI was launched in 2003 as a public-private partnership, led by Principal Investigator Michael W. Weiner, MD. The primary goal of ADNI has been to test whether serial magnetic resonance imaging (MRI), positron emission tomography (PET), other biological markers, and clinical and neuropsychological assessment can be combined to measure the progression of mild cognitive impairment (MCI) and early Alzheimer’s disease (AD).
In addition to the MRI images, ADNI also provides semi-automated segmentations of the hippocampus created using a high-dimensional brain mapping program SNT, which was commercially available from Medtronic Surgical Navigation Technologies (Louisville, CO). For more details on this procedure and a comparison with manual segmentation of the hippocampal volumes, see \cite{hsu}. For our analysis, we use all available hippocampal masks of the 101 Alzheimer's disease (AD) patients and 138 controls (CN) obtained from the MRI images of the first scanning session. 
To apply our quotient regression model to the hippocampus data, we need to extract two-dimensional outlines (Fig. \ref{fig:hippocampus_data}, right) from the three-dimensional hippocampal masks (Fig. \ref{fig:hippocampus_data}, left). To do this, we perform the following steps for the left and right hippocampus separately. First, each hippocampus is rotated around the left-right axis using principal component analysis so that its first principal component lies in the horizontal plane. Then we project the data onto the horizontal plane and use the function \texttt{ocontour} from the \texttt{R}-package ``EBImage'' \citep{EBImage} to extract a closed outline curve. After alignment to the overall mean, the outlines of the hippocampus are  sliced at the tail in the same location to obtain meaningful open curves, since the hippocampus merges into the fornix at the tail, i.e. it is not anatomically closed. In the last step the number of points per curve is reduced to improve the computational efficiency of model estimation, via keeping only points whose time stamps are at least 0.015 apart after alignment to the overall mean.
\begin{figure}[ht]
    \centering
    \includegraphics[scale=0.65]{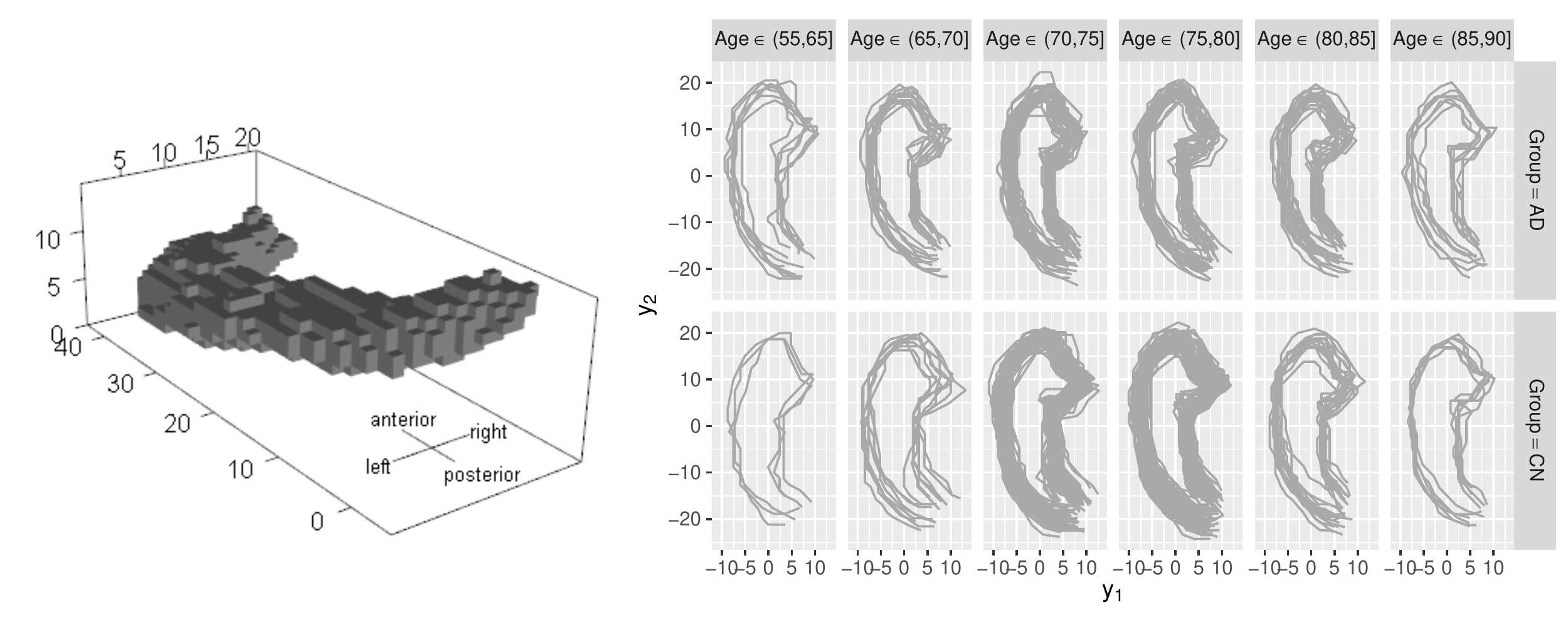}
    \caption{\label{fig:hippocampus_data}
    Left: Three-dimensional left hippocampal mask for one person. Right: Open outlines of the left hippocampus for different age groups separately for Alzheimer's disease (AD) patients and control group (CN).}
\end{figure}

The result of this preprocessing are left and right hippocampal outlines of 101 Alzheimer's disease (AD) patients and 138 control subjects (CN) with 33 to 48 points per curve. The explanatory variables considered are Age, Group = AD, CN and Sex = M, F of the subjects. These covariates are roughly balanced, the mean age for AD and CN is 76 years, ranging from 57 to 89 and 62 to 90 years, respectively, and about 49\% of the subjects in both groups are female. Visual inspection of the outlines, taking into account the covariates Age and Group (Fig. \ref{fig:hippocampus_data}, right), reveals no clear relationship with the shape of the hippocampus. This might be due to the large overall variance of the outlines.

\subsection{Regression analysis of hippocampal shapes}
To see how age, Alzheimer's disease and the sex of a subject influence the shape of the hippocampus, we model the hippocampus outlines using the quotient linear model for elastic curves (Def. \ref{defi:quotient_reg_elastic}). Precisely, we assume

\begin{align*}
    \mathcal{E}([\boldsymbol{Y}] | x_{\text{Age}}, x_{\text{Group}}, x_{\text{Sex}}) = [Q^{-1}\left( \boldsymbol{\beta}_0 + \boldsymbol{\beta}_{\text{Age}} x_{\text{Age}} +  \boldsymbol{\beta}_{\text{Group}} x_{\text{Group}} + \boldsymbol{\beta}_{\text{Sex}} x_{\text{Sex}} \right)],
\end{align*}

where the conditional Fr\'echet mean $\mathcal{E}([\boldsymbol{Y}] | x_{\text{Age}}, x_{\text{Group}}, x_{\text{Sex}})$ of the hippocampal outlines is defined with respect to the elastic distance on the product space of elastic curves for the left and right hippocampus. That is $\dist(\boldsymbol{y}_1, \boldsymbol{y}_2) = \sqrt{\dist_{\text{left}}(\boldsymbol{y}_{1, \text{left}}, \boldsymbol{y}_{2, \text{left}})^2 + \dist_{\text{right}}(\boldsymbol{y}_{1, \text{right}}, \boldsymbol{y}_{2, \text{right}})^2}$ with $\boldsymbol{y}_i = (\boldsymbol{y}_{i, \text{left}}, \boldsymbol{y}_{i, \text{right}}), i = 1,2$, where $\dist_{\text{left}}$ and $\dist_{\text{right}}$ are the separate elastic distances for the left and the right hippocampal curves, respectively. With this product space distance, the optimization problem defining the metric regression model  $\argmin_{f \in \mathcal{F}} \sum_{i = 1}^n  \dist(\boldsymbol{y}_i, f(\boldsymbol{x}_i))^2$ becomes $\argmin_{f \in \mathcal{F}} \sum_{i = 1}^n  \dist_{\text{left}}(\boldsymbol{y}_{\text{left}}, f_{\text{left}}(\boldsymbol{x}_i))^2 + \sum_{i = 1}^n \dist_{\text{right}}(\boldsymbol{y}_{\text{right}}, f_{\text{right}}(\boldsymbol{x}_i))^2 $ with $f = (f_{\text{left}}, f_{\text{right}})$ and therefore can be solved separately for the left and right hippocampal shapes.

The parameters $\boldsymbol{\beta}_0,\boldsymbol{\beta}_{\text{Age}}, \boldsymbol{\beta}_{\text{Group}}, \boldsymbol{\beta}_{\text{Sex}} \in \mathbb{L}_2$ of this intrinsic metric regression model are estimated using linear spline functions with 21 equidistant knots for the left and the right hippocampus each. Since this leads to piecewise linear predictions on SRV level, the predicted outlines are smooth curves (Fig. \ref{fig:hippocampus_effects}).  %However, because we model the SRV transforms in $\mathbb{L}_2$ and not the observed curves, the estimated effects are not straightforward to interpret directly. Therefore, we visualize them via model predictions in which one covariate is varied at a time.
Linear effects on SRV level are visualized on curve level by varying one covariate at a time to illustrate effect directions via corresponding predictions. 

\begin{figure}[ht]
    \centering
    \includegraphics[scale=0.7]{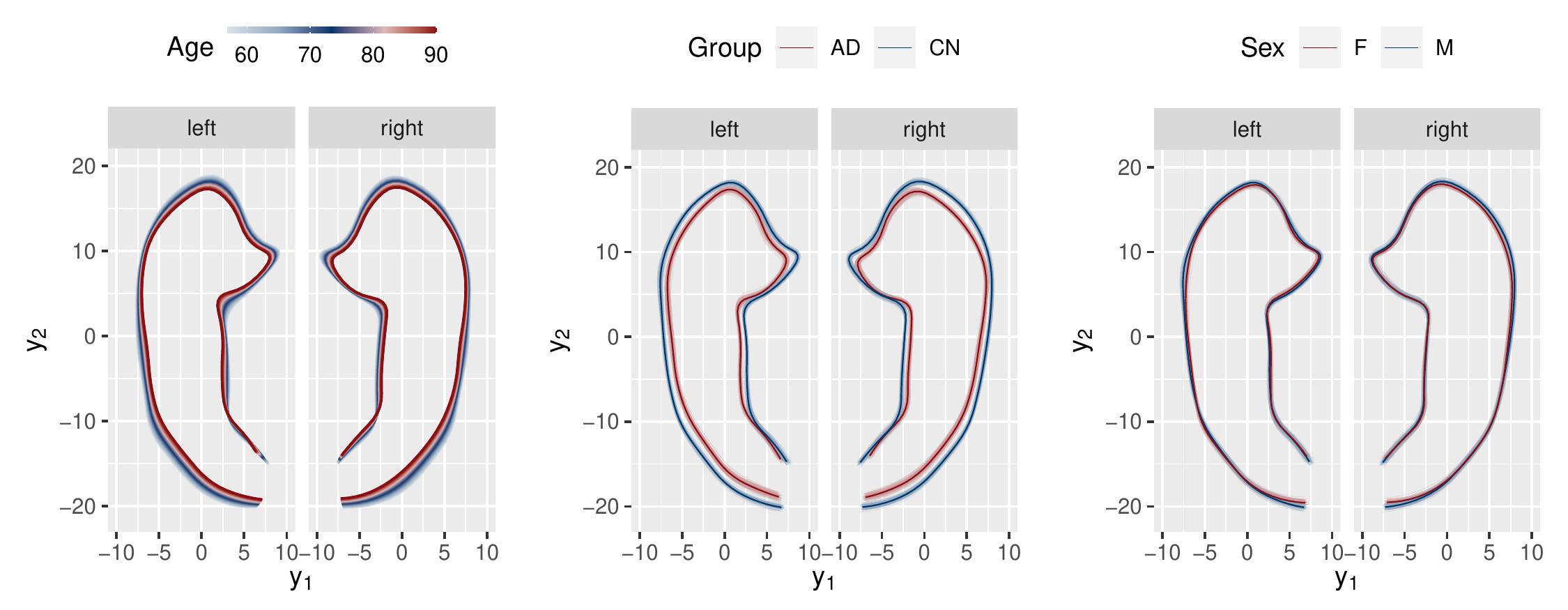}
    \caption{\label{fig:hippocampus_effects} 
    Effects displayed via model predictions with one varying covariate at a time. As a common reference, the remaining covariates are set to Age = mean(Age) = 76, Group = CN, and Sex = M. For the binary covariates Group and Sex, bootstrap predictions (lighter color) are added to the model predictions.}
\end{figure}

As expected, we observe similar effects for left and right hippocampus and the hippocampal volume decreases with age and for AD patients. Moreover, the Sex effect appears to be small compared to Age and Group effect. Since age and Alzheimer's disease appear at first glance to have a comparable effect on the hippocampus, the question arises to what extent the covariates Age and Group affect the shape of the hippocampus differently. To answer this, we use the linear structure underlying the quotient regression model and project $\hat{\boldsymbol{\beta}}_{\text{Group}}$ onto $\hat{\boldsymbol{\beta}}_{\text{Age}}$. The scalar projection of $\hat{\boldsymbol{\beta}}_{\text{Group}}$ onto $\hat{\boldsymbol{\beta}}_{\text{Age}}$ is 12.8 years, which means that having Alzheimer's diseases shrinks the hippocampus about as much as 12.8 years of aging would do, but the angle between $\hat{\boldsymbol{\beta}}_{\text{Group}}$ and $\hat{\boldsymbol{\beta}}_{\text{Age}}$ is 47 degree, which means only about half of the Group effect shows in the same direction as the Age effect does. To visualize the remaining effect, we plot the prediction for a subject with Alzheimer's disease alongside the prediction for a model where the Group effect is replaced by its linear projection on the Age effect (Fig. \ref{fig:hippocampus_effects_supplement}, left). This allows us to see which parts of the hippocampus are effected differently.  While both age and Alzheimer's disease reduce the volume at the \ovalbox{1} label in Fig. \ref{fig:hippocampus_effects_supplement} (left), the width of the hippocampal head, i.e., the distance between \ovalbox{2} and \ovalbox{4}, appears to be reduced substantially more by AD than by normal aging. In contrast, the distance between \ovalbox{1} and \ovalbox{3} appears to be similarly affected by both covariates, although for the right hippocampus this distance might be smaller for AD patients compared to someone in the control group who is 12.8 years older.

\begin{figure}[ht]
    \centering
    \includegraphics[scale=0.65]{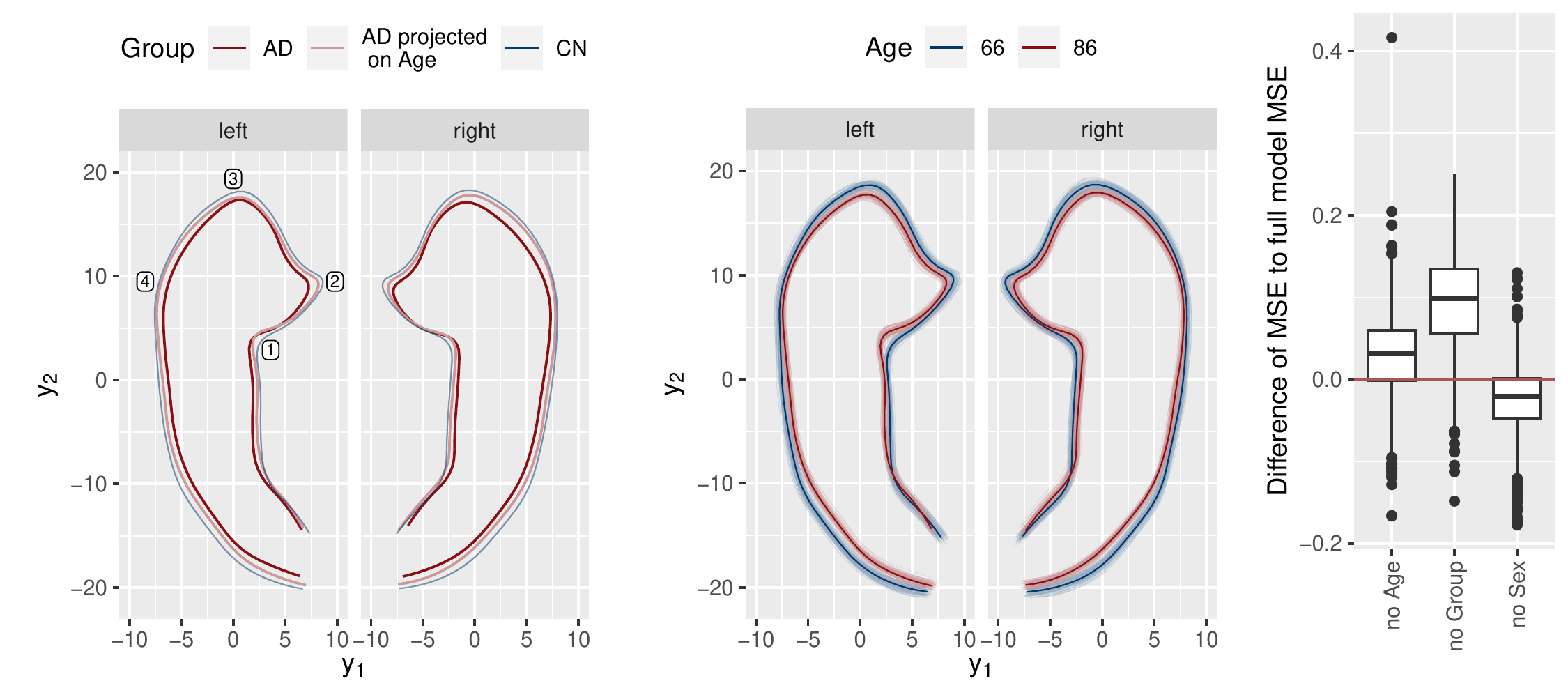}
    \caption{\label{fig:hippocampus_effects_supplement} All non mentioned covariates are fixed to Age = mean(Age) = 76, Group = CN, and Sex = M. Left: Prediction for AD compared to the prediction for a model where the Group effect is replaced by its linear projection onto the Age effect. Labels indicate prominent features. Middle: Bootstrap predictions for Age effect. Right: Difference of the MSE of the models with an omitted covariate and the MSE of the full model computed on the out-of-bootstrap sample.}
\end{figure}

\subsection{Model inference for the hippocampus regression model}
To assess the significance of the estimated effects, we apply the model inference and selection tools described in Section \ref{sec:inference} to this regression model for the hippocampal shapes. First, we test the global null hypothesis $H_0$ that none of the covariates Age, Group, and Sex has an effect by using the Fréchet coefficient of determination $\tilde{R}^2$ as a test statistic and approximating its distribution under $H_0$ through permutation sampling (number of samples = 500). We observe $\tilde{R}^2 = 0.033$, which is significantly different from the estimated mean under $H_0$ ($\overline{\tilde{R}^2} = 0.016$, p-value of the one-sided t-test $< 2.2\times 10^{-16}$). Next we compare the adjusted coefficients of determination $\tilde{R}^2_{\text{adj}}$ for all possible sub-models to decide which covariates improve the model fit. Here, the largest value is obtained for the full model ($\tilde{R}^2_{\text{adj}} = 0.021$, see Tab. \ref{tab:adj_R_squared}), while all reported values for $\tilde{R}^2$ and $\tilde{R}^2_{\text{adj}}$ are relatively small. Thus, although %the variables considered can be deemed relevant based on this criterion, they 
explaining only a small proportion of the total variation, with a larger remaining fraction corresponding to %. That is, most of the variation in the curves cannot be explained by the variables considered but is 
individual variation among individuals, all variables considered can be deemed relevant based on this criterion.

To account for the variability of the model predictions, we draw 1000 bootstrap samples and estimate the model parameters on each. The distribution of the bootstrapped spline coefficients is shown in Fig. \ref{fig:bootstrap_coefs_left} and \ref{fig:bootstrap_coefs_right} in the appendix. The large variation in some of the bootstrap coefficients, much larger than for the bootstrap predictions in Fig. \ref{fig:hippocampus_effects} and Fig. \ref{fig:hippocampus_effects_supplement}, middle, indicates that a test based on the coefficients might lose power due to warping variability, as discussed for a setting with many spline coefficients (here 42 2-dimensional coefficients per covariate for left and right in total) in Section \ref{subsec:curveCI_vs_coefCI}. %In fact, if we construct the test based on the 42 spline coefficients per effect described in Section \ref{sec:inference}, only the Group effect is found to be significant. 
Therefore, we construct confidence regions directly for the predictions. To this end, for each prediction, we compute the bootstrap predictions for the same combinations of covariates as well as their distance from the original model prediction. We then construct the simultaneous 95\% confidence region based on the closest 950 bootstrap predictions. The result for the binary effects Group and Sex is shown in Fig. \ref{fig:hippocampus_effects}. For the continuous covariate Age, we show the bootstrap predictions with confidence regions for Age = 66 and Age = 86 in Fig. \ref{fig:hippocampus_effects_supplement}, middle. Note that the resulting confidence regions not only include the variance of the effects, but also that of the intercept. It can be seen that the confidence regions constructed in this way clearly separate the Alzheimer's group from the control group, whereas the confidence regions for Sex overlap in all parts of the hippocampus. For Age, there are parts (especially at the head of the hippocampus) where the regions are separated, as well as parts where no clear separation is found. This is consistent with the previously obtained results that the group effect is the most pronounced and the Sex effect is the least evident.

In addition, we evaluate the importance of the included covariates by comparing how much the model estimation error increases when we remove single covariates from the model. Here we look at the difference of the mean squared error (MSE) of the model with an omitted covariate and the MSE for the full model (Fig. \ref{fig:hippocampus_effects_supplement}, right), where we estimate the MSE for each bootstrap model  on the out-of-bootstrap sample. That is, we evaluate the model prediction on the data that were not used for model estimation on the bootstrap sample (about 36.8\% of the data). On average, we compute an out-of-bootstrap MSE of 13.34 for the full model and 13.37, 13.43 or 13.31 for a model without Age, Group or Sex effect, respectively. Omitting one effect increases the MSE for 740, 943 and 263 out of 1000 bootstrap samples for the Age, Group and Sex effect, respectively. Thus, based on these two variables improving the out-of-bootstrap prediction error in most of the samples, we would choose the model including Age and Group but no Sex effect.

\subsection{Discussion of the hippocampus application}
Overall, Alzheimer's disease has the largest and most stable effect on the hippocampus among the covariates considered. The direction of this effect, i.e., the way the hippocampus shrinks, differs from normal aging. Although females appear to have slightly smaller hippocampi, this Sex effect is not clearly significant in our model.
In further studies, it would be interesting to include in the analysis the mild cognitive impairment (MCI) group, for whom hippocampal masks are also available from ADNI. Since this group is known to be heterogeneous, with likely unknown subgroups, we did not include them in this study. In addition, it may be worthwhile to explore more complex model equations, for example, including further covariates or interaction effects (e.g. of Age and Sex) and to examine if and how the effects differ between the left and right hippocampus. 

\section{Conclusion and Outlook} \label{sec:discussion}
In this work, we developed an elastic regression approach, which we motivate as a special case of a general type of regression models we call quotient linear regression, for curves with respect to the elastic distance as response values and multiple scalar covariates. It allows modeling curves in two or more dimensions, e.g. outlines of anatomical features, while being invariant to their parameterization. Using the projections of affine linear functions as the model space for the SRV transformed curves permits iterative estimation of the model by alternating between alignment and estimation of a functional linear model. We provide an implementation of this algorithm in the \texttt{R}-package \texttt{elasdics} \citep{elasdics}.

To deal with sparsely and/or irregularly observed curves, we use splines to model the SRV curves. Since certain of these splines are identifiable modulo parameterization, inference based on the estimated spline coefficients is possible in these cases and also allows to investigate local effects. Here, however, as with our proposed inference methods based on distances and/or on the predicted curves, we rely on re-sampling methods such as bootstrap and permutation re-sampling. Further research will be needed to develop tests and confidence sets with formal guarantees.

Placing the proposed elastic regression model into the more general context of quotient (linear) regression, allows us to point out direct connections to similar approaches on other quotient spaces in literature and to present results on properties of the model space, consistency and existence of Fréchet mean estimation in a higher level of generality. Moreover, we also pave the way to for quotient regression beyond linear model spaces:

Using affine linear functions as underlying model space  includes constant speed geodesics in the quotient space of curves modulo re-parameterization,
but is somewhat more flexible. Using this larger space not only enables the estimation strategy described above, but we have also shown through examples (in the simulations in Section \ref{sec:simulation}) that geodesic regression lines alone are not sufficient to model all  changes to curves that naturally appear in practice, and a larger model space thus is beneficial. However, our proposed model space is still linear in the SRV space, which may be too restrictive for some real data applications. To allow more flexible smooth dependence of curves on  covariates, quotient models could be extended to an additive linear regression model. %using tensor product splines on SRV level in future work.

Another appealing direction for further research is to develop the quotient regression model for elastic shapes, i.e.\ curves modulo translation, scaling, rotation and parametrization, which goes beyond quotient spaces of a Hilbert space. Along the lines of Section \ref{quotient:riemannian}, this space can be seen as the quotient of the sphere, which is a submanifold of $\mathbb{L}_2$, on which the product of the rotation and re-parametrization groups acts by isometries. Such a model could be implemented building on 'generalized linear' regression models on the sphere as suggested in \cite{stoecker} for planar (in-elastic) shapes and forms. %In this work, additive models for the shape space of curves modulo translation, scaling and rotation, but not parametrization, are already discussed, and it seems promising to combine these with the parametrization invariant model discussed here.

\section*{Acknowledgements}
We gratefully acknowledge funding by grants GR 3793/3-1 'Flexible regression methods for curve and shape data'  and GR 3793/5-1 'Combining geometry-aware statistical and deep learning for neuroimaging data' (project P7 in the research unit FOR 5363) from the German research foundation (DFG). We would like to thank Kerstin Ritter (Charité Berlin) for sharing her domain knowledge on Alzheimer's disease and the ADNI data.

Data collection and sharing for this project was funded by the Alzheimer's Disease Neuroimaging Initiative (ADNI) (National Institutes of Health Grant U01 AG024904) and DOD ADNI (Department of Defense award number W81XWH-12-2-0012). ADNI is funded by the National Institute on Aging, the National Institute of Biomedical Imaging and Bioengineering, and through generous contributions from the following: AbbVie, Alzheimer’s Association; Alzheimer’s Drug Discovery Foundation; Araclon Biotech; BioClinica, Inc.; Biogen; Bristol-Myers Squibb Company; CereSpir, Inc.; Cogstate; Eisai Inc.; Elan Pharmaceuticals, Inc.; Eli Lilly and Company; EuroImmun; F. Hoffmann-La Roche Ltd and its affiliated company Genentech, Inc.; Fujirebio; GE Healthcare; IXICO Ltd.; Janssen Alzheimer Immunotherapy Research \& Development, LLC.; Johnson \& Johnson Pharmaceutical Research \& Development LLC.; Lumosity; Lundbeck; Merck \& Co., Inc.; Meso Scale Diagnostics, LLC.; NeuroRx Research; Neurotrack Technologies; Novartis Pharmaceuticals Corporation; Pfizer Inc.; Piramal Imaging; Servier; Takeda Pharmaceutical Company; and Transition Therapeutics. The Canadian Institutes of Health Research is providing funds to support ADNI clinical sites in Canada. Private sector contributions are facilitated by the Foundation for the National Institutes of Health (\url{www.fnih.org}). The grantee organization is the Northern California Institute for Research and Education, and the study is coordinated by the Alzheimer’s Therapeutic Research Institute at the University of Southern California. ADNI data are disseminated by the Laboratory for Neuro Imaging at the University of Southern California.

\newpage
\bibliographystyle{plainnat}
\bibliography{literatur}

\newpage
\appendix
\section{Proofs and Computations}
\subsection{Proof of Lemma \ref{lem:consistency}}
\label{subsec:proof_consistency}
\cite{huckemann} generalizes the notion of the Fréchet mean to the Fréchet $\rho$-mean:
\begin{defi}[Fréchet $\rho$-mean]
Let $X, X_1, X_2, \dots$ be random elements mapping from a probability space $\Omega, \mathcal{A}, \mathcal{P}$ to a topological space Q. Let $(P, d)$ be a topological space with distance $d$. For a continuous function $\rho:Q \times P \to [0, \infty]$ define the set of population Fréchet $\rho$-means of $X$ in $P$ by
\begin{align*}
    E^{(\rho)} = \argmin_{\mu \in P} \mathbb{E}(\rho(X, \mu)^2).
\end{align*}
For $\omega \in \Omega$, denote by 
\begin{align*}
    E^{(\rho)}_n(\omega) = \argmin_{\mu \in  P} \sum_{i = 1}^n \rho(X_i(\omega), \mu)^2
\end{align*}
the set of sample Fréchet $\rho$-means.
\end{defi}
With this definition, the usual Fréchet mean is a special case where $\rho$ is the distance function. Similarly as for Fréchet means, one can show that sample Fréchet $\rho$-means are consistent in the sense of \cite{ziezold}, that is $\bigcap_{n = 1}^\infty \overline{\bigcup_{k = n}^\infty E^{(\rho)}_k(\omega)} \subseteq E^{(\rho)}$ for almost all $\omega \in \Omega$. 

\begin{theorem}[\cite{huckemann}] \label{theo:huckemann}
Let $\rho:Q \times P \to [0, \infty[$ be a continuous function on the product of a topological space with a separable space with distance $(P, d)$. Then strong consistency holds in the Ziezold sense for the set of Fréchet $\rho$-means in $P$ if:
\begin{enumerate}[label=(\roman*)]
    \item $X$ has compact support, or if
    \item $\mathbb{E}(\rho(X, p)^2) < \infty$ for all $p \in P$ and $\rho$ is uniformly continuous in the second argument.
\end{enumerate}

\end{theorem}
Our statement on consistency is then a straightforward consequence using (ii), with $\mathcal{X} \times \mathcal{Y}$ taking the role of $Q$ and $(\mathcal{F}, d_\mathcal{F})$ that of $(P, d)$.
\begin{proof}[Proof of Lemma \ref{lem:consistency}]
 Define $\rho: \mathcal{X} \times \mathcal{Y} \times \mathcal{F} \to [0, \infty)$, $\rho(x,y,f) = d(y, f(x))$. This loss function $\rho$ is continuous in the first two arguments since $d$ as a metric and $f$ are continuous. Furthermore, $\rho$ is uniformly continuous in the last argument since for all $x \in \mathcal{X}$, $y \in \mathcal{Y}$ and $f, \tilde{f} \in \mathcal{F}$ it holds that $d(y, f(x)) \leq d(y, \tilde{f}(x)) + d(\tilde{f}(x), f(x))$ via the triangle inequality and therefore because of symmetry: $| d(y, f(x)) -  d(y, \tilde{f}(x))| \leq d(\tilde{f}(x), f(x)) \leq d_\mathcal{F}(f, \tilde{f})$. $C(\mathcal{X}, \mathcal{Y})$ is separable by the proof of Theorem 2.4.3 in \cite{srivastava_borel_sets} and therefore $\mathcal{F}$ is separable as a subspace of a separable metric space.
\end{proof}

\subsection{Proof of Lemma \ref{lem:existence} and Lemma \ref{lem:existence_quotient}} \label{proof:existence}
\begin{proof}
$C(\mathcal{X}, \mathcal{Y})$ is complete since $\mathcal{Y}$ is complete \citep[e.g.,][Theorem 3.45]{burkill2002second} and therefore $\mathcal{F}$ in Lemma \ref{lem:existence} and $\Phi$ in Lemma \ref{lem:existence_quotient} complete as a closed subsets. Thus $\mathcal{F}$ and $\Phi$ are compact since they are complete and totally bounded. Since $f \mapsto \mathbb{E}(\dist(Y, f(X))$ and $\varphi \mapsto \mathbb{E}(\dist([Y], [\varphi(X)])$ are continuous as a compositions of continuous functions, they attain their minimum on $\mathcal{F}$ and $\Phi$, respectively.
\end{proof}

\subsection{Proof of Lemma \ref{lem:separable_complete}}
\begin{proof}
\begin{enumerate}
    \item[i)] Since $\mathcal{Y}$ is separable, there is a countable, dense subset $\mathcal{Z} \subseteq \mathcal{Y}$. Let $[y] \in \mathcal{Y}/G$. Since $\mathcal{Z}$ is dense in $\mathcal{Y}$, there is a sequence $(z_k)_{k \in \mathbb{N}} \subset \mathcal{Z}$ such that $\lim_{k \to \infty} z_k = y$. Therefore $\lim_{k \to \infty} d_G([z_k], [y]) \leq \lim_{k \to \infty} d(z_k, y) = 0$. Hence $\{[z] | z \in \mathcal{Z}\}$ is a countable, dense subset in $\mathcal{Y}/G$ and $\mathcal{Y}/G$ thus separable.
    \item[ii)] Let $([y_k])_{k \in \mathbb{N}} \subset \mathcal{Y}/G$ be a Cauchy sequence. W.l.o.g. assume $d_G([y_k], [y_{k + 1}]) < \frac{1}{2^k}$ for all $k \in \mathbb{N}$, otherwise consider a subsequence (and $([y_k])_{k \in \mathbb{N}}$ as a Cauchy sequence will converge to the same limit if it exists). We construct a sequence $(g_k)_{k \in \mathbb{N}} \subseteq G$ such that $(g_1 \circ \dots \circ g_{k-1} \circ y_k)_{k \in \mathbb{N}}$ is a Cauchy sequence in $\mathcal{Y}$. To do so set $g_1 = \text{id}$ and for $k \geq 2$ assume we already picked $g_1, \dots, g_{k-1}$. Then choose $g_k$ such that
    \begin{align*}
        d(g_1 \circ \dots \circ g_{k-1} \circ y_k, g_1 \circ \dots \circ g_k \circ y_{k+1}) < \frac{1}{2^{k-1}}.
    \end{align*}
    This is possible since $g_1, \dots, g_{k-1}$ are isometries and therefore
    \begin{align*}
        \inf_{g \in G} d(g_1 \circ \dots \circ g_{k-1} \circ y_k, g_1 \circ \dots g_{k - 1} \circ g \circ y_{k+1}) =
        \inf_{g \in G} d(y_k, g \circ y_{k+1}) = 
        d_G([y_k], [y_{k + 1}]) < \frac{1}{2^k}.
    \end{align*}
    Thus, $(g_1 \circ \dots \circ g_{k-1} \circ y_k)_{k \in \mathbb{N}}$ is a Cauchy sequence and converges to a $y \in \mathcal{Y}$, since $\mathcal{Y}$ is complete. Hence
    \begin{align*}
       d_G([y_k], [y]) &= d_G([g_1 \circ \dots \circ g_{k-1} \circ y_k], [y])\\ &= \inf_{g \in G} d(g_1 \circ \dots \circ g_{k-1} \circ y_k, g \circ y) \leq d(g_1 \circ \dots \circ g_{k-1} \circ y_k, y) \xrightarrow{k \to \infty} 0.
    \end{align*}
    As $([y_k])_{k \in \mathbb{N}}$ has a limit in $\mathcal{Y}/G$, $\mathcal{Y}/G$ is complete. 
    \end{enumerate}
\end{proof}

\subsection{Proof of Theorem \ref{theo:existence}}
\begin{proof}
We first show that $\Psi$ is coercive, that is $\Psi(\varphi) \to \infty$ if $\|\varphi\|_\Phi \to \infty$. To do so note that 
\begin{align*}
   d_G([Y], [\varphi(X)]) = \inf_{g \in G} \|Y - g \circ \varphi(X)\| &\geq \|g \circ \varphi(X)\|_\mathcal{Y} - \|Y\|_\mathcal{Y} \\
   &\geq \| \varphi(X)\|_\mathcal{Y} - C_1 
\end{align*}
for some $C_1 \in \mathbb{R}$, where the first inequality is due to the triangle inequality and the second due to the assumption that $[Y]$ is bounded. Therefore, 
\begin{align*}
    \Psi(\varphi) = \mathbb{E}(d_G([Y], [\varphi(X)])^2) \geq \mathbb{E}((\| \varphi(X)\|_\mathcal{Y} - C_1 )^2) \geq (\mathbb{E}(\| \varphi(X)\|_\mathcal{Y}) - C_1 )^2
\end{align*}
due to Jensen's inequality since $x \mapsto (x - C_1)^2$ is convex. Note that 
$\mathbb{E}(\| \varphi(X)\|_\mathcal{Y})$ defines a norm on $\Phi$ since supp$(X) = \mathcal{X}$ and all $\varphi \in \Phi$ are continuous. Since all norms are equivalent on $\Phi$ (finite dimensional vector space) this means $\mathbb{E}(\| \varphi(X)\|_\mathcal{Y}) \to \infty$ if $\|\varphi\|_\Phi \to \infty$ and therefore $\Psi(\varphi) \to \infty$ if $\|\varphi\|_\Phi \to \infty$.

Since $\Psi$ is continuous as a composition of continuous functions and coercive it attains its minimum. This is a standard argument that we repeat here for the sake of completeness. Pick a $\varphi_0 \in \Phi$. Since $\Psi$ is coercive, there is a $C_2 \in \mathbb{R}$ such that $\Psi(\varphi) \geq \Psi(\varphi_0) + 1$ if $\|\varphi\|_\infty \geq C_2$. Hence $\inf_{\| \varphi \|_\infty \leq C_2} \Psi(\varphi) = \inf_{\varphi \in \Phi} \Psi(\varphi)$. Since $\{\varphi \in \Phi | \|\varphi\|_\infty \leq C_2\}$ is a closed and bounded subset of a finite dimensional vector space, it is compact (Heine-Borel) and therefore $\Psi$ attains its minimum on $\{\varphi \in \Phi | \|\varphi\|_\infty \leq C_2\}$, which is also a global minimizer.
\end{proof}

\subsection{Proof of Lemma \ref{lem:shortest_path_quotient}}
\begin{proof}
Since $(\mathcal{Y}, d)$ is a length metric space and $\gamma$ a shortest path, for the length $l(\gamma)$ of $\gamma$ it holds by definition that
\begin{align*}
l(\gamma) = \sup_{a = t_0 < t_1 < \dots < t_n = b} \sum_{i = 0}^{n - 1} d(\gamma(t_i), \gamma(t_{i + 1})) = d(y_1, \tilde{g} \circ y_2),
\end{align*}
where $t_0, t_1, \dots, t_n$, $n \in \mathbb{N}$ is a partition of $[a, b]$.
The length of $[\gamma]$ is bounded from above by the length of $\gamma$ as
\begin{align*}
l([\gamma]) 
&= \sup_{a = t_0 < t_1 < \dots < t_n = b} \sum_{i = 0}^{n - 1} d_G( [\gamma(t_i)], [ \gamma(t_{i + 1})]) \\
&= \sup_{a = t_0 < t_1 < \dots < t_n = b} \sum_{i = 0}^{n - 1} \inf_{g \in G} d(\gamma(t_i), g \circ \gamma(t_{i + 1})) \\
&\overset{g = \text{ id}}{\leq} \sup_{a = t_0 < t_1 < \dots < t_n = b} \sum_{i = 0}^{n - 1} d(\gamma(t_i), \gamma(t_{i + 1})) 
= l(\gamma)
\end{align*}
for all partitions $a = t_0 < t_1 < \dots < t_n = b$, $n \in \mathbb{N}$.
To see that $l([\gamma]) = l(\gamma)$ choose the trivial partition $a = t_0 < t_1 = b$ and observe 
\begin{align}
l([\gamma]) 
&\geq d_G([\gamma(t_0)], [\gamma(t_{ 1})]) =
\inf_{g \in G} d(\gamma(a), g \circ \gamma(b)) = d(y_1, \tilde{g} \circ y_2) = l(\gamma). \label{lengthlonger}
\end{align}
Thus, $[\gamma]$ is a shortest path in $\mathcal{Y}/G$ as $l([\gamma]) = d_G([y_1], [y_2])$ and $l([\tilde\gamma]) \geq d_G([y_1], [y_2])$ as in \eqref{lengthlonger} for all other paths~$ \tilde\gamma$.
\end{proof}

\subsection{Proof of Corollary \ref{cor:multivariate_geodesic}}
This corollary is an immediate consequence of the following lemma.
\begin{lemma} \label{lem:convex_cone}
Let $(\mathcal{Y}, \langle \cdot \rangle)$ be a real inner product space and $G$ act on $\mathcal{Y}$ by isometries with $g \circ 0 = 0$ for all $g \in G$. Let $y_1, y_2 \in \mathcal{Y}$ be aligned to $y_0 \in \mathcal{Y}$ and $\alpha_1, \alpha_2 \geq 0$. Then $\alpha_1 y_1 + \alpha_2 y_2$ is aligned to $y_0$, which means the set of elements which are aligned to $y_0$ is a convex cone.
\end{lemma}

\begin{proof}
Let $y \in \mathcal{Y}$ be aligned to $y_0 \in \mathcal{Y}$, that is $\| y_0 - y\| = \inf_{g \in G} \| y_0 - g \circ y \|$. This is equivalent with $\| y\|^2 = \langle y, y \rangle$ to
\begin{align*}
\| y_0 - y\|^2 &= \inf_{g \in G} \langle y_0 - g \circ y, y_0 - g \circ y \rangle \\
&= \| y_0\|^2 + \inf_{g \in G} \{ \langle g \circ y, g \circ y \rangle - 2\langle y_0, g \circ y \rangle \} \\
&= \| y_0\|^2 +\| y\|^2 - 2 \sup_{g \in G} \langle y_0, g \circ y \rangle \tag{$g$ isometry} \\
&= \| y_0 - y\|^2 + 2\langle y_0, y \rangle - 2 \sup_{g \in G} \langle y_0, g \circ y \rangle.
\end{align*}
Hence, $y \in \mathcal{Y}$ being aligned to $y_0 \in \mathcal{Y}$ is equivalent to $ \sup_{g \in G} \langle y_0, g \circ y \rangle = \langle y_0, y \rangle$.

Let $g \in G$, thus $g$ is a bijective isometry on a real vector space with $g \circ 0 = 0$, which means $g$ is linear by the Mazur-Ulam theorem \citep{mazur_ulam}. Hence, for 
$y_1, y_2 \in \mathcal{Y}$ being aligned to $y_0 \in \mathcal{Y}$ and $\alpha_1, \alpha_2 > 0$ it holds that
\begin{align*}
\sup_{g \in G} \langle y_0, g \circ (\alpha_1y_1 + \alpha_2y_2) \rangle 
&= \sup_{g \in G} \left[ \alpha_1 \langle y_0,  g \circ y_1 \rangle
+ \alpha_2 \langle y_0, g \circ y_2 \rangle 
\right]\\
&\leq \alpha_1 \sup_{g \in G} \langle y_0,  g \circ y_1 \rangle 
+ \alpha_2 \sup_{g \in G} \langle y_0,  g \circ y_2 \rangle
\tag{$\alpha_1, \alpha_2 \geq 0$} \\
&=  \alpha_1 \langle y_0, y_1 \rangle + \alpha_2 \langle y_0, y_2 \rangle \tag{$y_1, y_2$ aligned to $y_0 $} \\
&= \langle y_0, \alpha_1y_1 + \alpha_2y_2 \rangle 
\end{align*}
On the other hand, we observe
$\sup_{g \in G} \langle y_0, g \circ (\alpha_1y_1 + \alpha_2y_2) \rangle \geq \langle y_0, \alpha_1y_1 + \alpha_2y_2 \rangle$
if we take $g =$ id and therefore $\alpha_1y_1 + \alpha_2y_2$ is aligned to $y_0$.
\end{proof}

\begin{proof}[Proof of the Corollary]
$\beta_0 + \sum_{j = 1}^k \lambda_j \beta_j =  \sum_{j = 1}^k  \lambda_j (y_0 + \beta_j)$ is aligned to $\beta_0$ according to Lemma \ref{lem:convex_cone} and, thus, $\tilde{f}$ is a shortest path in $\mathcal{Y}/G$ as $x \mapsto  \beta_0 + x \sum_{j = 1}^k \lambda_j \beta_j$ is linear and therefore a shortest path in the inner product space $\mathcal{Y}$ (Lemma \ref{lem:shortest_path_quotient}).
\end{proof}

\subsection{Geodesics \texorpdfstring{in $\mathbb{L}_2/\Gamma$}{} cannot be modeled with splines}
\label{subsec:geodesic_splines}
We construct a counterexample that shows that geodesics between two spline curves in $\mathbb{L}_2/\Gamma$ do not necessarily lie in a spline space. Even though we only show this for a specific example, we expect this to be true for most spline curves and that the geodesic will only actually lie in a spline space in exceptional cases.

Consider the linear SRV spline curves $q_1(t) = \left(\begin{smallmatrix} 1\\2t + 1\end{smallmatrix}\right)$ and $q_2(t) = \left(\begin{smallmatrix} 0\\1\end{smallmatrix}\right)$, $t \in [0,1]$. Since $q_2$ is constant, the optimal warping $\gamma$ of $q_2$ to $q_1$ is given via
\begin{align}
    \dot{\gamma}(t) = \frac{\langle q_1(t), \left(\begin{smallmatrix} 0\\1\end{smallmatrix}\right) \rangle_+}{\int_0^1 \langle q_1(t), \left(\begin{smallmatrix} 0\\1\end{smallmatrix}\right) \rangle_+ dt} 
    = \frac{2t + 1}{\int_0^1 2t + 1 dt} 
    = t + 0.5,
\end{align}
using the formula derived in the online supplement B.1 of \cite{steyer}. Here $<\cdot, \cdot>_+$ denotes the positive part of the scalar product. Hence $q_2(\gamma(t)) \sqrt{\dot{\gamma}(t)} = \left(\begin{smallmatrix} 0\\1\end{smallmatrix}\right) \sqrt{t + 0.5}$ is optimally aligned to $q_1$. This means the geodesic between $[q_1]$ and $[q_2]$ in $\mathbb{L}_2/\Gamma$ is given by 
\begin{align*}
\xi: [0, 1] &\to \mathbb{L}_2/\Gamma \\ x &\mapsto \left[ (1 - x)\left(\begin{smallmatrix} 1\\2t + 1\end{smallmatrix}\right) + x\left(\begin{smallmatrix} 0\\1\end{smallmatrix}\right) \sqrt{t + 0.5} \right]. 
\end{align*}
Thus, $\xi$ lies in a spline space only if $\xi(0.5)$ contains a spline, i.e there is a warping function $\tilde{\gamma}:[0,1] \to [0,1]$ such that 
\begin{itemize}
    \item[I.] $0.5\sqrt{\dot{\tilde{\gamma}}}$ and
    \item[II.] $(\tilde{\gamma} + 0.5)\sqrt{\dot{\tilde{\gamma}}} + 0.5\sqrt{(\tilde{\gamma} + 0.5) \dot{\tilde{\gamma}}}$
\end{itemize}
are splines of some degree $m$. From I. we conclude that $\tilde{\gamma}$ is a spline of degree $2m + 1$, $m \in \mathbb{N}_0$. But this means $(\tilde{\gamma}(t) + 0.5) \dot{\tilde{\gamma}}(t)$ is a piecewise polynomial with degree $4m + 1$, hence its square root $\sqrt{(\tilde{\gamma} + 0.5) \dot{\tilde{\gamma}}(t)}$ cannot be piecewise polynomial. This contradicts the assumption that II. is a spline.

\subsection{Additional algorithms}
\begin{algorithm}[H]
\caption{Quotient space regression for elastic closed curves \label{algo:quotient_regression_closed}}
\KwIn
{data pairs $(\boldsymbol{x}_i, \boldsymbol{\check{q}}_i)$, $i = 1, \dots, n$ where $\boldsymbol{\check{q}}_i$ are the SRV transformations of observed polygons and  $\boldsymbol{x}_i = (x_{i,1}, \dots, x_{i,k})$, $i = 1, \dots, n$ are observed covariates;
convergence tolerance $\epsilon > 0$
}
Compute initial estimate $\hat{\boldsymbol{\beta}}_{0,new}, \dots, \hat{\boldsymbol{\beta}}_{k,new} =\arginf_{\boldsymbol{\beta}_0, \dots, \boldsymbol{\beta}_k \in \mathbb{L}_2
} \sum_{i = 1}^n \| \boldsymbol{\beta}_0 + \sum_{j = 1}^k \boldsymbol{\beta}_j x_{i,j} - \boldsymbol{\check{q}}_i \|_{\mathbb{L}_2}^2$\;
Set $\hat{\boldsymbol{\beta}}_{j,old} = \text{Inf} \quad \forall j = 0, \dots, k$;\\
\While{$\max_{j = 0, \dots, k} \| \hat{\boldsymbol{\beta}}_{j,old} - \hat{\boldsymbol{\beta}}_{j,new} \| > \epsilon $}{
$\hat{\boldsymbol{\beta}}_{j,old} = \hat{\boldsymbol{\beta}}_{j,new}  \quad \forall j = 0, \dots, k$\;
\For{$i \in 1, \dots, n $}{
$\boldsymbol{p}_i = \hat{\boldsymbol{\beta}}_{0, old} + \sum_{j = 1}^k \hat{\boldsymbol{\beta}}_{j, old} x_{i,j}$; \tcp*[f]{compute predicted SRV curves} \\
$\boldsymbol{v}_i(t) = \boldsymbol{p}_i \|\boldsymbol{p}_i\| - \int_0^1 \boldsymbol{p}_i(s) \|\boldsymbol{p}_i(s)\| ds$; \tcp*[f]{compute derivative of closed predicted curves} \\
$\gamma_i = \arginf_{\gamma} \left\| \frac{\boldsymbol{v}_i}{\sqrt{\|\boldsymbol{v}_i\|}} - (\boldsymbol{q}_i \circ \gamma)\sqrt{\gamma} \right\|_{\mathbb{L}_2}^2$  \tcp*{warping step}
}
$\hat{\boldsymbol{\beta}}_{0,new}, \dots, \hat{\boldsymbol{\beta}}_{k,new} = \arginf_{\boldsymbol{\beta}_0, \dots, \boldsymbol{\beta}_k \in \mathbb{L}_2
} \sum_{i = 1}^n  \left\| \boldsymbol{\beta}_0 + \sum_{j = 1}^k \boldsymbol{\beta}_j x_{i,j} - (\boldsymbol{\check{q}}_i \circ \gamma_i)\sqrt{\gamma_i} \right\|_{\mathbb{L}_2}^2$ \\ \tcp*[f]{$\mathbb{L}_2$ spline fit via least-squares} 
}
\Return{$\hat{\boldsymbol{\beta}}_{j} = \hat{\boldsymbol{\beta}}_{j,new} \quad \forall j = 0, \dots, k$}
\end{algorithm}

\begin{algorithm}[ht]
\caption{Fr\'echet regression for elastic curves \label{algo:frechet_reg}}
\KwIn
{data pairs $(\boldsymbol{x}_i, \boldsymbol{q}_i)$, $i = 1, \dots, n$ where $\boldsymbol{q}_i$ are the SRV transformations of observed curves and  $\boldsymbol{x}_i = (x_{i,1}, \dots, x_{i,k})$, $i = 1, \dots, n$ are observed covariates with mean $\bar{\boldsymbol{x}} = \frac{1}{n} \sum_{i = 1}^n \boldsymbol{x}_i$ and empirical covariance matrix $\boldsymbol{\hat{\Sigma}} = \frac{1}{n} \sum_{i = 1}^n (\boldsymbol{x}_i - \bar{\boldsymbol{x}})(\boldsymbol{x}_i - \bar{\boldsymbol{x}})^T$;
new covariate values $\boldsymbol{\tilde{x}}_1, \dots, \boldsymbol{\tilde{x}}_N$;
convergence tolerance $\epsilon > 0$
}
Compute initial mean $\bar{\boldsymbol{p}}_{l, new} = \arginf_{\bar{\boldsymbol{p}}} \sum_{i = 1}^n \left\| \bar{\boldsymbol{p}} - \boldsymbol{q}_i \right\|_{\mathbb{L}_2}^2$, $l = 1, \dots, N$\;
\For{$l = 1, \dots, N$}{
 $s_{l,i} = s(\boldsymbol{x}_i, \boldsymbol{x}_l) = 1 + (\boldsymbol{x}_i - \bar{\boldsymbol{x}})^T \boldsymbol{\hat{\Sigma}}^{-1} (\boldsymbol{x}_l - \bar{\boldsymbol{x}}) \quad \forall i = 1, \dots, n$  \tcp*{compute weights}
 
Set $\bar{\boldsymbol{p}}_{l, old} = \text{Inf}$;\\
\While{$\| \bar{\boldsymbol{p}}_{l, old} - \bar{\boldsymbol{p}}_{l, new} \| > \epsilon$}{
$\bar{\boldsymbol{p}}_{l, old} = \bar{\boldsymbol{p}}_{l, new}$\;
$\gamma_i = \arginf_{\gamma} \left\| \bar{\boldsymbol{p}}_{l, old} - (\boldsymbol{q}_i \circ {\gamma}) \sqrt{\dot{\gamma}} \right\|_{\mathbb{L}_2}^2, \quad \forall i = 1, \dots, n$  \tcp*{warping step}
$\bar{\boldsymbol{p}}_{l, new} = \arginf_{\bar{\boldsymbol{p}}} \sum_{i = 1}^n s_{l,i} \left\| \bar{\boldsymbol{p}} - (\boldsymbol{q}_i \circ {\gamma_i}) \sqrt{\dot{\gamma_i}} \right\|_{\mathbb{L}_2}^2$  \tcp*[f]{weighted $\mathbb{L}_2$ spline fit}
} 
}
\Return{$\bar{\boldsymbol{p}_l} = \bar{\boldsymbol{p}}_{l,new}$ for all $l = 1, \dots, N$}
\end{algorithm}

For the the weighted $\mathbb{L}_2$ spline fit step note that the average of the weights fulfills 
\begin{align*}
  \frac{1}{n} \sum_{i = 1}^n s_{l, i} = 1 + \frac{1}{n} \sum_{i = 1}^n (\boldsymbol{x}_i - \bar{\boldsymbol{x}})^T \boldsymbol{\hat{\Sigma}}^{-1} (\boldsymbol{x}_l - \bar{\boldsymbol{x}}) =  1 +  (\frac{1}{n} \sum_{i = 1}^n \boldsymbol{x}_i - \bar{\boldsymbol{x}})^T \boldsymbol{\hat{\Sigma}}^{-1} (\boldsymbol{x}_l - \bar{\boldsymbol{x}}) = 1
\end{align*}
and therefore the weighted $\mathbb{L}_2$ mean in the spline fit step can be written as
\begin{align*}
    \arginf_{\bar{\boldsymbol{p}}} \sum_{i = 1}^n s_{l,i} \left\| \bar{\boldsymbol{p}} - (\boldsymbol{q}_i \circ {\gamma_i}) \sqrt{\dot{\gamma_i}} \right\|_{\mathbb{L}_2}^2 
    &= \arginf_{\bar{\boldsymbol{p}}} \sum_{i = 1}^n (s_{l,i} \left\| \bar{\boldsymbol{p}} \right\|^2 - 2 s_{l,i} \langle \bar{\boldsymbol{p}},  (\boldsymbol{q}_i \circ {\gamma_i}) \sqrt{\dot{\gamma_i}} \rangle_{\mathbb{L}_2}^2)\\
    &= \arginf_{\bar{\boldsymbol{p}}} \sum_{i = 1}^n  (\left\| \bar{\boldsymbol{p}} \right\|^2 - 2\langle \bar{\boldsymbol{p}},  s_{l,i} (\boldsymbol{q}_i \circ {\gamma_i}) \sqrt{\dot{\gamma_i}} \rangle_{\mathbb{L}_2}^2) \\
    &= \arginf_{\bar{\boldsymbol{p}}} \sum_{i = 1}^n  \left\| \bar{\boldsymbol{p}} - s_{l,i} (\boldsymbol{q}_i \circ {\gamma_i}) \sqrt{\dot{\gamma_i}} \right\|_{\mathbb{L}_2}^2.
\end{align*}
This means we can find a solution for the weighted least-squares via fitting a spline to the pseudo data $s_{l,i} (\boldsymbol{q}_i \circ {\gamma_i}) \sqrt{\dot{\gamma_i}}$, $i = 1, \dots, n$. I.e, we can use the $\mathbb{L}_2$ spline fit step described in Web Appendix A of \cite{steyer} to also fit the weighted $\mathbb{L}_2$ mean. In particular this allows us to compute closed weighted means and therefore perform Fréchet regression for closed elastic curves.

\section{Additional plots and tables}

\subsection{Additional simulation results}
We show here the prediction for all covariates $x = -1, 0.8, \dots, 0.8, 1$ for the simulation runs we picked as an example in Section \ref{sec:simulation}. All runs of simulations 1-4 show similar results as displayed in Fig. \ref{fig:simulation_1} but differ in the number of observed points and the added noise to the observations. Likewise, Fig. \ref{fig:simulation_5} shows an example of the second scenario (Simulations 5-8) and Fig. \ref{fig:simulation_11} an example of the third scenario (Simulations 9-12).
\begin{figure}[H]
    \centering
    \includegraphics[scale=0.65]{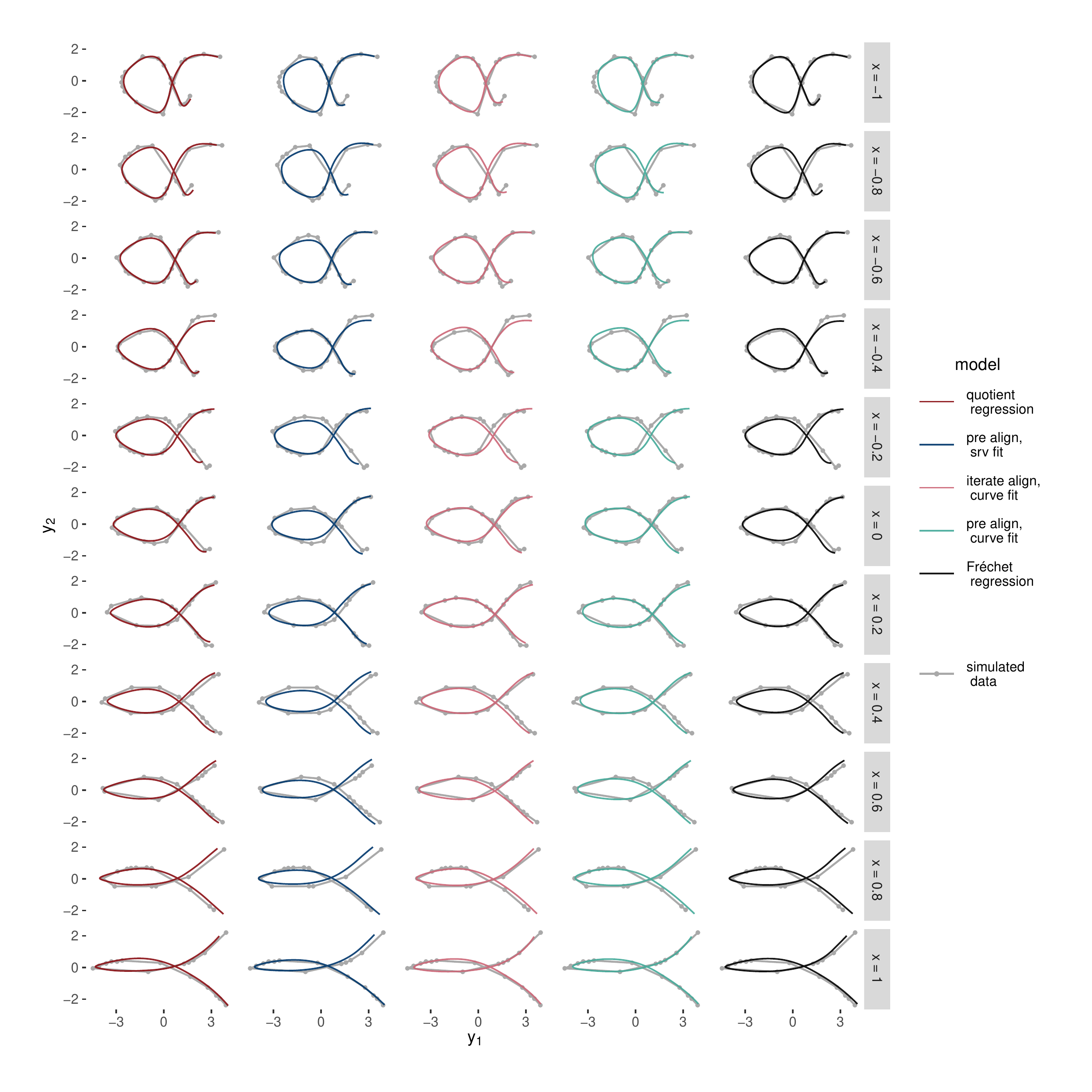}
    \caption{\label{fig:simulation_1} Predictions for an exemplarily selected run of simulation 1.}
\end{figure}

\begin{figure}[H]
    \centering
    \includegraphics[scale=0.65]{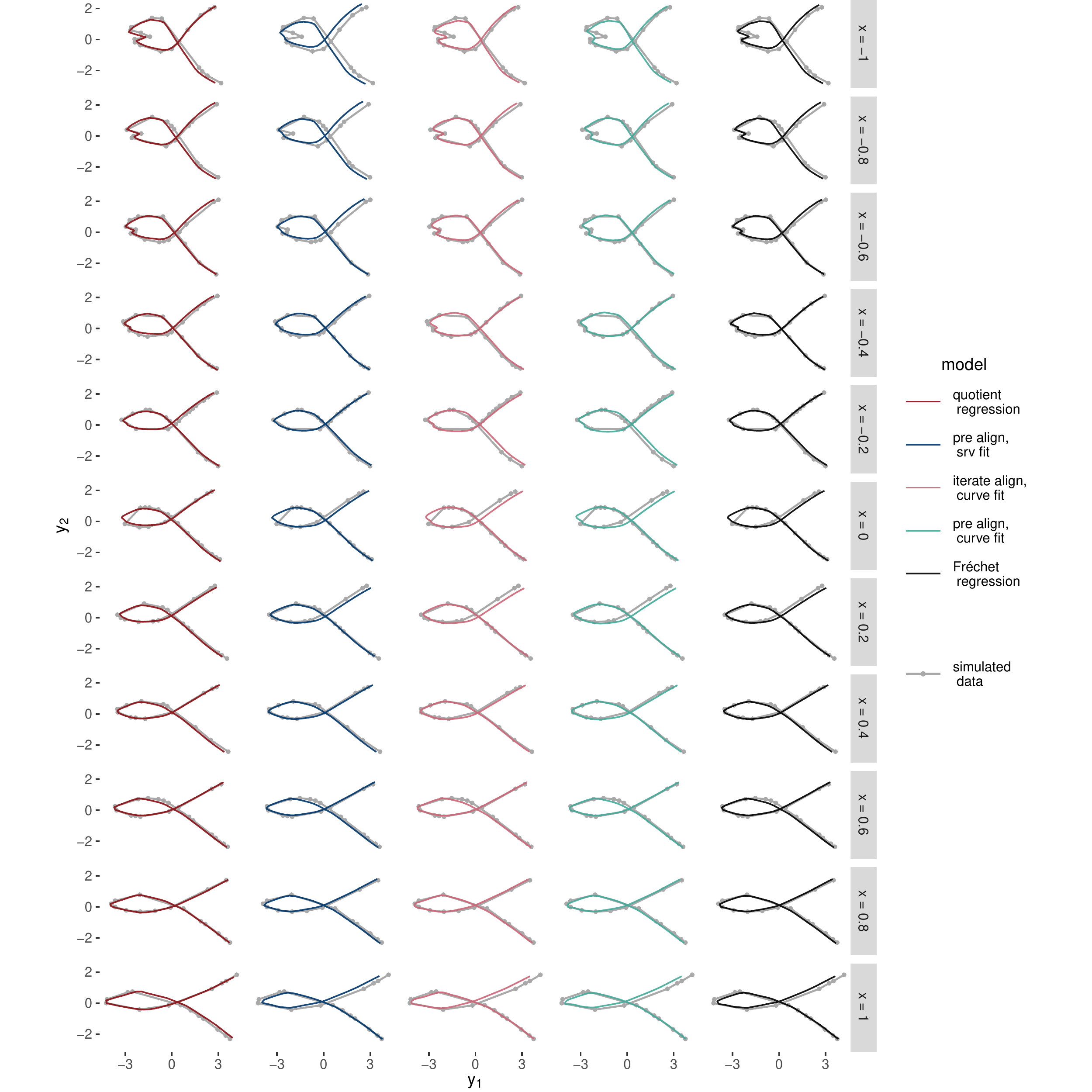}
    \caption{\label{fig:simulation_5}  Predictions for an exemplarily selected run of simulation 5.}
\end{figure}

\begin{figure}[H]
    \centering
    \includegraphics[scale=0.65]{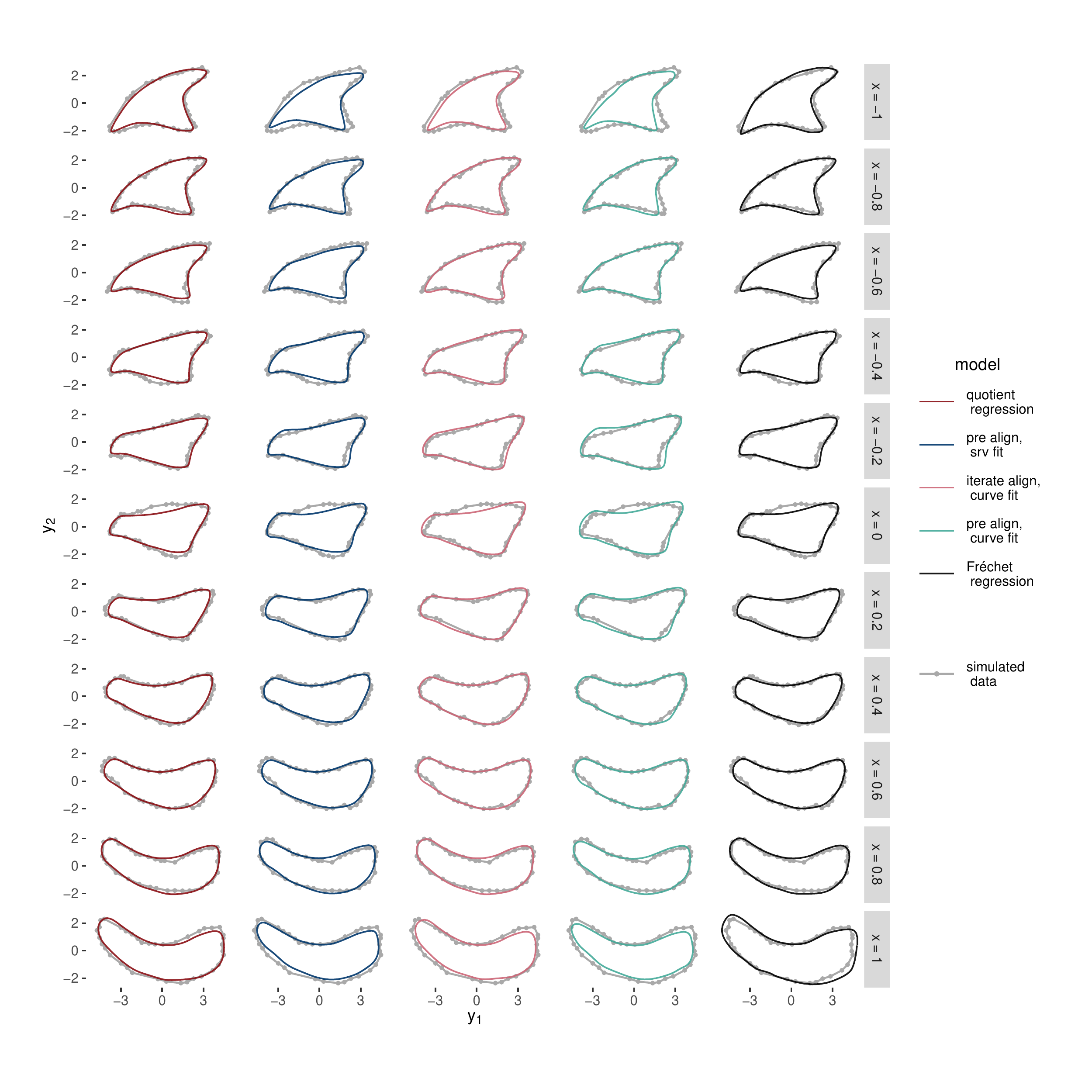}
    \caption{\label{fig:simulation_11}  Predictions for an exemplarily selected run of simulation 11.}
\end{figure}

\begin{figure}[H]
    \centering
    \includegraphics[scale=0.65]{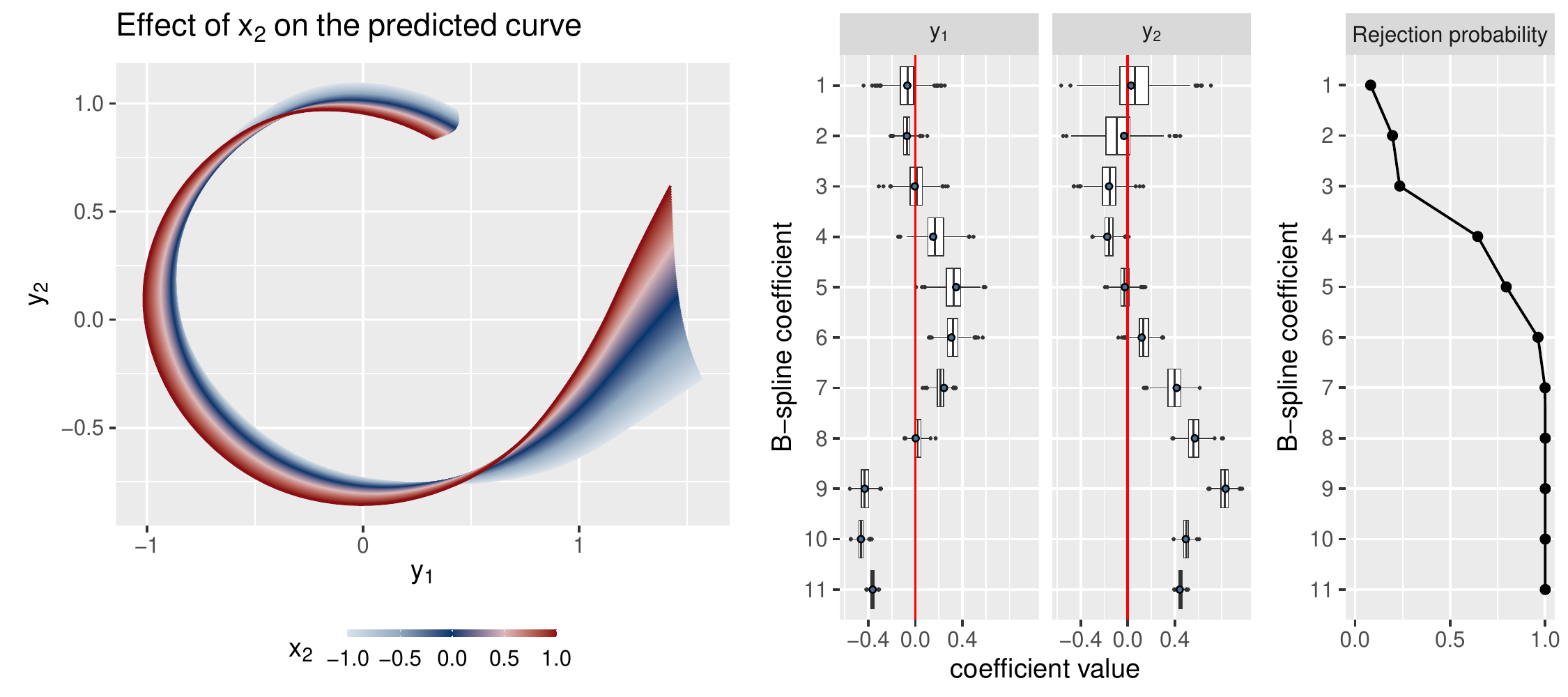}
    \caption{\label{fig:test_beta2_more_coefs} Results for the simulation for $n = 60$ and $N_{boot} = 1000$ using 11 instead of 6 spline coefficients per dimension. Left: Effect of $x_2$ given $x_1 = 0$ estimated on a large simulated sample (n = 500) as an approximation of the closest model to the true model in our model space. Note that the curves are translation invariant, hence the effect appears to be small in the left part of the curve. Middle: Distribution of mean bootstrapped coefficients $\bar{\boldsymbol{\xi}}_{2,m}$, $m = 1, \dots, 12$ over the 1000 repetitions; the coefficients of the effect estimated on the large sample are displayed as blue dots. Coefficients ${\boldsymbol{\xi}}_{j,1}, \dots, {\boldsymbol{\xi}}_{j,11}$  correspond to regions on the curves from the top anti-clockwise to the right. Right: Rejection probabilities for the tests of the individual spline coefficients.}
\end{figure}

\subsection{Additional application results}

\begin{table}[ht]
\centering
\begin{tabular}{rrrrrrr}
  \hline
 full & no Age & no Group & no Sex & only Group & only Sex & only Age \\ 
  \hline
 0.021 & 0.014 & 0.010 & 0.019 & 0.007 & -0.001 & 0.011 \\ 
   \hline
\end{tabular}
\caption{ \label{tab:adj_R_squared} Adjusted coefficients of determination $\tilde{R}^2_{\text{adj}}$ for all possible sub-models.}
\end{table}

\begin{figure}[H]
    \centering
    \includegraphics[scale=0.65]{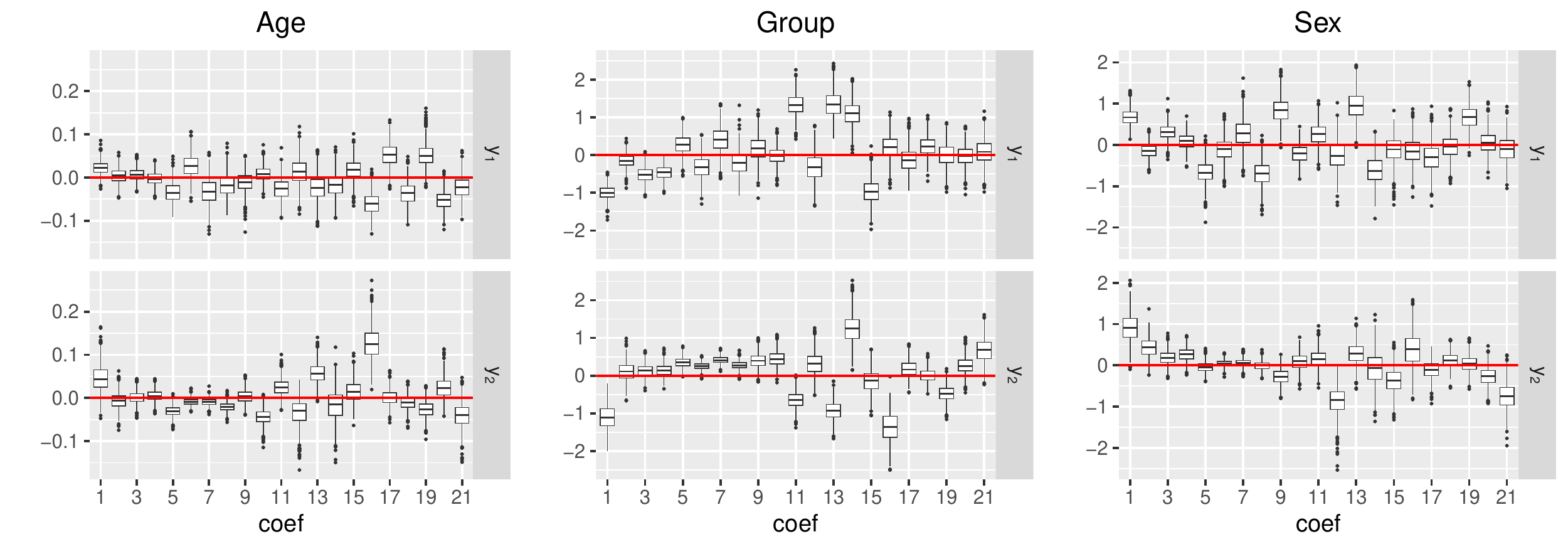}
    \caption{\label{fig:bootstrap_coefs_left} Distribution of the bootstrapped coefficients $\boldsymbol{\xi}_{j,m}^{b}\in \R^2$, $m = 1, \dots, 21$, $j \in \{$Age, Group, Sex$\}$, $b =1, \dots, 1000$ of the left hippocampus.}
\end{figure}

\begin{figure}[H]
    \centering
    \includegraphics[scale=0.65]{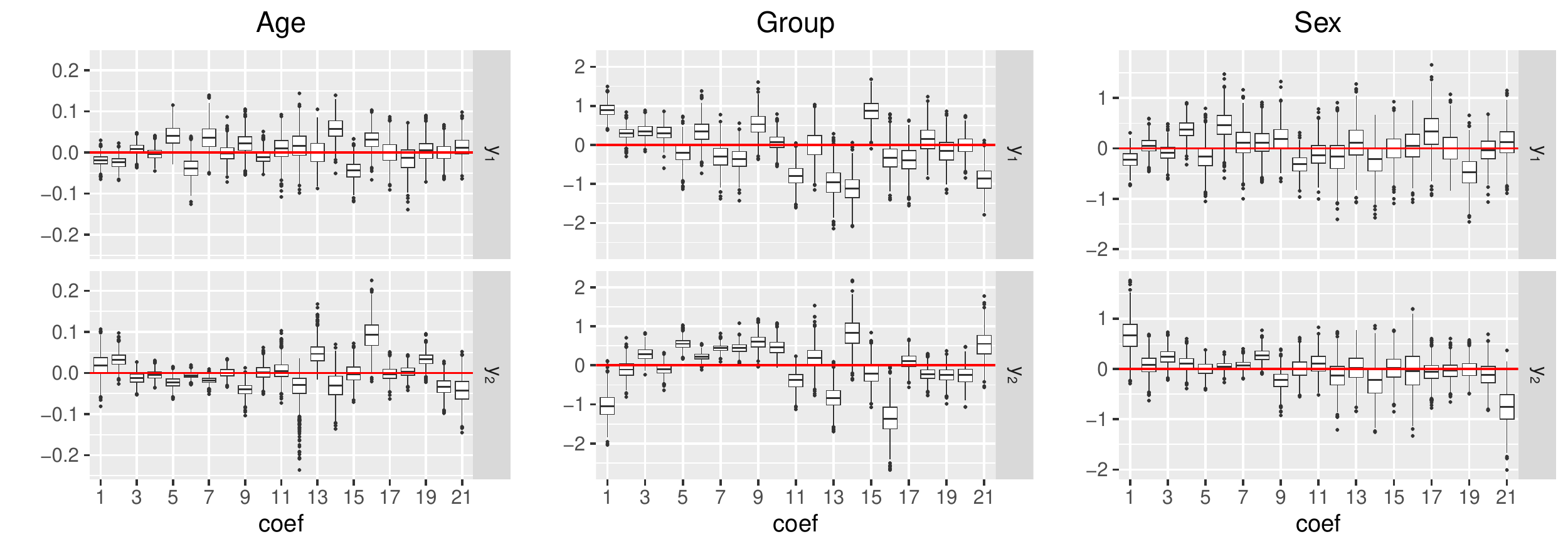}
    \caption{\label{fig:bootstrap_coefs_right} Distribution of the bootstrapped coefficients $\boldsymbol{\xi}_{j,m}^{b}\in \R^2$, $m = 1, \dots, 21$, $j \in \{$Age, Group, Sex$\}$, $b =1, \dots, 1000$ of the right hippocampus.}
\end{figure}

\end{document}